\documentclass[letterpaper, 10pt, journal]{IEEEtran}

\usepackage{float}
\usepackage{times}
\usepackage{multicol}
\usepackage[most]{tcolorbox}

\usepackage{amsmath}
\usepackage{graphicx} 
\usepackage[subrefformat=parens,labelformat=simple]{subfig}
\usepackage{amsfonts,mathrsfs}
\usepackage{xcolor}
\usepackage{accents}
\usepackage{amsthm}

\usepackage[noend]{algpseudocode}
\usepackage{algorithm}
\usepackage[font=small]{caption}

\usepackage[pdftex, colorlinks=true, urlcolor=MyDarkBlue, citecolor=MyDarkRed, linkcolor=MyDarkGreen]{hyperref}
\usepackage{color}
\definecolor{MyDarkRed}{rgb}{0.5,0,0.1}
\definecolor{MyDarkBlue}{rgb}{0.1,0.1,0.5}
\definecolor{MyDarkGreen}{rgb}{0.1,0.5,0.1}
\definecolor{MyRed}{rgb}{1.0,0,0}
\definecolor{MyBlue}{rgb}{0,0,1.0}
\definecolor{MyGreen}{rgb}{0,0.8,0}
\definecolor{purple}{rgb}{0.8,0,0.8}
\definecolor{lightgray}{rgb}{0.96,0.96,0.96}
\definecolor{gray}{rgb}{0.6,0.6,0.6}
\definecolor{darkgray}{rgb}{0.4,0.4,0.4}
\definecolor{MyCoffee}{rgb}{0.435, 0.306, 0.216}
\definecolor{MyCyan}{rgb}{0.0, 0.5, 0.5}

\allowdisplaybreaks

\newtheorem{proposition}{\bf Proposition}
\newtheorem{definition}{\bf Definition}
\newtheorem{lemma}{\bf Lemma}

\newcommand{\algcomment}[1]{{\color{gray}{\small~~// #1}}}

\newcommand{\changed}[1]{{#1}} 
\newcommand{\changedB}[1]{{#1}} 
\newcommand{\changedC}[1]{{#1}} 
\newcommand{\changedF}[1]{{#1}} 

\begin{document}
	
\title{
	\changed{Inter-cluster Transmission Control Using Graph Modal Barriers}
}

\author{
	Leiming Zhang 
	\!,~
	Brian M. Sadler
	\!,~
	Rick S. Blum
	~and~ Subhrajit Bhattacharya 
\thanks{L. Zhang and S. Bhattacharya are with the Department of Mechanical Engineering and Mechanics, Lehigh University, Bethlehem, PA 18015 USA. e-mail: \texttt{[lez316,sub216]@lehigh.edu}.}%
\thanks{B. M. Sadler is with Army Research Laboratory, MD USA. \texttt{Brian.m.sadler6.civ@mail.mil}.}%
\thanks{R. S. Blum is with the Department of Electrical \& Computer Engineering, Lehigh University, Bethlehem, PA 18015 USA. e-mail: \texttt{rb0f@lehigh.edu}}%
\thanks{Manuscript received: ; revised: .}
	\vspace{-0.18in}
}

\maketitle
\thispagestyle{empty}
\pagestyle{empty}

\begin{abstract}
\changed{In this paper we consider the problem of transmission across a \changedF{graph} 
and how to effectively control/restrict it with limited resources. 
\changedF{Transmission} can \changedB{represent} information \changedB{transfer} across a social network, \changedB{spread} of a malicious virus across a computer network, or \changedB{spread} of an infectious disease across communities. The key insight is to assign \changedB{proper} weights to \changedF{bottleneck} edges of the \changedF{graph} based on their role in 
\changedF{reducing the connection} between two or more strongly-connected clusters within the \changedF{graph}. 
\changedF{Selectively reducing the weights (implying reduced transmission rate) on the critical} edges helps limit the transmission from one cluster to another. We refer to these as \emph{barrier weights} and their computation is based on the eigenvectors of the graph Laplacian. 
\changedB{Unlike other work on} graph partitioning and clustering,  
we completely circumvent the associated computational complexities by assigning weights to edges instead of performing discrete graph cuts. \changedF{This} allows us to provide strong theoretical results on our proposed methods. We also develop approximations that allow \changedF{low complexity} distributed computation of \changedF{the barrier} weights using only neighborhood communication on the \changedF{graph}.}
\end{abstract}

\section{Introduction}

\changed{
\IEEEPARstart{W}{e} \changedF{consider transmission} or flow control across a \changedF{graph} with the aim of lowering the rate of transmission from one cluster (a strongly-connected subgraph) to another. 
This is an important problem in many different contexts. In computer networks, for example, this may represent the transmission of a virus between computers, and hence the problem is to limit its spread from \changedB{one} sub-network to another. In \changedB{the} context of an epidemic or pandemic this represents the problem of controlling/slowing the transmission of a commutable disease from one community to another.
The premise of this work is based on identifying the relatively weak inter-cluster edges, and restricting flow/transmission across those 
edges. 
	
We use eigenvectors of a graph Laplacian for identifying edges that form weak connections between well-connected clusters within the graph. Instead of determining the cuts in the graph explicitly, we compute values, called \emph{resistance}, that indicate the potential of an edge of being an inter-cluster edge (as opposed to an intra-cluster edge).
This gives us a means of adjusting flow/transmission rates across edges that prevent or slow down transmission from one cluster to another.
We provide strong theoretical results to that end, and describe a decentralized algorithm that can be used to compute the edge resistances in a distributed manner through neighborhood communication only.
}

\subsection{Related Work}


Graph clustering and partitioning are well-researched topics in graph and network theory.
In \cite{cheng2017network} the authors use \changedF{the} Cheeger constant as a measure of \changedF{a graph bottleneck} and \changedF{solve} the clustering problem by formulating it as a linear program. This method also allows reducing the bottleneck by adding additional relay nodes and links, but the overall time complexity is relatively high.
In \cite{hagen1992new}, eigenvectors \changedF{of the graph Laplacian are} used for partitioning of circuit netlists. This is a similar approach to our work. However, the authors didn't go beyond the second smallest eigenvalue and corresponding eigenvector.
\changedF{A method is proposed in \cite{le2017approximate}} that approximately obtains Fast Fourier Transforms on graphs using \changedF{the} graph Laplacian matrix and eigenvector matrix.
In \changedF{the} context of image segmentation, \cite{shi2000normalized} developed a graph partitioning and segmentation method using more than one \changedF{eigenvector} of 
of the graph Laplacian\changedF{, and the idea was further developed for multi-layer graphs \cite{dong2012clustering}.}
However, when using multiple eigenvectors, \changed{this} approach requires a separate clustering methods, such as k-means. Furthermore, the cluster assignment (\emph{i.e.}, determination of the \emph{cuts}) is a binary/discrete process.
We, on the other hand, assign
numerical values to \changedF{those} edges representing weak connections between strongly-connected clusters, which is more suitable for applications such as inter-cluster transmission control. 
%
\changedF{Online multi-agent optimization methods can be used to form local consensus 
within larger graphs~\cite{koppel2017proximity}, and these methods have some similarity to the approach presented in this paper.}
An ordered transmission approach is proposed in \cite{chen2019testing} to test the covariance matrix of Gaussian graph model where the global transmissions can be reduced without compromising
\changedF{performance, which is particularly effective for graphs that have sparsely connected clusters.}
\changedF{Graph clustering can also be related to discrete and continuous time dynamic processes on graphs such as diffusion and epidemics using a Z-Laplacian framework \cite{yan2017graph}.}
%
\changedF{Identifying critical edges may also be a precursor to placement of protection nodes whose purpose is to stop cascading network failures~\cite{Yu:2018:cascading}.}




\begin{figure*}
	\centering
	{\includegraphics[width = 1.95\columnwidth]{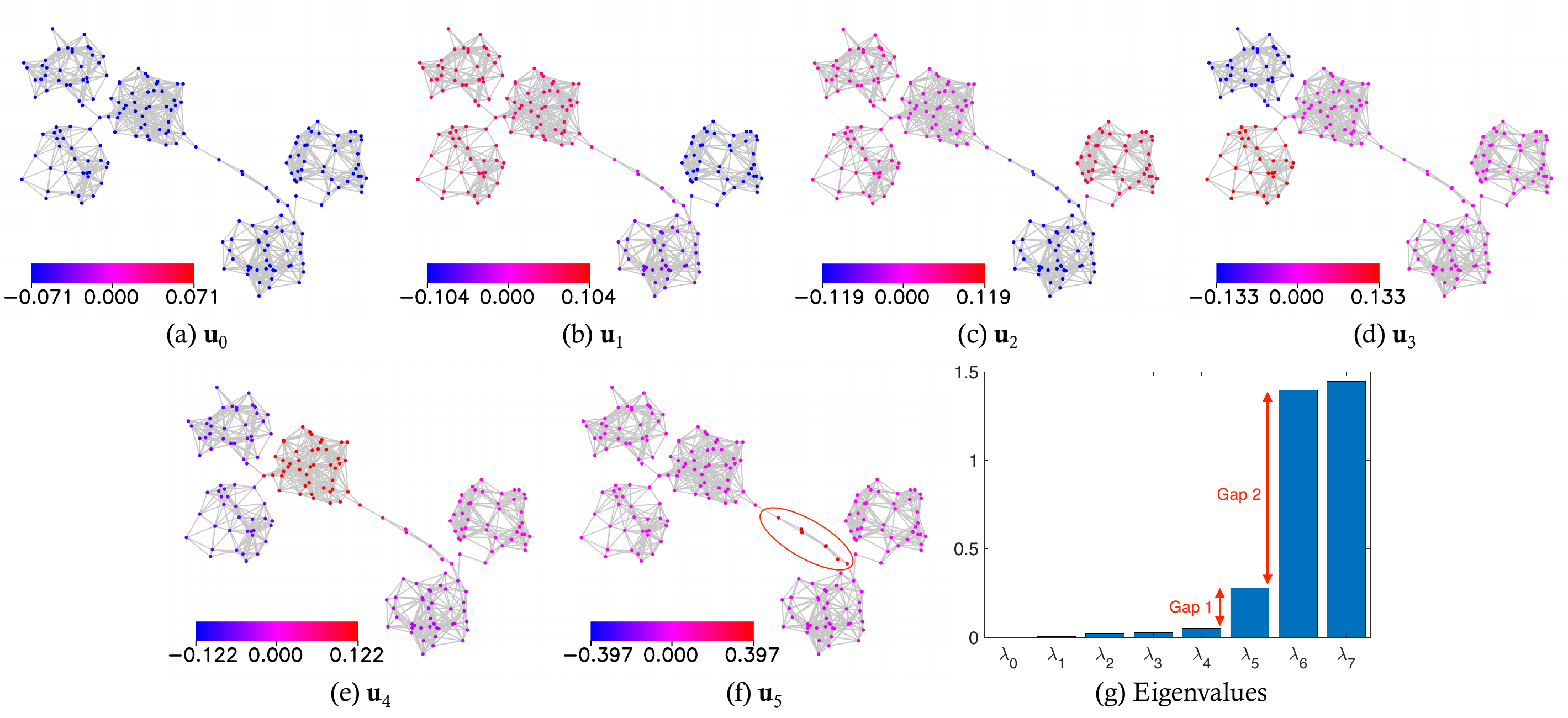}} \hspace{0.0in}
	\caption{Eigenvectors and first $8$ eigenvalues of a graph Laplacian with 5 well-defined clusters.}
	\label{fig:5smallExampleExample}
\end{figure*}

\begin{figure*}
	\centering
	{\includegraphics[width = 1.95\columnwidth]{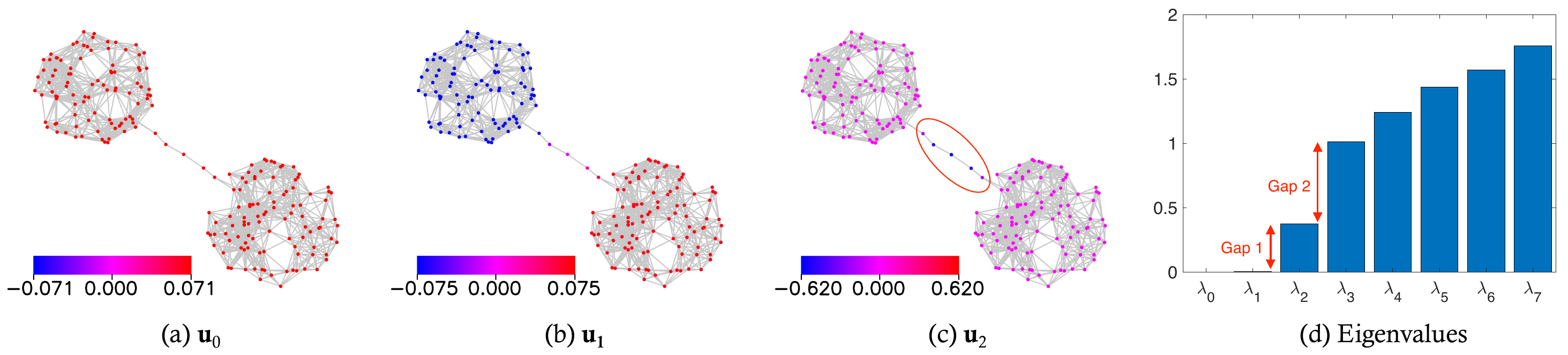}} \hspace{0.0in}
	\caption{Eigenvectors and first $8$ eigenvalues of a graph Laplacian with 2 big clusters.}
	\label{fig:2bigClustersExample}
\end{figure*}

\subsection{Contributions and Overview of Paper}

\changed{
While \changedB{the} Fiedler vector (the eigenvector corresponding to the first non-zero eigenvalue \changedF{of the graph} Laplacian) has been used to partition a graph into two \changedB{parts}~\cite{Chung:1997,hagen1992new}, we make use of the entire spectrum of the Laplacian \changedB{to control flow between} 
multiple clusters in \changedB{a} graph. 
Even in prior work where \changedB{multiple} 
eigenvectors have been used for graph partitioning/clustering (\cite{shi2000normalized,dong2012clustering} \changedF{use the top K} eigenvectors in \changedB{the} context of image segmentation), very little theoretical guarantees \changedF{have been} provided \changedB{for multi-way graph partitioning using the eigenvectors}. This is partly because the set of all possible partitions or cuts of a graph is extremely large and designing algorithms that find the best partitioning have high computational complexity. \changedF{We circumvent} this issue by not seeking to construct explicit partitions using the eigenvectors of the graph Laplacian. Instead, we assign values (called \emph{resistance}) to each edge that indicates the potential of an edge for being a \emph{cut} in the optimal partitioning. These resistances can thus be considered a \changedB{continuous, real-valued proxy} 
for the discrete cuts, and are then sufficient for establishing cost/weight \changedF{\emph{barriers}} that can be used to restrict flow/transmission across the graph, and in particular, restrict the transmission from one cluster to another.
We also develop a decentralized method for computation of \changedF{the} resistance values through distributed, neighborhood connection only. Our \changedF{numerical} evaluation demonstrates how our approach can slow inter-cluster transmission more effectively compared to 
\changedB{baseline methods.}

\changedB{While a dual problem of \changedF{strengthening} the \changedF{bottleneck edges} (\changedF{inter-cluster edges}) may also be considered (\emph{e.g.}, by adding more edges or increasing the weights on the bottleneck edges), in this paper we only consider the problem of restricting inter-cluster flow by means of \changedF{weakening} those edges.
\changedF{However, the methods we propose can be adopted for regulating or increasing flow, as well.}
Furthermore, the problem we consider seeks to restrict connectivity, which is in sharp contrast to that required for consensus on a graph~\cite{10.5555/2821576,10.5555/2440190}. Our objective is to restrict transmission on a graph instead of allowing information flow to happen, which is essential for consensus attainment.}
}

\changed{The paper is organized as follows: In Section~\ref{sec:preliminaries} we provide background on \changedB{the} graph Laplacian, \changedB{including some} known properties of its spectrum.
In Section~\ref{sec:mode-based-rep-of-clusters} we first provide a general intuitive description of how the first $q$ eigenvalues and eigenvectors of the graph Laplacian can be used for identifying clusters in a graph. Following that, in Sections~\ref{sec:theory} we provide mathematical underpinning for the said intuition, formally defining what is meant by strongly-connected \emph{clusters}, how to quantify the weakness of inter-cluster connections, and why the first $q$ eigenvectors of the Laplacian provides a description of such clusters. 
In Section~\ref{sec:barrier} we describe how the eigenvectors can be used to compute the \emph{resistance} values on the edges, and in Section~\ref{sec:approx} we describe two approximations that allow distributed computation of the resistance.
Finally in Section~\ref{sec:information-containment} we describe models of transmission \changedB{over} a network and how the resistances can be used for slowing the inter-cluster transmission rate.
\changedB{Numerical results are presented} in Section~\ref{sec:results}.
\changedB{Conclusions are presented in Section~\ref{sec:conclusion}.}}



\section{Preliminaries} \label{sec:preliminaries}


\begin{figure*}
	\centering
	\subfloat[]{\fbox{\includegraphics[width = .22\linewidth]{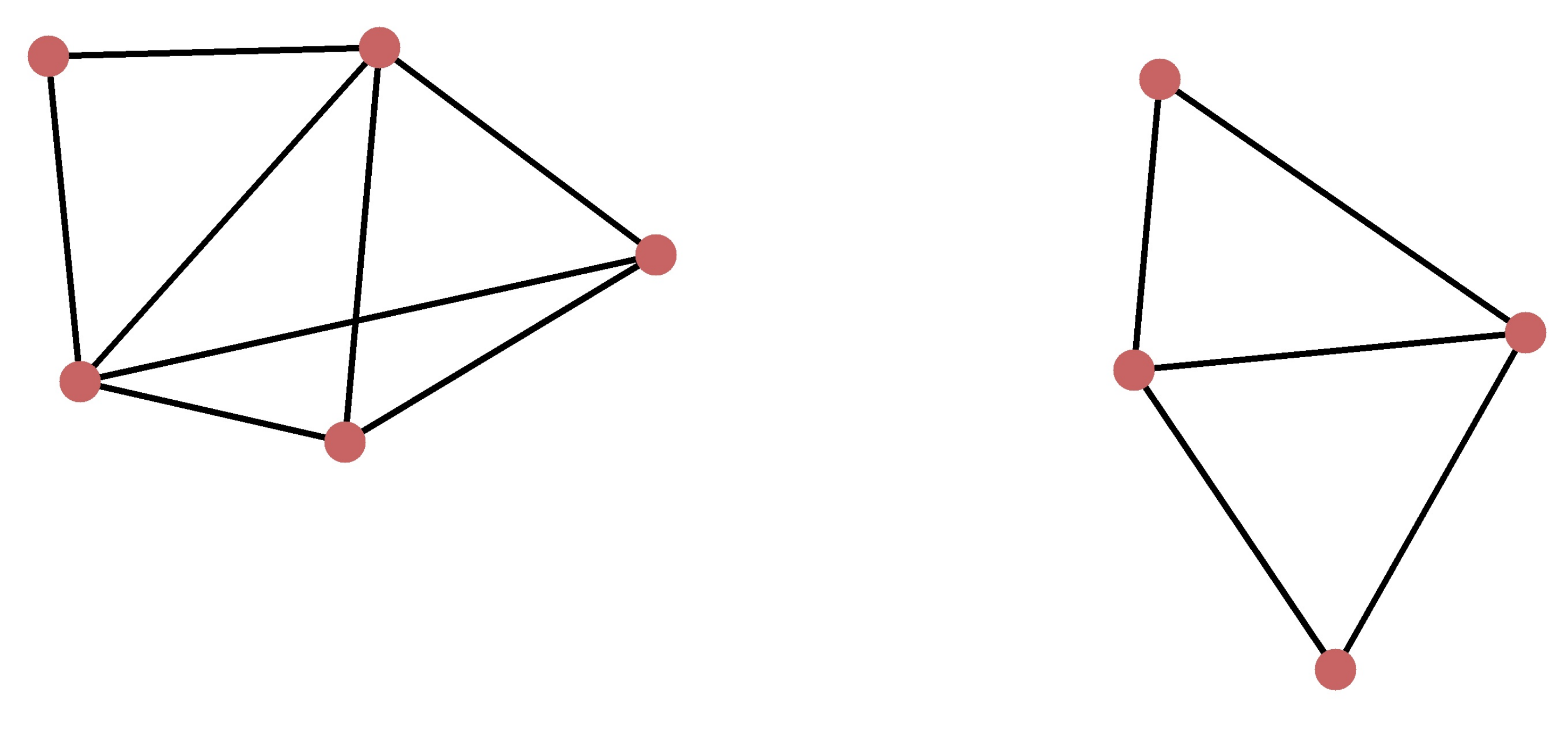}}\label{fig:2smallClustersA}}\hspace{0.1in}
	\subfloat[]{\fbox{\includegraphics[width = .22\linewidth]{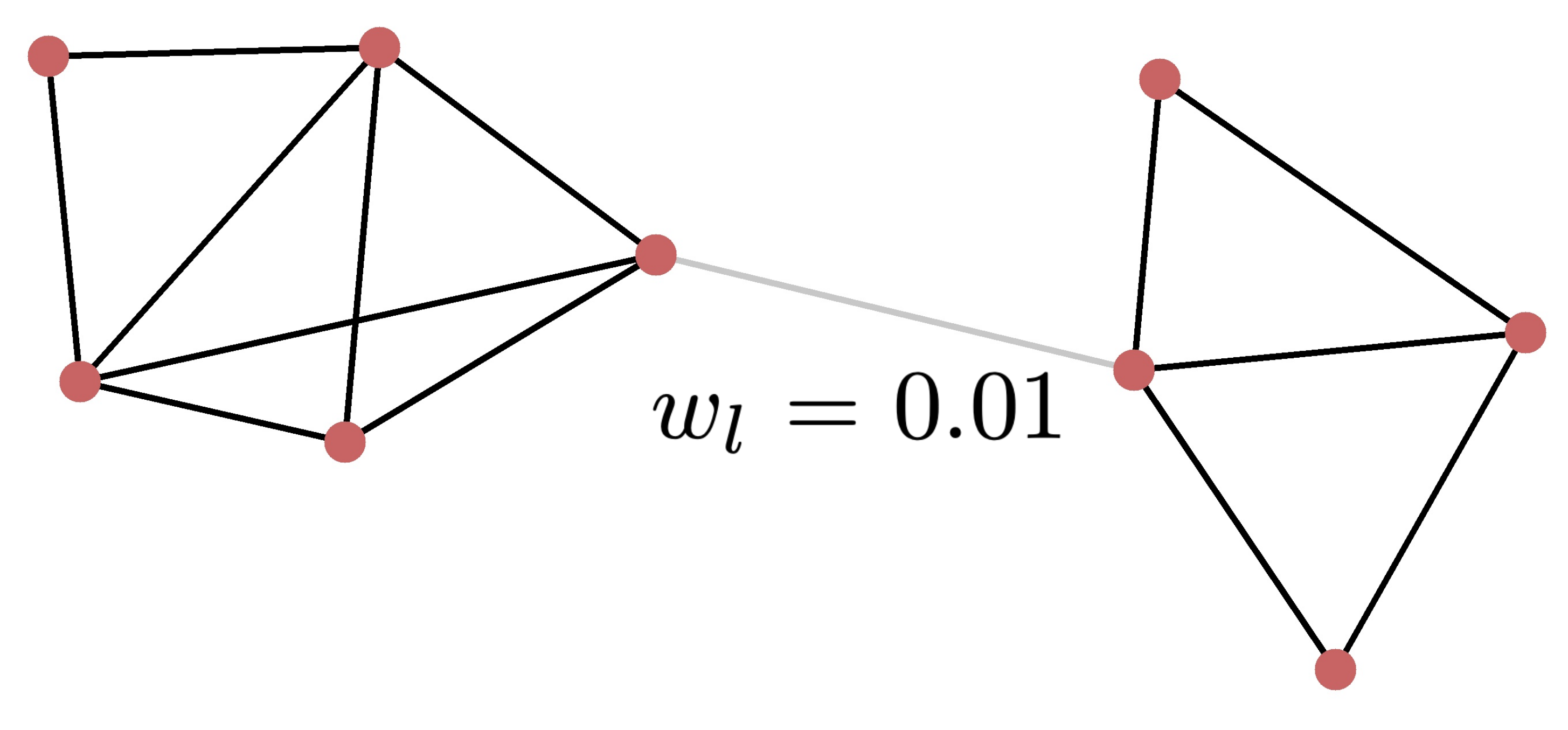}}\label{fig:2smallClustersB}}\hspace{0.1in} 
	\subfloat[]{\fbox{\includegraphics[width = .22\linewidth]{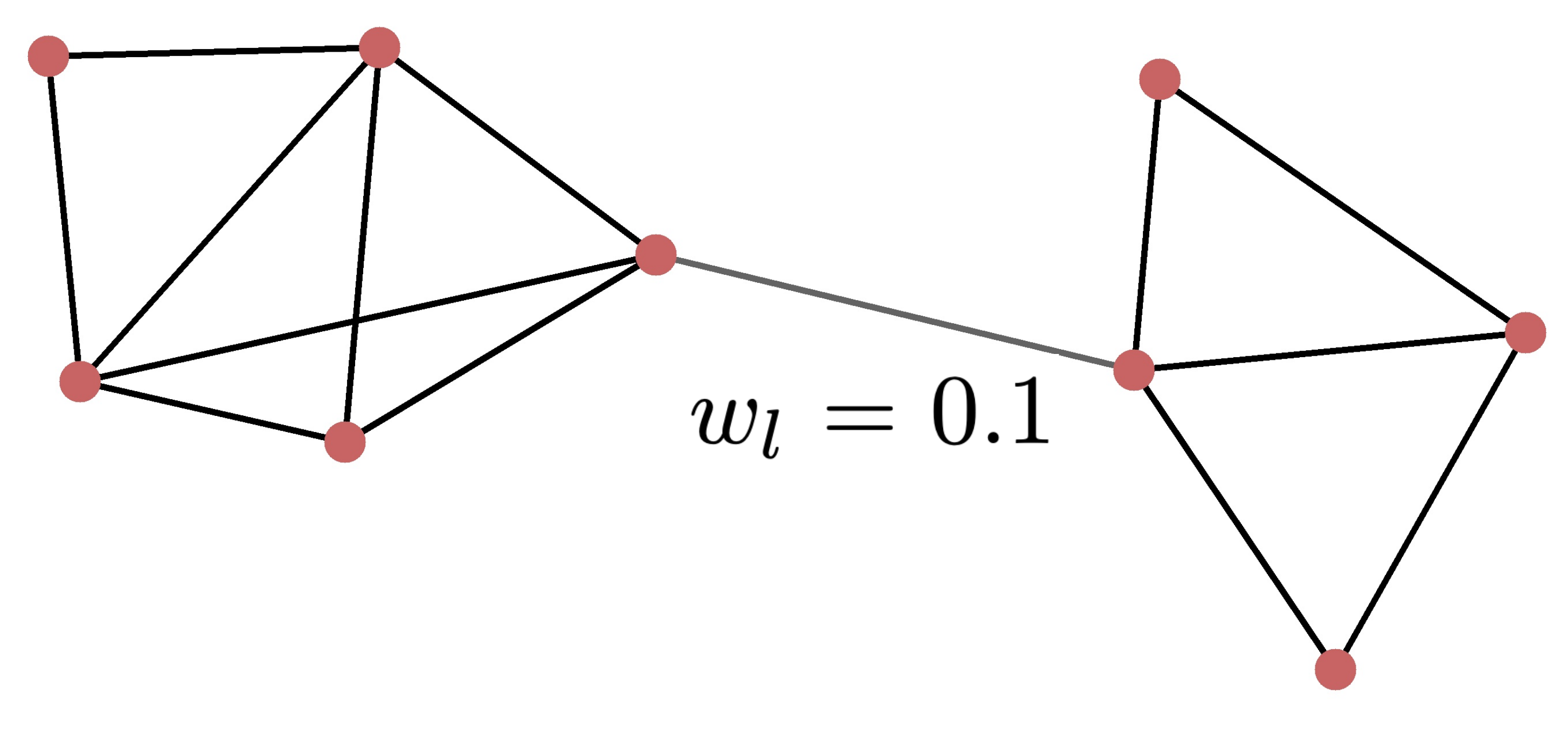}}\label{fig:2smallClustersC}}\hspace{0.1in}
	\subfloat[]{\fbox{\includegraphics[width = .22\linewidth]{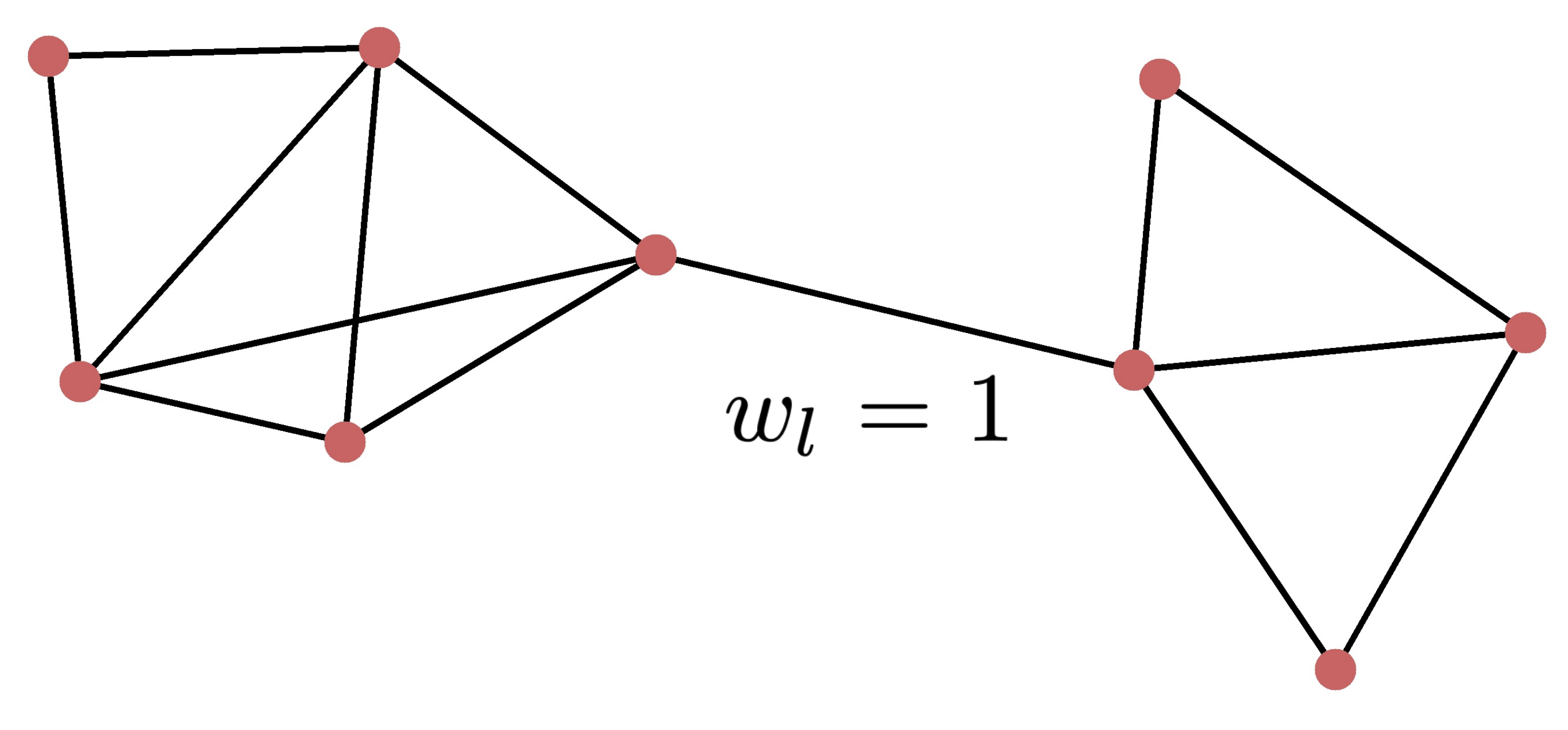}}\label{fig:2smallClustersD}}\\
	\subfloat[] {\includegraphics[width = .23\linewidth]{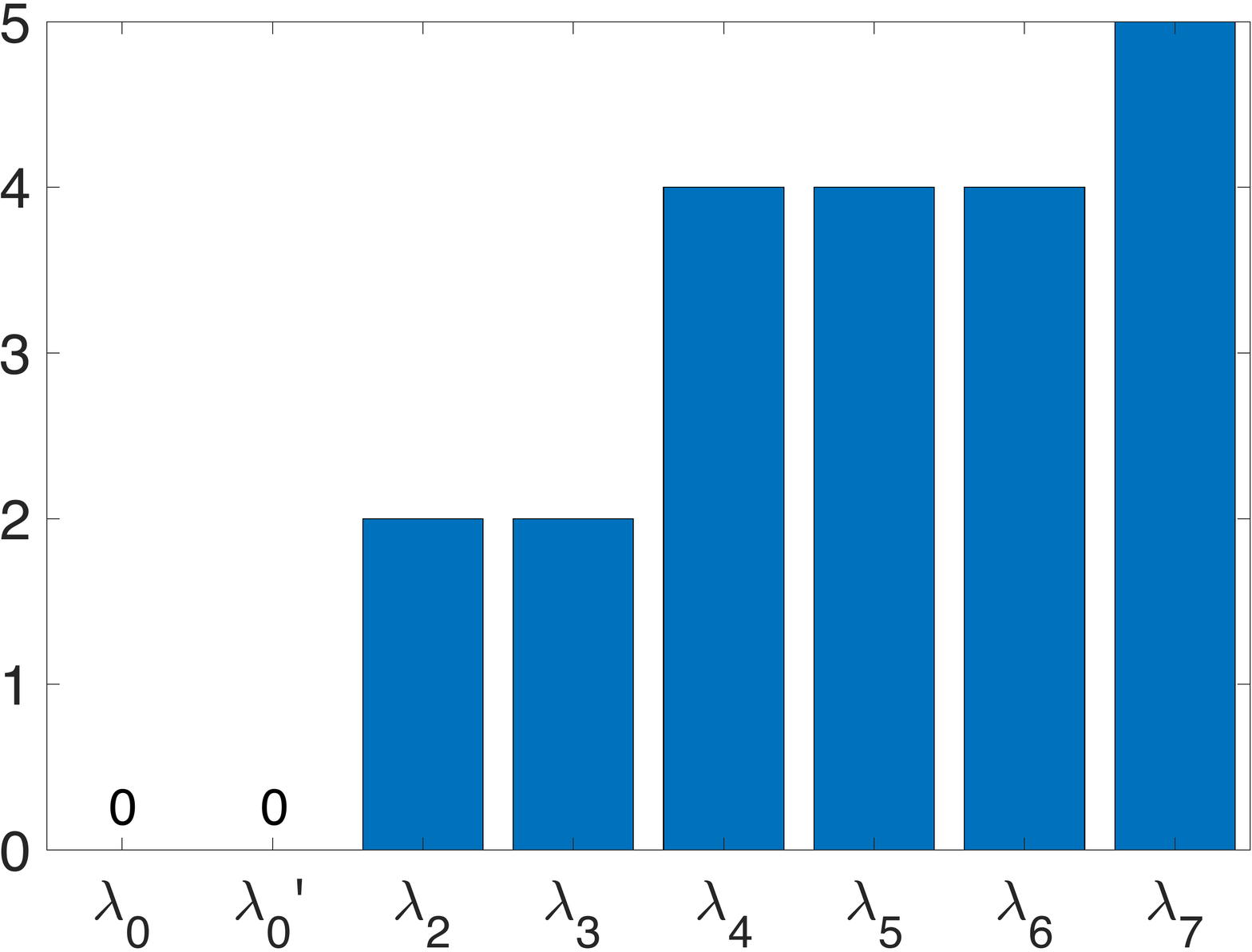}}\label{fig:2smallClustersE}\hspace{0.1in}
	\subfloat[] {\includegraphics[width = .23\linewidth]{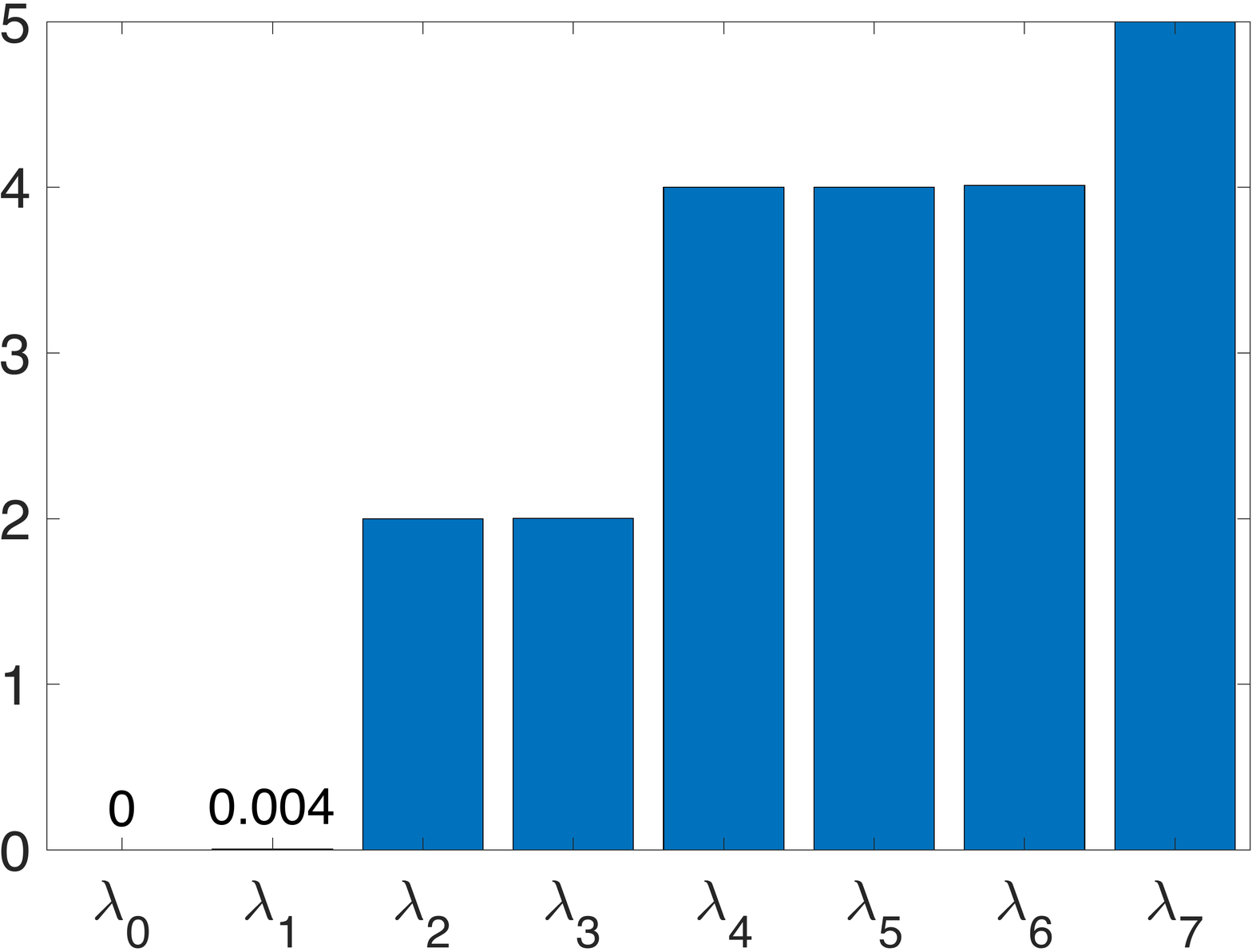}}\label{fig:2smallClustersF}\hspace{0.1in}
	\subfloat[] {\includegraphics[width = .23\linewidth]{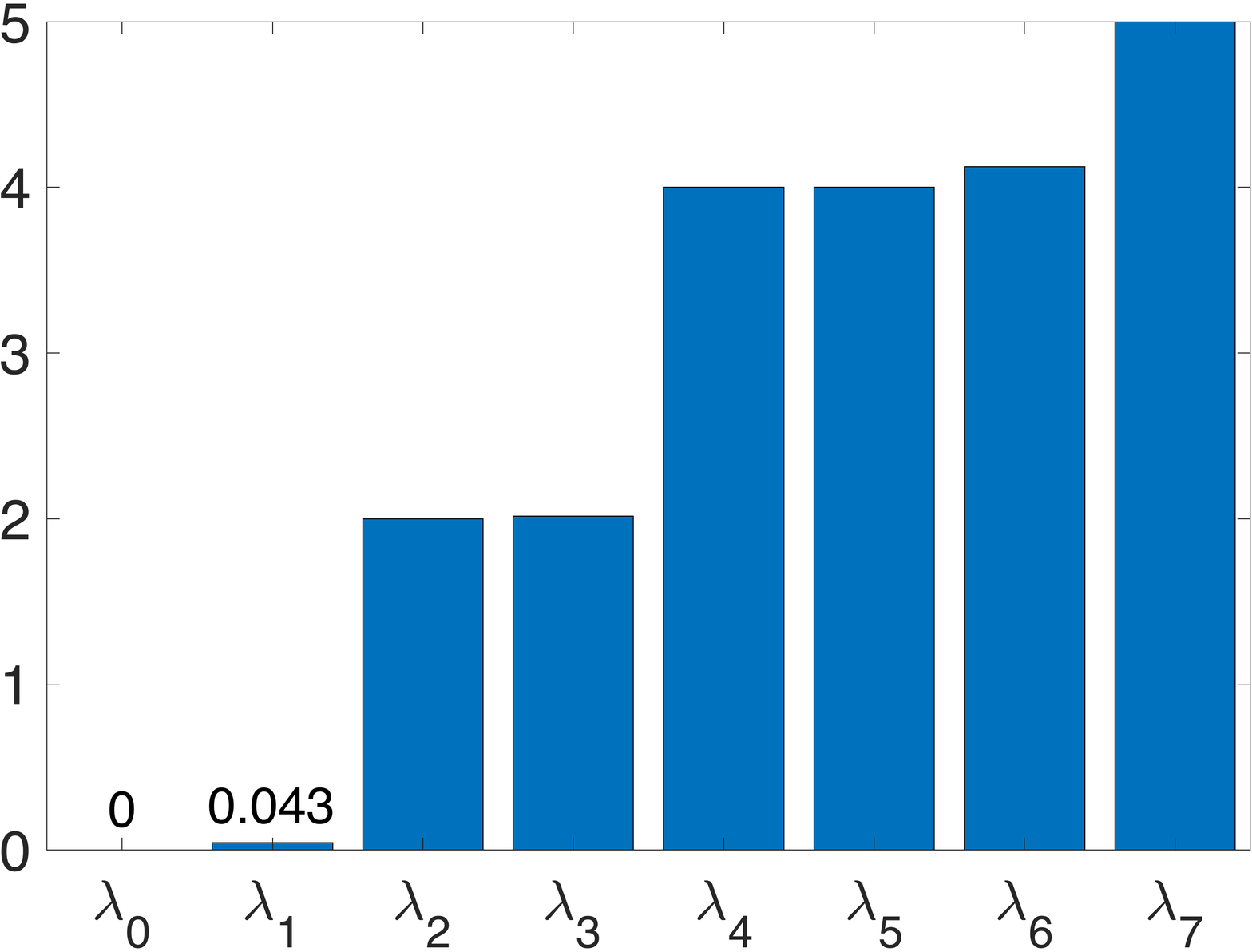}}\label{fig:2smallClustersG}\hspace{0.1in}
	\subfloat[] {\includegraphics[width = .23\linewidth]{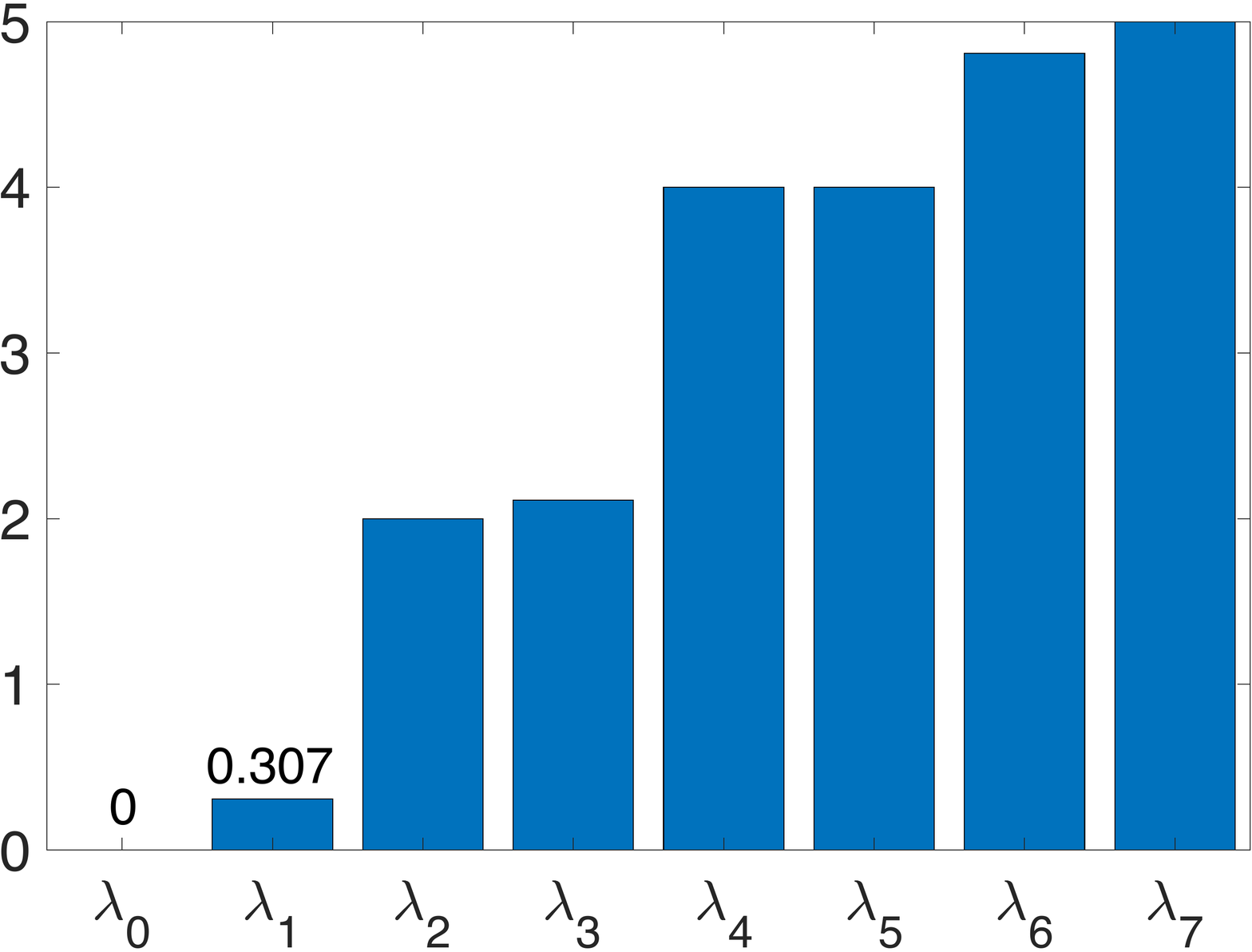}}\label{fig:2smallClustersH}
	\caption{\changed{Illustration of how a transition from a graph with two connected components (a) to a connected graph with a weak connection between two strongly connected components (b-d) \changedB{changes} the eigenvalues.} Eigenvalues of the \changed{respective} graph Laplacians are  shown directly below the graphs. (a) has $2$ separated components without any connection in between (equivalently, zero weight on the inter-component edge). All other weights on the edges are set to $1$. In graphs (b), (c) and (d), we add an inter-cluster edge with increasing weights of $0.01$, $0.1$ and $1$ respectively. As a result the second smallest eigenvalue increases gradually but stays relative small.}
	\label{fig:2smallClusters}
\end{figure*}

For a given undirected graph, $G$, with vertex set \changed{$\mathcal{V}(G) = \{v_1, v_2, \cdots, v_n\}$} and weighted edges, a weighted adjacency matrix, $A$, is defined such that the $(i, j)$-th element of the matrix is the weight of the edge $(v_i, v_j)$ if the edge exists, or $0$ otherwise. The weighted degree matrix, $D$, of the undirected graph is a diagonal matrix with $k$-th diagonal element being the summation of row $k$ (or column $k$) of the adjacency matrix $A$. 
Finally, the Laplacian of the undirected graph is defined as $L = D - A$.

The spectrum of the Laplacian matrix has been \changed{a topic of intense study in} network theory.  
\changedF{The} eigenvalues of the Laplacian are non-negative and the number of connected components of a graph \changed{is equal} to the dimension of the \changed{null space (the multiplicity of the zero eigenvalue)} of the Laplacian matrix.
For a single connected graph, the second smallest eigenvalue \changed{(often referred to as the \textit{Fiedler \changedB{value}} \cite{fiedler1973algebraic})} indicates how connected the graph is.
\changed{Throughout the paper}, we denote $0 = \lambda_0 \le \lambda_1 \le \lambda_2 \le \cdots \le \lambda_{n-1}$ as the eigenvalues of $L$ and $\mathbf{u}_0, \mathbf{u}_1, \cdots, \mathbf{u}_{n-1}$ as corresponding \changed{orthogonal unit eigenvectors (in case of eigenvalues of multiplicities greater than $1$, we select an arbitrary orthogonal basis in the corresponding eigenspace as the eigenvectors)}.


An \changed{eigenvector} of a graph Laplacian, $L$, can be \changed{interpreted} as a real-valued \changedB{distribution} over the vertex set ($\mathbf{u} \in \mathbb{R}^{n}$).
So, the $i$-th element of an eigenvector $\mathbf{u}$ corresponds to the  value \changed{on} vertex $v_i$.
\changed{We refer to these distributions (corresponding to the different eigenvectors of the Laplacian) as \emph{modes} of the graph.}
\changedB{Figures \ref{fig:5smallExampleExample} and \ref{fig:2bigClustersExample} show} \changed{examples} of such modes and corresponding eigenvalues.

\section{\changed{Mode-based Representation of Clusters}} \label{sec:cluster}

\subsection{\changed{Qualitative Description}} \label{sec:mode-based-rep-of-clusters}

\changed{Building on the previous discussion, one can expect that in a single connected graph with weakly connected components, the eigenvectors corresponding \changedF{to the smallest} eigenvalues would correspond to relatively uniform distributions over each strongly connected \changedF{component} (or \changedF{cluster}) in the graph.
This is illustrated in} Figure \ref{fig:2smallClusters}. Starting from 2 isolated clusters \changed{(in which case the zero eigenvalue has multiplicity of two)}, and then adding a weak edge between the two clusters and gradually increasing the weight of the edge, \changed{one of the two zero eigenvalues gradually \changedB{increases} to attain low positive values}.
We consider the eigenvectors corresponding to the smallest eigenvalues to qualify the structure of the graph and identifying strongly-connected components.
\changed{In particular, we identify \emph{jumps} or \emph{gaps} in the ordered set of eigenvalues to determine which eigenvalues are to be considered as \emph{small}.}


In \changed{the example of Figure~\ref{fig:5smallExampleExample}}, there are $2$ big \changedF{clusters (upper left and lower right in the figure)}.
\changed{As a result, a linear combination of the first two modes, $\mathbf{u}_0$ and $\mathbf{u}_1$, will give a distribution that is almost-uniform \changedB{over the} clusters.}
There are also $5$ smaller sub-clusters \changed{in the same example} which can likewise be readily identified from \changed{linear combinations of} $\mathbf{u}_0$ to $\mathbf{u}_4$. 
In Figure \ref{fig:5smallExampleExample}(g), the gap between $\lambda_{4}$ and $\lambda_{5}$ is relatively large \changedF{indicating there are 5 well-defined clusters within the 2 primary clusters} in the graph.
\changedF{Note that a} small cluster circled in Figure \ref{fig:5smallExampleExample}\changedB{(f)} is also captured and its eigenvalue $\lambda_5$ is relatively larger compared to the first $5$ eigenvalues due to its smaller scale.
However, since this is the last well identified cluster, there is an even larger gap between $\lambda_5$ and $\lambda_6$.

In \changed{the example of} Figure \ref{fig:2bigClustersExample}(b), there are $2$ \changed{large} clusters. \changed{This is once again manifested in the fact that there is a relatively large jump/gap (``gap $1$'') in the value of eigenvalues after the first two eigenvalues (\emph{i.e.}, a gap between $\lambda_1$ and $\lambda_2$).}
A second gap \changed{(``gap $2$'' between $\lambda_2$ and $\lambda_3$)} corresponds to a smaller scale cluster circled \changedF{in Figure} \ref{fig:2bigClustersExample}(c).
\changed{An absence of any \changedF{relatively} large gap following the third eigenvalue indicates the absence of any other significant cluster or sub-cluster in the graph.}


\subsection{\changed{Theoretical Foundation}} \label{sec:theory}

The qualitative \changedB{introduction} and general \changedB{observations} about the first $q$ eigenvectors representing the clusters in a graph \changedF{are} formalized by the following \changedF{development} leading to Proposition~\ref{prop:q-conn-dist}.

\subsubsection{Mode-based Representation of \changedF{Clusters}}
We first formalize the notion of \changedF{eigenvectors} representing clusters in Definition~\ref{def:q-conn}, followed by a formal description of strongly-connected clusters with \changedB{weak} inter-cluster connections. The relationship between these two definitions is then established in Proposition~\ref{prop:q-conn-dist}. 
\changedF{We start by describing \changedF{some notation}.}

	
	Suppose a graph $G$ is cut into $q$ disjoint sub-graphs, $G_0,G_1,\cdots, G_{q-1}$, to construct the \emph{$q$-partitioned sub-graph},
	 %
	 $\overline{G}$. 
	 \changedF{That is,} $\mathcal{V}(\overline{G}) = \mathcal{V}(G) = \cup_{j=0}^{q-1} \mathcal{V}(G_j)$ and $\mathcal{E}(\overline{G}) = \cup_{j=0}^{q-1} \mathcal{E}(G_j) \subseteq \mathcal{E}(G)$, or more compactly, $\overline{G} = \cup_{j=0}^{q-1} G_j$,
	 so that $\overline{G}$ has at least $q$ connected components (each of $G_j, j=0,1,\cdots,q-1$).
	 We denote the set of the sub-graphs themselves that constitute the partition by $\widetilde{\mathcal{G}} = \{G_0,G_1,\cdots, G_{q-1}\}$.
	 
	 Let $\overline{N} = \{\overline{\mathbf{u}}_0, \overline{\mathbf{u}}_1, \cdots, \overline{\mathbf{u}}_{q-1}\}$ 
	 be the set of normalized (unit) vectors such that $\overline{\mathbf{u}}_j\in \mathbb{R}^n$ (for $0\leq j \leq q-1$) corresponds to \changedF{a} uniform positive distribution on the vertices of the sub-graph $G_j$ and zero on all other vertices.
	 The vectors in $\overline{N}$ thus form an orthogonal set in the null-space 
	 of the Laplacian, $\overline{L}$, of $\overline{G}$.
	 Let the set of all other orthogonal unit eigenvectors 
	 of $\overline{L}$ be $\overline{N}^\perp = \{\overline{\mathbf{u}}_{q}, \overline{\mathbf{u}}_{q+1}, \cdots, \overline{\mathbf{u}}_{n-1}\}$\changedF{, so that $\text{span}(\overline{N})$ and $\text{span}(\overline{N}^\perp)$ are orthogonal subspaces, and $\overline{N} \cup \overline{N}^\perp$ spans the entire $\mathbb{R}^n$}.
	 %
	 
	 Let the set of orthogonal unit eigenvectors corresponding to the lowest $q$ eigenvalues of the Laplacian, $L$, of $G$ be $N = \{\mathbf{u}_0, \mathbf{u}_1, \cdots, \mathbf{u}_{q-1}\}$ (the first $q$ \emph{modes}).
	 Let the set of all other orthogonal unit eigenvectors 
	 of $L$ be $N^\perp = \{\mathbf{u}_{q}, \mathbf{u}_{q+1}, \cdots, \mathbf{u}_{n-1}\}$.
	 
	 \vspace{0.5em}
	 \begin{definition}\label{def:q-conn}~\changedF{[\emph{$q$-modal Distance from a $q$-partitioned Graph}] }The \emph{$q$-modal distance} of $G$ from the partitioned graph $\overline{G}$ (\emph{i.e.}, the distance of the first $q$ modes of $G$ \changedB{from the uniform distributions over the sub-graphs in $\widetilde{\mathcal{G}} = \{G_0,G_1,\cdots, G_{q-1}\}$}) is then defined as
	 \footnote{\tiny 
	 	The second equality in~\eqref{eq:prop1} holds because of the following: For orthonormal basis $\{\mathbf{u}_0,\mathbf{u}_1,\cdots,\mathbf{u}_{n-1}\}$ and $\{\overline{\mathbf{u}}_0,\overline{\mathbf{u}}_1,\cdots,\overline{\mathbf{u}}_{n-1}\}$, either of which span $\mathbb{R}^n$, we have
	 	\begin{equation*}\begin{array}{rl}
	 	& \displaystyle \sum_{k=0}^{n-1} | \overline{\mathbf{u}}_j^\mathsf{T} \mathbf{u}_k |^2 = \|\overline{\mathbf{u}}_j\|_2^2 = 1 
	 	~~\text{and}~~
	 	\sum_{j=0}^{n-1} | \overline{\mathbf{u}}_j^\mathsf{T} \mathbf{u}_k |^2 = \|\mathbf{u}_k\|_2^2 = 1 \\
	 	\Rightarrow & \displaystyle \sum_{j=0}^{q-1} \sum_{k=0}^{n-1} | \overline{\mathbf{u}}_j^\mathsf{T} \mathbf{u}_k |^2
	 	~=~ q ~=~
	 	\sum_{j=0}^{n-1} \sum_{k=0}^{q-1} | \overline{\mathbf{u}}_j^\mathsf{T} \mathbf{u}_k |^2 \\
	 	\Rightarrow & \displaystyle \sum_{j=0}^{q-1} \sum_{k=0}^{q-1} | \overline{\mathbf{u}}_j^\mathsf{T} \mathbf{u}_k |^2 ~+~ \sum_{j=0}^{q-1} \sum_{k=q}^{n-1} | \overline{\mathbf{u}}_j^\mathsf{T} \mathbf{u}_k |^2 \\
	 	& \qquad\qquad \displaystyle ~=~~ \sum_{j=0}^{q-1} \sum_{k=0}^{q-1} | \overline{\mathbf{u}}_j^\mathsf{T} \mathbf{u}_k |^2 ~+~ \sum_{j=q}^{n-1} \sum_{k=0}^{q-1} | \overline{\mathbf{u}}_j^\mathsf{T} \mathbf{u}_k |^2 \\
	 	\Rightarrow &  \displaystyle \sum_{j=0}^{q-1} \sum_{k=q}^{n-1} | \overline{\mathbf{u}}_j^\mathsf{T} \mathbf{u}_k |^2 ~=~ \sum_{j=q}^{n-1} \sum_{k=0}^{q-1} | \overline{\mathbf{u}}_j^\mathsf{T} \mathbf{u}_k |^2
	 	\end{array}\end{equation*}}
	 
	 {\small \begin{equation}\label{eq:prop1} \begin{array}{l}
	 	\mathrm{modaldist}_{G}(\widetilde{\mathcal{G}}) ~~=~~ \sqrt{ \displaystyle \frac{1}{q(n-q)}  ~ \sum_{j=0}^{q-1} \sum_{k=q}^{n-1} | \overline{\mathbf{u}}_j^\mathsf{T} \mathbf{u}_k |^2 } \\
	 	\qquad\qquad\qquad\qquad =~~ \sqrt{ \displaystyle \frac{1}{q(n-q)}  ~ \sum_{j=q}^{n-1} \sum_{k=0}^{q-1} | \overline{\mathbf{u}}_j^\mathsf{T} \mathbf{u}_k |^2 }
%
	 \end{array}\end{equation}}
	\end{definition}

	 \emph{Key Property of $q$-modal distance:} 
	 It is easy to observe that 
	 the $q$-modal distance is zero
	 if and only if $\text{span}(N) = \text{span}(\overline{N})$ -- that is, the first $q$ eigenvectors, $\mathbf{u}_0, \mathbf{u}_1, \cdots, \mathbf{u}_{q-1}$, of $L$, 
	 are linear combinations of $\overline{\mathbf{u}}_0, \overline{\mathbf{u}}_1, \cdots, \overline{\mathbf{u}}_{q-1}$,
	 thus corresponding to distributions that are uniform over the vertices of each of the sub-graphs $G_0,G_1,\cdots, G_{q-1}$.
	 More generally, $\mathrm{modaldist}_{G}(\widetilde{\mathcal{G}})$ measures how far the first $q$ modes of $G$ are from the distributions which are uniform over the vertices of $G_0,G_1,\cdots, G_{q-1}$.

	\vspace{0.5em}
	\begin{definition}~ [\emph{Relative Outgoing Weight of a Sub-Graph}] \label{def:relout}
	Suppose the weights on the outgoing\footnote{Since $G$ is undirected, \emph{outgoing} and \emph{incoming} are equivalent when referring to edges. We however choose to use the former terminology.} edges from the sub-graph $H$ of $G$ are $w_{1},w_{2},\cdots,w_{m_H}$ (\emph{i.e.}, these are the weights on the edges connecting a vertex in $H$ to a vertex not in $H$). 
	We define the \emph{relative outgoing weight} of $H$ in $G$ as
	\begin{eqnarray*}
		\mathrm{relout}_G(H) & = & \frac{1}{\sqrt{|\mathcal{V}(H)|}} \sum_{h=1}^{m_H} w_h \\
	& = & \frac{1}{\sqrt{|\mathcal{V}(H)|}} \sum_{\{k,l \,|\, v_k \in \mathcal{V}(H), v_l\notin \mathcal{V}(H) \}} \!\!\!\! A_{kl}
	 \end{eqnarray*}
	where $|\mathcal{V}(H)|$ is the number of vertices in $H$ and $A$ is the adjacency matrix of $G$.
	Thus $\mathrm{relout}_G(H)$ is \changedB{smaller} for sub-graphs that are large and are weakly-connected to the rest of the graph.
	\end{definition}

	\vspace{0.5em}
	\begin{definition}~ [\emph{Average Relative Outgoing Weight of a $q$-partitioned Graph}] \label{def:avg-relout}
	Given a $q$-partition of $G$, 
	$\widetilde{\mathcal{G}} = \{G_0,G_1,\cdots, G_{q-1}\}$, its \emph{average relative outgoing weight} is defined as the root mean square of the relative outgoing weights of the sub-graphs:
	{\small \[ \mathrm{avgrelout}_{G}(\widetilde{\mathcal{G}}) ~=~ \sqrt{ \frac{1}{q} \sum_{j=0}^{q-1} \left( \mathrm{relout}_G(G_j) \right)^2 } \]}
	\end{definition}
	
	\vspace{0.5em}
	\begin{proposition} \label{prop:q-conn-dist}
		For any $q$-partition, $\widetilde{\mathcal{G}} = \{G_0,G_1,\cdots, G_{q-1}\}$, of $G$,
	 \[ \begin{array}{l}
	 \mathrm{modaldist}_{G}(\widetilde{\mathcal{G}}) 
		~~\leq~~ 
		\displaystyle 
		\changedC{	\frac{1}{\displaystyle \lambda_q} \sqrt{\displaystyle \frac{2}{n-q}} } 
	 	~\mathrm{avgrelout}_{G}(\widetilde{\mathcal{G}}) 
	 \end{array} \]
	 where 
	 $\lambda_q$ is the $(q+1)$-th eigenvalue (in order of magnitude) of the Laplacian, $L$, of $G$.
	\end{proposition}
	\begin{proof}
		\changedF{See} Appendix~\ref{appendix:prop-q-conn-dist-proof}.
	\end{proof}
	\vspace{0.5em}
	
	\noindent \emph{Remarks on Proposition~\ref{prop:q-conn-dist}:}
	\begin{enumerate}
	\item[i.] As a consequence of the above proposition, out of all possible $q$-partitions, $\widetilde{\mathcal{G}} = \{G_0,G_1,\cdots, G_{q-1}\}$, of $G$, \changedF{the} one that gives the tightest upper bound on $\mathrm{modaldist}_{G}(\widetilde{\mathcal{G}})$ is the one in which the 
	average relative outgoing weight is the lowest.
	In particular, 
	{\small \begin{equation}
	\min_{\widetilde{\mathcal{G}}} ~ \mathrm{modaldist}_{G}(\widetilde{\mathcal{G}}) 
	~~\leq~~ \changedC{\frac{1}{\displaystyle \lambda_q} \sqrt{\frac{2}{n-q}}}  \displaystyle
		\min_{\widetilde{\mathcal{G}}} ~ \mathrm{avgrelout}_{G}(\widetilde{\mathcal{G}}) 
	\end{equation}}
	\noindent where $\min_{\widetilde{\mathcal{G}}}$ implies minimization over all possible $q$-partitions of $G$.
	
	\item[ii.] Considering the sub-graphs, \changedB{$\{G_j\}_{j=0,\cdots,q-1}$}, as the \emph{clusters},
	this also implies that if there exists a $q$-partitioning of the graph such that the relative outgoing \changedB{weights} of each sub-graph (cluster) are \changedB{small} (\emph{i.e.}, the inter-cluster connections relative to the size of the clusters are weak), then the $q$-modal distance of the graph from that partition will be \changedB{small}\changedF{, \emph{i.e.}, the first $q$ modes of the graph, $G$, will 
	be close to distributions that are uniform over each of the sub-graphs $G_j,j=0,1,\cdots,q-1$}.
	\end{enumerate}

\vspace{0.5em}
\subsubsection{\changed{On the Existence of the ``gap''}} 
\changed{\begin{definition}~[\emph{$\alpha$-realizable $q$-partition}] \label{def:alpha-realizable}
	A graph $G$ is said to have \changedB{an} \emph{$\alpha$-realizable $q$-partition} 
	if there exists a $q$-partition $\widetilde{\mathcal{G}} = \{G_0,G_1,\cdots, G_{q-1}\}$ such that 
	\begin{equation} \textrm{avgrelout}_{G}(\widetilde{\mathcal{G}}) ~\leq~ \frac{\alpha \lambda_q}{\changedC{\sqrt{2q}}} \label{eq:def-alpha-relizable} \end{equation}
	A $q$-partition that satisfies the above condition is called \changedB{an} \emph{$\alpha$-realizable $q$-partition}.
	
	A graph that admits a well-defined $q$-partition (\emph{i.e.}, a $q$-partition with weak inter-cluster connections, and hence low $\textrm{avgrelout}_{G}(\widetilde{\mathcal{G}})$) will thus have \changedB{an} $\alpha$-realizable $q$-partition for a \changedB{small} value of $\alpha$. In particular, we say a graph is $q$-partitionable only if \changedC{it admits} an $\alpha$-realizable $q$-partition for $\alpha\in [0,1)$. \changedB{Also, note that if $\widetilde{\mathcal{G}}$ is an $\alpha$-realizable $q$-partition, then it implies that it is also an $\alpha'$-realizable $q$-partition for any $\alpha' \geq \alpha$. The lowest value of $\alpha$ for which $G$ is $\alpha$-realizable (\emph{i.e.}, $\alpha^* = \frac{\changedC{\sqrt{2q}}}{\lambda_q} \arg\!\min_{\widetilde{\mathcal{G}}} (\textrm{avgrelout}_{G}(\widetilde{\mathcal{G}}))$) thus gives a topological characterization of the graph\changedF{. A} low value implies that the graph is well-clustered and allows a good $q$-way partitioning.}
\end{definition}
\vspace{0.5em}}

\vspace{0.5em}
Suppose we consider the eigenvectors corresponding to the first $q$ eigenvalues of the Laplacian, $L$, of the graph $G$.
Assuming that in a $q$-partition the $q$ strongly-connected clusters in the graph are connected to each other by weak inter-cluster connections, \changedF{then }based on the \changedF{development} in Section~\ref{sec:mode-based-rep-of-clusters} it is expected that $\lambda_0,\lambda_1,\cdots,\lambda_{q-1}$ are small (see Figure~\ref{fig:5smallExampleExample}(g)).
This observation is formalized by the following proposition.

\vspace{0.5em}
\begin{proposition} \label{prop:lambda-q-1-upper-bound}
	If a graph $G$ admits an $\alpha$-realizable $q$-partition, $\widetilde{\mathcal{G}} = \{G_0,G_1,\cdots, G_{q-1}\}$
	(\emph{i.e.}, $\textrm{avgrelout}_{G}(\widetilde{\mathcal{G}}) \leq \frac{\alpha\lambda_q}{\changedC{\sqrt{2q}}}$), with $\alpha\in [0,1)$, then
	\changedC{\begin{eqnarray} \label{eq:lambda-q-1-simplified-upper-bound}
	\frac{\lambda_{q-1}}{\lambda_q} & \leq & \frac{\alpha}{\sqrt{1-\alpha^2}}
%
	\end{eqnarray}}
\end{proposition}

\begin{proof}
	\changedF{See} Appendix~\ref{appendix:prop-lambda-q-1-upper-bound}.
\end{proof}
\vspace{0.5em}

Assuming $G$ admits \changedB{an} $\alpha$-realizable $q$-partition for a \changedB{small} value of $\alpha (< 1)$,
the key insight from the above proposition is that 
%
%
 $\lambda_{q-1}$ will be small \changedC{compared to $\lambda_q$,}
%
giving the formal underpinning for the 
\changedF{eigenvalue gaps as discussed in Section~\ref{sec:mode-based-rep-of-clusters} and illustrated in Figure ~\ref{fig:5smallExampleExample}(g)}.

\changed{
	\vspace{0.5em}
	\subsubsection{\changedB{Sensitivity of Modal Shapes to Edge Weights}} \label{sec:sensitivity}
	
	
	As a result of the above, if $\textrm{avgrelout}_{G}(\widetilde{\mathcal{G}})$ \changedB{and \changedC{hence} $\alpha$ are} small \changedC{(close to zero)}, the vector space spanned by $\mathbf{u}_0,\mathbf{u}_1,\cdots,\mathbf{u}_{q-1}$ is very close to being a degenerate eigenspace (nullspace) of $L$ \changedB{with eigenvalues bounded above by $\frac{\displaystyle \changedC{\alpha \lambda_q}}{\displaystyle \sqrt{1-\alpha^2}}$}. 
	Furthermore, Proposition~\ref{prop:q-conn-dist} indicates that a \changedB{small} $\textrm{avgrelout}_{G}(\widetilde{\mathcal{G}})$ \changedC{results in a small $\textrm{modaldist}_{G}(\widetilde{\mathcal{G}})$, which} implies that $\mathbf{u}_0,\mathbf{u}_1,\cdots,\mathbf{u}_{q-1}$ are distributions that are close \changedF{to linear} combinations of the distributions that are uniform over each $G_j$ (the distributions corresponding to $\overline{\mathbf{u}}_0,\overline{\mathbf{u}}_1,\cdots,\overline{\mathbf{u}}_{q-1}$).
	\changedB{However, the proposition does not give any indication}
	\changedF{of the}
	exact nature of the linear combination of $\overline{\mathbf{u}}_0,\overline{\mathbf{u}}_1,\cdots,\overline{\mathbf{u}}_{q-1}$ that is closest to $\mathbf{u}_0,\mathbf{u}_1,\cdots,\mathbf{u}_{q-1}$ 
	\changedB{which generally depends on the \changedF{specific} value of the weights on the inter- and intra-cluster edges.}
	
}

\section{\changedF{Mode Gradients for Barrier Representation}
} \label{sec:barrier}

\changed{As described in the previous section, the modes corresponding to the smallest $q$ eigenvalues \changedB{prior to} a relatively large jump or gap in the ordered set of eigenvalues form the basis \changedF{for} identifying clusters in a graph. \changedF{While it may be challenging to determine} the value of $q$  just from the ordered set of eigenvalues (in fact in graphs without well-defined clusters there may not even exist a well-defined gap in the ordered set of eigenvalues), for the purpose of discussion in this section we will assume that the value of $q$ and the \changedB{first} $q$ eigenvectors of $L$ are given.}

\subsection{\changed{Gradient of Modes and Resistance}}

\changed{\changedB{As a consequence of Proposition~\ref{prop:q-conn-dist},} each of the first $q$} eigenvectors correspond \changed{to modes with} relatively uniform distributions over each strongly connected cluster \changed{of the graph} \changedB{with a well-defined $q$-partition (small average relative outgoing weight)}.
\changedB{Given that we have $q$ linearly independent modes, each with almost-uniform values over the clusters,} the uniform value on each cluster differs from \changed{the uniform \changedB{values} on its neighboring clusters \changedF{in at least one} of the modes. Consequently, the difference of values across edges that connect clusters will be non-zero in those modes (we call these the inter-cluster \emph{barrier} edges), while inside each cluster differences across edges will be very close to zero because of the intra-cluster uniformity.}

\subsubsection{Gradient of $\mathbf{u}$}

\begin{figure}[h]
\centering
\includegraphics[width=0.95\columnwidth]{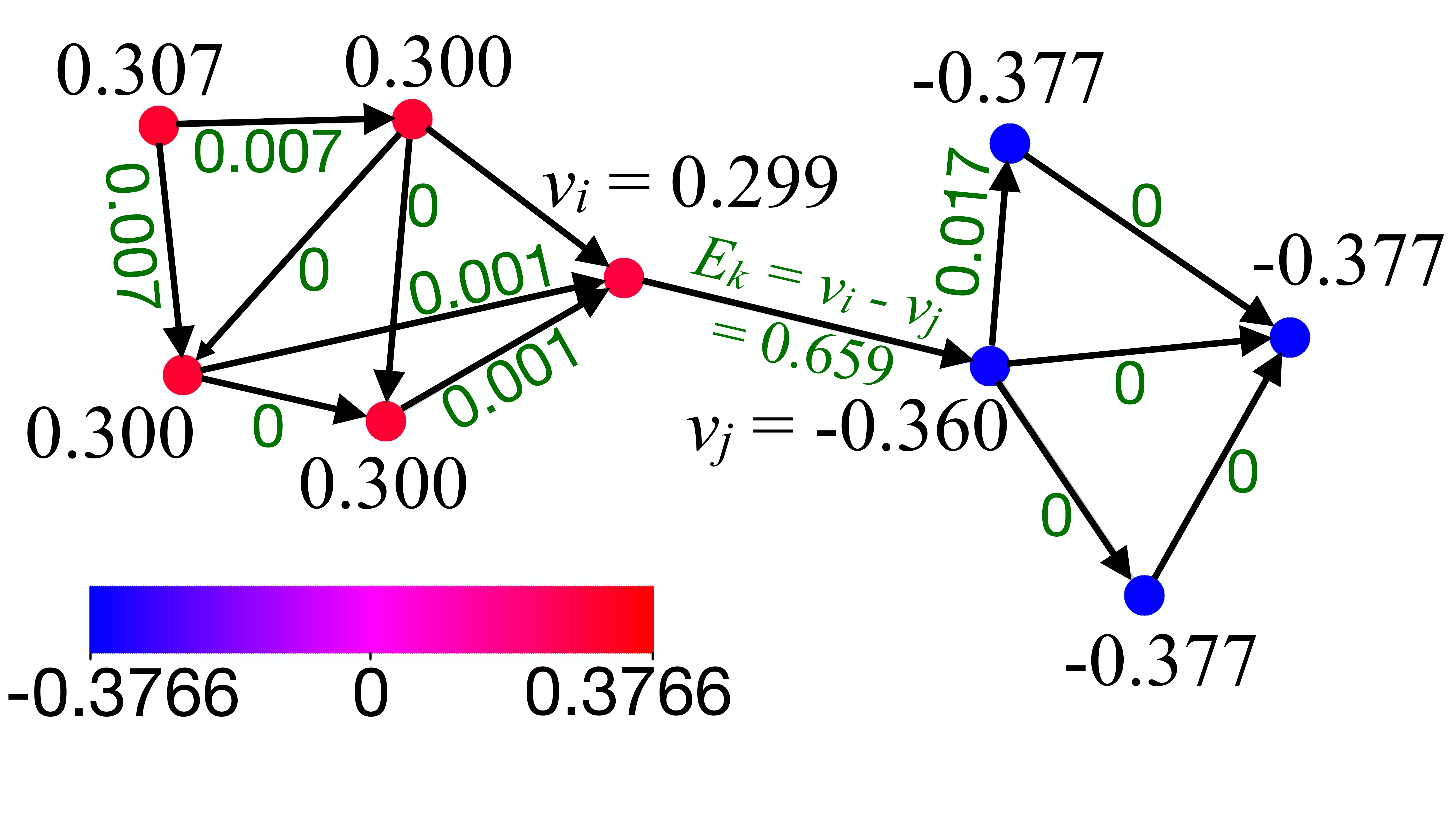}
\vspace{-0.1in} \caption{\changed{\changedF{Gradient example: Colors} and numbers on the vertices \changedF{represent the} distribution corresponding to} the eigenvector $\mathbf{u}_{1}$ of the graph Laplacian. The direction of each edge is arbitrarily \changed{chosen}. 
	The numbers on the edges are gradients along the edges and represent elements of the vector $\mathbf{d}_1 = B^\mathsf{T} \mathbf{u}_1$.}
\label{fig:inceExample}
\end{figure}

\changed{We denote the edge set of the graph as $\mathcal{E}(G) = \{e_1,e_2,\cdots,e_m\}$ and}
give every edge an arbitrary direction. \changed{For a \changedB{given} distribution $\mathbf{u}\in\mathbb{R}^n$ on the vertices,} the difference \changed{of values on two vertices across an edge $e_k$ connecting vertices $v_i,v_j \in \changedB{\mathcal{V}(G)}$ 
is interpreted as} the ``\emph{gradient}'' \changed{of the distribution on the graph}.
\changedF{The gradient yields} 
a real-valued distribution over the edge set \changed{and can thus be written as a vector $\mathbf{d} \in \mathbb{R}^{m}$ in which the} $k$-th element corresponds to the value on edge \changed{$e_k = (v_i, v_j) \in \mathcal{E}(G)$ connecting vertices $v_i,v_j \in \mathcal{V}(G)$ and can be expressed as $d_k = u_j - u_i$.}
\changed{More concisely, \changedF{the} gradient can be written as 
\begin{equation}
\mathbf{d} = 
B^\mathsf{T} \mathbf{u}
\end{equation}
where $\mathbf{u}\in\mathbb{R}^n$ is \changedF{a column vector}, and $B \in \mathbb{R}^{n\times m}$ is the incidence matrix in which the $(i,k)$-th element is defined as
\begin{equation} \label{eq:incidence}
	B_{ik} = \begin{cases}
		1, &\text{if} ~~\exists ~e_{k} = (v_{i}, v_{j}) \in \mathcal{E}(G).\\
		-1, &\text{if} ~~\exists ~e_{k} = (v_{j}, v_{i}) \in \mathcal{E}(G).\\
		0, &\text{otherwise.}
	\end{cases}
\end{equation}

This is illustrated in}
the example of Figure \ref{fig:inceExample} \changed{which shows a modal distribution and the gradient values \changedF{labeled} on each edge.} The inter-cluster edge has relatively large gradient value ($0.659$) while the \changedF{gradients} over the intra-cluster edges are much smaller (order of $10^{-3}$).

%
%
%

Because of the arbitrary direction \changed{assigned} to edges, a gradient can be positive or negative. 
\changed{Thus, to identify the edges that separate the clusters, we square the individual elements of the gradient vector to \changedF{obtain} vector $\mathbf{r} = [d_1^2, d_2^2, \cdots, d_m^2]^\mathsf{T}$, which is a positive-valued distribution over the edges. We call this the \emph{resistance vector} or the \emph{resistance distribution over the edges}, and it is easy to observe that it can be compactly written as
\begin{equation}
\mathbf{r} = \text{diag}(\mathbf{d} \mathbf{d}^\mathsf{T}) = \text{diag}(B^\mathsf{T} \mathbf{u} \mathbf{u}^\mathsf{T} B)
\end{equation}
Figure \ref{fig:inceExampleEdge} shows an example of \changedB{the} resistance vector for $\mathbf{u}_{1}$.
The value of the inter-cluster edge is significantly higher and readily identified.
}

\begin{figure}
	\centering
	\includegraphics[width=0.95\columnwidth]{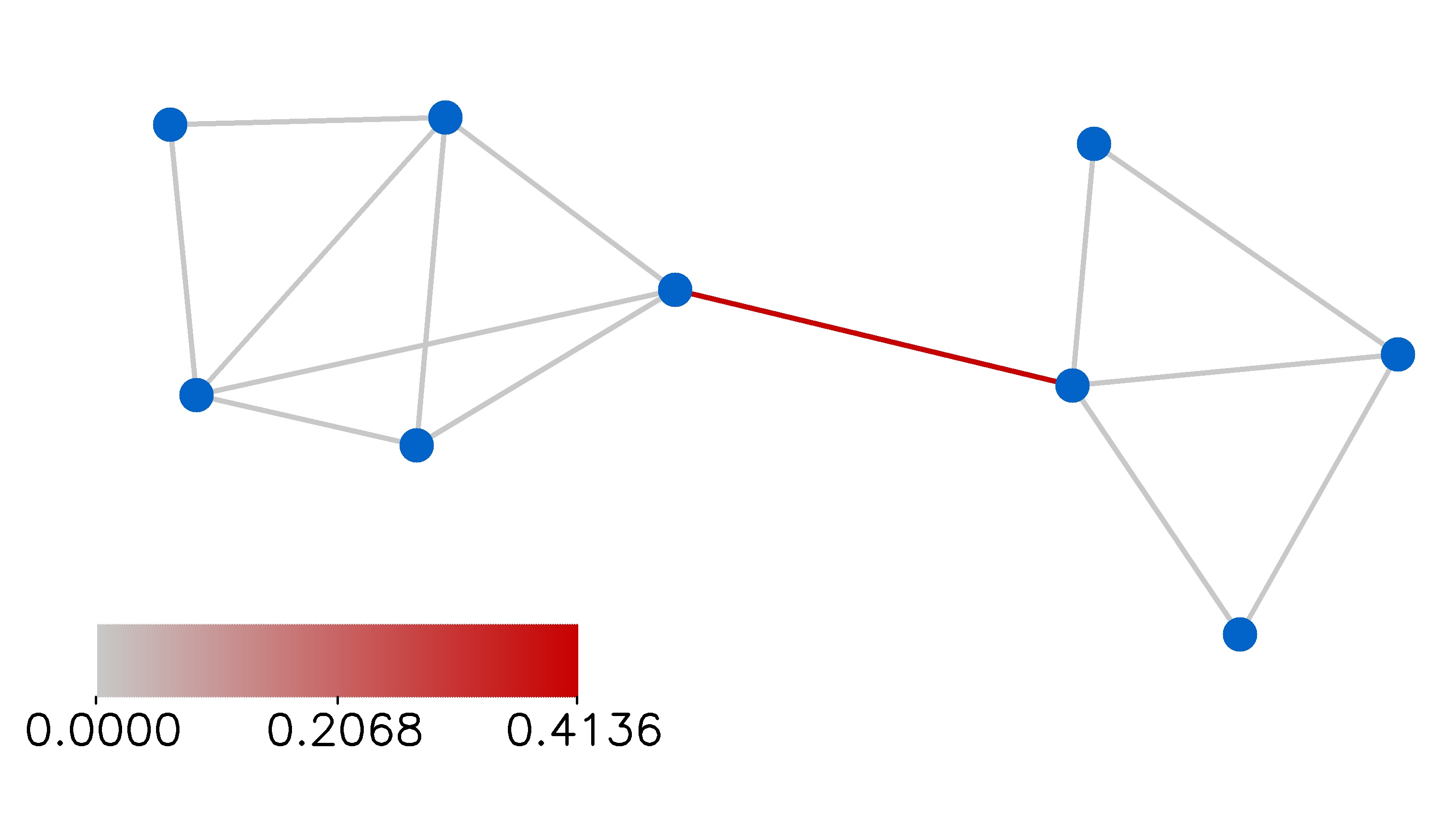}
	\vspace{-0.1in} \caption{\changed{The resistance vector} $\mathbf{r}_{1}$ 
		\changed{as a distribution over the edges. This resistance distribution corresponds to the gradient} $\mathbf{d}_1$ \changed{shown} in Figure \ref{fig:inceExample} and is computed as $\mathbf{r}_1 = \text{diag}(\mathbf{d}_1 \mathbf{d}_1^\mathsf{T})$. Color on an edge indicates the resistance value on the edge.}\vspace{-0.1in} 
	\label{fig:inceExampleEdge}
\end{figure}

\subsubsection{\changed{Aggregated Resistance from Gradient of First $q$ Modes}}

Define \changed{the \emph{resistance vector} corresponding to the $l$-th mode, $\mathbf{u}_l$, as
\begin{equation}
\mathbf{r}_l = \text{diag}(B^\mathsf{T} \mathbf{u}_l \mathbf{u}_l^\mathsf{T} B)
\end{equation}
In order to determine all the barrier edges from the first $q$ modes we aggregate the resistances from the vectors $\mathbf{r}_l, l=0,2,\cdots,q-1$ to define the \emph{$q$-aggregated resistance as}
\begin{eqnarray}
\breve{\mathbf{r}} & = & \sum_{l=0}^{q-1} \mathbf{r}_{l} \nonumber \\
  & = & \sum_{l=0}^{q-1} \text{diag}(B^\mathsf{T} \mathbf{u}_l \mathbf{u}_l^\mathsf{T} B) \nonumber \\
  & = & \text{diag}\left( B^\mathsf{T} \left( \sum_{l=0}^{q-1} \mathbf{u}_l \mathbf{u}_l^\mathsf{T} \right) B \right) \nonumber \\
  & = & \text{diag}\left( B^\mathsf{T} \left( \breve{U} \breve{U}^\mathsf{T} \right) B \right) \label{eq:tilde-r}
\label{eq:directEqn}
\end{eqnarray}}
where 
\begin{equation}
\breve{U} = [\mathbf{u}_{0}~\mathbf{u}_{1}~\cdots~\mathbf{u}_{q-1}]
\label{eq:tildeU}
\end{equation}
\changed{is the $n\times q$ matrix with the columns corresponding to the first $q$ eigenvectors of $L$.}

\begin{figure}
\centering
\includegraphics[width=0.8\columnwidth]{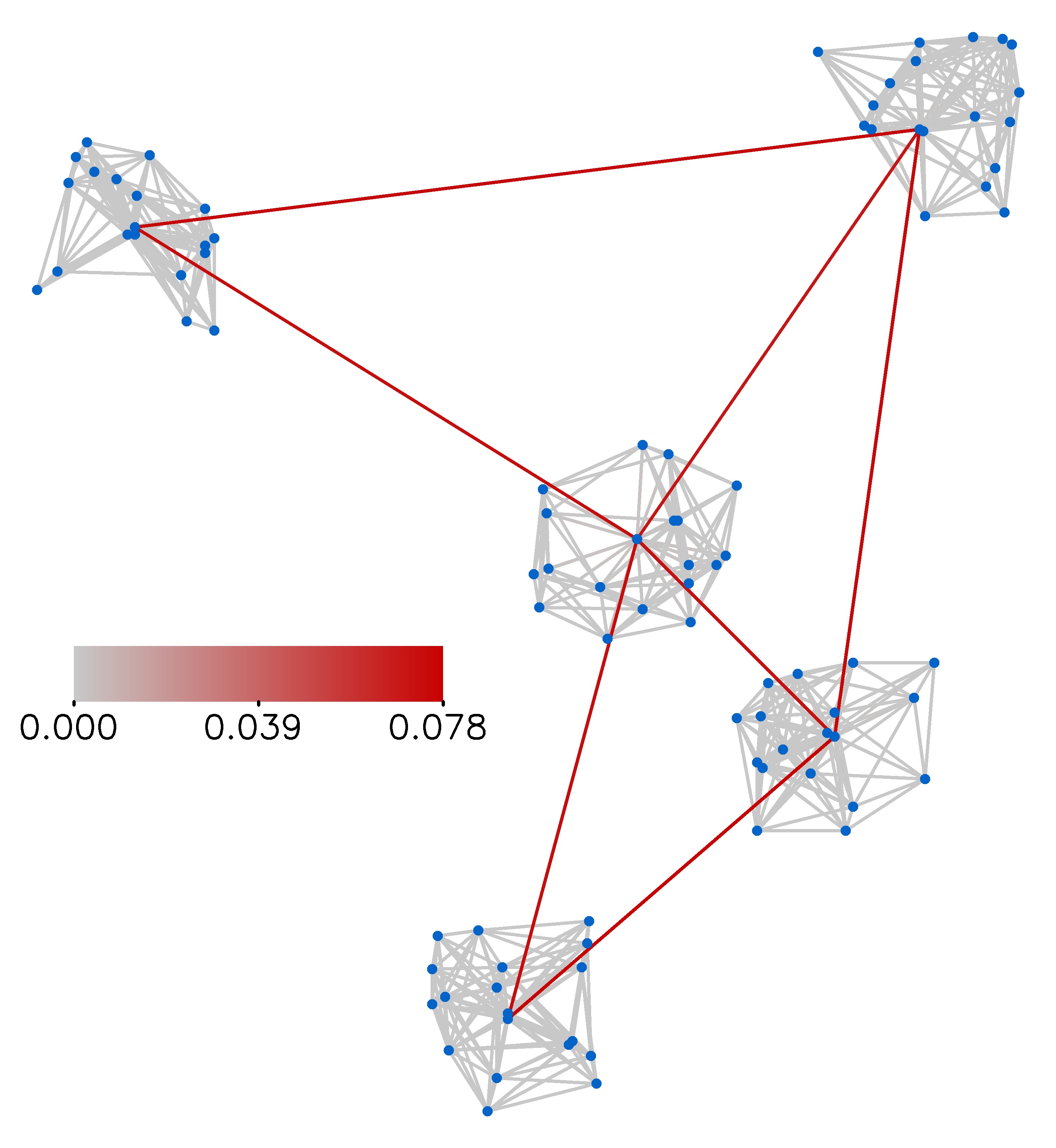}
\vspace{-0.1in} 
\caption{
\changedC{An example of a graph, $G$, with $5$ well-defined clusters:}
\changedB{\changedC{All the edges of $G$ have the same weight of $1$.} With the $5$-partitioned graph, $\widetilde{\mathcal{G}}$, being one with the five well-defined clusters in this graph disconnected, for this example, $\mathrm{modaldist}_G(\widetilde{\mathcal{G}}) = 0.0035$ and $\mathrm{avgrelout}_G(\widetilde{\mathcal{G}}) = 0.6689$. For this graph \changedC{$n=94$ and} $\lambda_5 = 4.0456$,
thus satisfying the inequality of Proposition~\ref{prop:q-conn-dist}.
\changedF{From Definition~\ref{def:alpha-realizable}, it is easy to check that this 
	$\widetilde{\mathcal{G}}$ is an $\alpha$-realizable $5$-partition for any $\alpha \geq 0.523$. 
}
The $5$-aggregated resistance vector, $\breve{\mathbf{r}}$, \changed{is also visualized in this figure as a distribution} over edges -- \changed{darker} shade of red indicates a higher value.
}}\label{fig:multiLinkEdgeValues} \vspace{-0.1in} 
\end{figure}


Figure \ref{fig:multiLinkEdgeValues} shows an example of  $\breve{\mathbf{r}}$.
The inter-cluster edges have \changedF{noticeably} higher values than \changed{intra}-cluster edges.
\changed{Once these values are computed,} we can apply a threshold to values of the vector $\breve{\mathbf{r}}$ 
\changed{to identify the inter-cluster/}bottleneck edges.

\changedB{Due to Definition~\ref{def:q-conn} and Proposition~\ref{prop:q-conn-dist}, if a graph admits a $q$-partition with low $\mathrm{avgrelout}$, the columns of $\breve{U}$ are close to a linear combination of the columns of $\breve{\overline{U}} = [\overline{\mathbf{u}}_{0}~\overline{\mathbf{u}}_{1}~\cdots~\overline{\mathbf{u}}_{q-1}]$. Since the columns are all orthonormal, this implies $\breve{U} \simeq \breve{\overline{U}} R$ for some $q\times q$ orthogonal matrix $R$. Thus we observe that $\breve{U} \breve{U}^\mathsf{T} \simeq \breve{\overline{U}} R R^\mathsf{T} \breve{\overline{U}}^\mathsf{T} = \breve{\overline{U}} \breve{\overline{U}}^\mathsf{T}$ is independent of the exact nature of the linear combination. This shows that the resistance vector, $\breve{\mathbf{r}}$, is independent of the exact nature of the linear combination of the distributions that are uniform over the subgraphs $\{G_j\}_{j=0,1,\cdots,q-1}$ that constitute the low-eigenvalue modes (see Section~\ref{sec:sensitivity}). This general observation is formalized in the following proposition.}


\vspace{0.5em}
\begin{proposition}~[\changedF{\emph{Robustness of Resistance Vector}}] \label{prop:r-robustness}
	\changedB{Consider graphs $G$ and $G'$ with the same topology (\emph{i.e.}, same vertex and edge sets, but possibly different edge weights) such that both permit \changedC{the same} $q$-partition, $\widetilde{\mathcal{G}}$, \changedC{that is $\alpha$-realizable in both the graphs} for some $\alpha\in [0,1)$. Let the resistance vectors \changedC{for the two graphs} be $\breve{\mathbf{r}}$ and $\breve{\mathbf{r}}'$ respectively. Then 
	\begin{equation}
	 \| \breve{\mathbf{r}} - \breve{\mathbf{r}}' \|_2 ~~\leq~~ 2 \alpha \sqrt{2 d_\text{max}}
	\end{equation}
	where $d_\text{max}$ is the maximum degree over the vertices of \changedC{either of the graphs}.
}
\end{proposition}
\begin{proof}
	\changedB{\changedF{See} Appendix~\ref{appendix:prop-r-robustness}.}
\end{proof}
\vspace{0.5em}

\subsection{\changed{Approximations for Distributed Computation Without Explicit Eigendecomposition}} \label{sec:approx}

The computation of $\breve{\mathbf{r}}$ in Equation (\ref{eq:tilde-r}) requires the computation of eigenvectors \changedB{and} eigenvalues, and \changed{determining} how many eigenvectors ($q$) to be used \changed{in computing $\breve{U}$}.
In this section we \changed{propose an approximation that does not require explicit eigendecomposition of $L$ and is \changedB{amenable} to distributed computation. Furthermore, the choice of $q$ is replaced by the choice of a parameter, $\epsilon$, that is directly related to the desired value of $\textrm{avgrelout}_G(\widetilde{\mathcal{G}})$.} 


\subsubsection{Approximation for Mode-independent Computation} \label{sec:approx-mode-independent}

\changed{The following proposition establishes that for a low value of $q$ (compared to $n$) and an appropriately chosen value of $\epsilon$, if there exists a well-defined $q$-partition of the \changedF{graph, then} $\breve{U} \breve{U}^\mathsf{T}$ can be reasonably approximated by $\epsilon (\epsilon I + L)^{-1}$. We call this the \emph{approximation for mode-independent computation}.}

\vspace{0.5em}
\begin{proposition} \label{prop:approximation}
\changed{
Given a desired value for an average relative outgoing weight, $\widehat{a}$, suppose there exists a $q$-partition, $\widetilde{\mathcal{G}}$, such that
\begin{itemize}
	\item[i.] $\widetilde{\mathcal{G}}$ is $\alpha$-realizable with \changedC{some} $\alpha\in[0,\changedC{\widehat{\alpha}}]$, and,
	\item[ii.] $ \textrm{avgrelout}_{G}(\widetilde{\mathcal{G}}) \in [\beta\widehat{a},\,\widehat{a}/\beta]$ for some $\beta\leq 1$ (\emph{i.e.}, the actual average relative outgoing weight of the partition is within a $\beta$ proportion of the desired value of $\widehat{a}$ in either direction).
\end{itemize}
If 
$\epsilon = \changedC{\displaystyle \frac{\sqrt{2q} ~\widehat{a}}{(\widehat{\alpha}\sqrt{1-\widehat{\alpha}^2})^{1/2}} }$\,,
then
\[\frac{ \|\breve{U} \breve{U}^\mathsf{T} - \epsilon (\epsilon I + L)^{-1}\|_{2} }{ \|\breve{U} \breve{U}^\mathsf{T} \|_{2} } 
~<~
 \frac{1}{\beta} \changedC{\sqrt{\frac{\widehat{\alpha}}{\sqrt{1-\widehat{\alpha}^2}}}}
\]
where $\|\cdot\|_{2}$ is the matrix operator $2$-norm.
}
\end{proposition}
\begin{proof}
\changedC{\changedF{See} Appendix~\ref{appendix:prop-approximation}.}
%
%
\end{proof}
\vspace{0.5em}

\changed{As a consequence, \changedC{if the graph admits a good $q$-partition for a desired value of $\mathrm{avgrelout}$}, and if the value of $\epsilon$ is chosen appropriately, the \changedC{$q$-aggregated} resistance vector in \eqref{eq:tilde-r} can be approximated as
\begin{equation*} 
\breve{\mathbf{r}} ~\simeq~ \breve{\mathbf{r}}^\text{aprox-I} ~:=~ \epsilon~ \text{diag}\left( B^\mathsf{T} (\epsilon I + L)^{-1} B \right)
\end{equation*}
}
This approximation automatically factors out the influence from eigenvectors corresponding to larger eigenvalues \changed{\changedF{and only relies on} the contribution for the first $q$ eigenvectors. Figure~\ref{fig:accuracyComparison}(b) illustrates \changedF{an example} in computation of the resistance vector using this approximation.}



\subsubsection{Approximation for Distributed Computation} \label{sec:approx-distributed}

\changed{The computation of \changedF{the} inverse of the matrix $(\epsilon I + L)$, and hence the approximated resistance vector, $\breve{\mathbf{r}}$, as described above is not immediately distributable.}
For \changedF{distributed} computation, \changed{we note that $\epsilon I + L = (\epsilon I + D) -A $, where $\epsilon I + D$ is a diagonal matrix and $A$ is the weighted adjacency matrix.
We define the \emph{left-normalized adjacency matrix} as $\widehat{A}_\epsilon = (\epsilon I + D)^{-1} A$.} $(\epsilon I + L)^{-1}$ can then be further \changedB{approximated as} 
\begin{eqnarray}\begin{array}{rl}
& (\epsilon I + L)^{-1} \\
= & \left( (\epsilon I + D) - A \right)^{-1} \\
= & \left( (\epsilon I + D) (I - \widehat{A}_\epsilon) \right)^{-1} \\
= & (I - \widehat{A}_\epsilon)^{-1} (\epsilon I + D)^{-1} \\ 
= & \left( \sum_{t=0}^{\infty} \widehat{A}_\epsilon^{t} \right) (\epsilon I + D)^{-1} 
			~~\text{\small \changedB{(using \textit{Neumann series}~\cite{yosida2013functional})}} \\
\simeq & \left( \sum_{t=0}^{p} \widehat{A}_\epsilon^{t} \right) (\epsilon I + D)^{-1} ~~\text{\small (partial sum)}
\end{array} \label{eq:neumannSeries}
\end{eqnarray}


\changed{The convergence of the infinite sum in the Neumann series above, and hence the approximation using the partial sum in the last step, 
is formalized in the following proposition:}

\vspace{0.5em}
\begin{proposition}
\changed{	
	~
\begin{enumerate}
	\item $\rho := \| \widehat{A}_\epsilon \|_{\infty} \leq \frac{1}{1+\epsilon/d_{\text{max}}} < 1$, \changedB{where $d_{\text{max}}$ is the maximum degree over the vertices of the graph}.
	\item Define the partial sum $S_N = \sum_{t=0}^{N} \widehat{A}_\epsilon^t$. Then $\| S_{N} \|_{\infty} < \frac{1 - \rho^{N+1}}{1 - \rho}$, $\|S_{\infty}\|_{\infty}$ is bounded above by $\frac{1}{1 - \rho}$, and $\| S_{\infty} - S_{N} \|_{\infty} < \frac{\rho^{N+1}}{1 - \rho}$.
\end{enumerate}
%
}
\end{proposition}
\begin{proof}
	\changed{ 
Part 1: Since the elements of $\widehat{A}_\epsilon$ are all positive,
\begin{eqnarray}
	\|\widehat{A}_\epsilon\|_\infty  & = & \max_{i} \left( \sum_{j} [\widehat{A}_\epsilon]_{ij} \right)\nonumber \\
				& & \qquad~\text{\small (where $[\widehat{A}_\epsilon]_{ij}$ is the $(i,j)$-th element of $\widehat{A}_\epsilon$)} \nonumber \\
	& = & \max_{i} \left(\sum_{j} [(\epsilon I + D)^{-1} A ]_{ij} \right) \nonumber \\
	& = & \max_{i} \left(\sum_{j} \frac{1}{\epsilon + \text{deg}(v_{i})} A_{ij} \right)\nonumber \\
	& = & \max_{i} \left(\frac{1}{\epsilon + \text{deg}(v_{i})} \sum_{j} A_{ij} \right)\nonumber \\
	& = & \max_{i} \left(\frac{\text{deg}(v_{i})}{\epsilon + \text{deg}(v_{i})} \right)\nonumber \\
	& = & \max_{i} \left(\frac{1}{1 + \epsilon / \text{deg}(v_{i})} \right)\nonumber \\
	& = & \frac{1}{1 + \epsilon / d_{\text{max}} } \nonumber \\
	& < & 1
\end{eqnarray}

Part 2:
{\small \begin{eqnarray}
\|S_N\|_\infty ~=~ \left\| \sum_{t=0}^{N} \widehat{A}_\epsilon^t \right\|_\infty 
 & \leq & \sum_{t=0}^{N} \|\widehat{A}_\epsilon \|_\infty^t \nonumber \\
 & =  & \sum_{t=0}^{N} \rho^t ~=~ \frac{1 - \rho^{N+1}}{1 - \rho}
\end{eqnarray}}

Thus, $\lim_{N\rightarrow \infty} \|S_N\|_\infty \leq \frac{1}{1 - \rho}$.

Furthermore,
{\small \begin{eqnarray}
\|S_\infty - S_N\|_\infty & = & \left\| \sum_{t=N+1}^{\infty} \widehat{A}_\epsilon^t \right\|_\infty
~=~ \left\| \widehat{A}_\epsilon^{N+1} \sum_{t=0}^{\infty} \widehat{A}_\epsilon^t \right\|_\infty \nonumber \\
&\leq & \rho^{N+1} \sum_{t=0}^{\infty} \rho^t ~=~ \frac{\rho^{N+1}}{1 - \rho}
\end{eqnarray}}
Thus the partial sums, $S_N$, get closer to $S_\infty$ exponentially fast 
as the value of $N$ increases.
}
\end{proof}
\vspace{0.5em}

\changedC{
The consequence of the above proposition is that the $q$-aggregated resistance vector can be further approximated as
\begin{equation*} 
\breve{\mathbf{r}}^\text{aprox-II} ~:=~ \epsilon~ \text{diag}\left( B^\mathsf{T} \left( \sum_{t=0}^{p} \widehat{A}_\epsilon^{t} \right) (\epsilon I + D)^{-1} B \right)
\end{equation*}
where $\widehat{A}_\epsilon = (\epsilon I + D)^{-1} A$.
We next discuss how the computation of $\breve{\mathbf{r}}^\text{aprox-II}$ can be distributed over the graph.
\vspace{0.5em}
}

\changed{
Since $(\epsilon I + D)$ is a diagonal matrix and $A$ \changedF{is} the adjacency matrix, 
the computation of the $t$-th power of $\widehat{A}_\epsilon = (\epsilon I + D)^{-1} A$ is \changedB{amenable} to distributed computation.
In order to describe the distributed computation algorithm we observe that the $(i,j)$-th element of $\widehat{A}_\epsilon^{t+1}$ can be written as}
\begin{eqnarray}
[\widehat{A}_\epsilon^{t + 1}]_{ij} & = & \sum_{l} [\widehat{A}_\epsilon]_{il} ~[\widehat{A}_\epsilon^{t}]_{lj} \nonumber \\ 
  & = & \changed{\sum_{\changedB{\{l | v_l \in \mathscr{N}_G(v_i)\}}} [\widehat{A}_\epsilon]_{il} ~[\widehat{A}_\epsilon^{t}]_{lj}} \nonumber \\
   & = & \changed{\sum_{\changedB{\{l | v_l \in \mathscr{N}_G(v_i)\}}} \frac{A_{il}}{\epsilon + D_{ii}} ~[\widehat{A}_\epsilon^{t}]_{lj}}
\label{eq:iterN}
\end{eqnarray}
\changed{where $\mathscr{N}_G\changedB{(v_i) = \{v_j \,|\, (v_i,v_j) \in \mathcal{E}(G) \}}$ indicates the set of vertices that are neighbors of the $i$-th vertex.
The \changedB{second} equality in \changedF{(\ref{eq:iterN})}  holds because $[\widehat{A}_\epsilon]_{il}$ is non-zero only if \changedB{$(v_i,v_l)\in \mathcal{E}(G)$}.

In a distributed computation, the vertex $i$ will store only the elements of the $i$-th row of $\widehat{A}_\epsilon^{t}$ and the computation of the $i$-th row of the partial sum matrix, $S\!p := \sum_{t=0}^{p} \widehat{A}_\epsilon^t$. In practice each vertex can maintain associative arrays for the rows, which are populated as new non-zero entries are computed.
This is illustrated in Algorithm~\ref{alg:distributedComputationt}.}

\changed{
\begin{algorithm}[h]
	\caption{\changedB{Distributed} computation of $i$-th row of the partial sum $S\!p$, \changedB{and the resistance values on the edges emanating from $v_i$.}}\label{alg:distributedComputationt} 
	
	\textbf{Inputs:} 
	\textit{\textbf{i.}} $(i,l)$-th element of Adjacency Matrix, $\mathsf{A}_{il} ~\forall~ \changedB{\{l \,|\, v_l \in \mathscr{N}_G(v_i)\}}$; 
	~\textit{\textbf{ii.}} Degree of the vertex, $\mathsf{D}_i$; 
	~\textit{\textbf{iii.}} $\epsilon$.
	
	\textbf{Outputs:}
	\textit{\textbf{i.}} Associative array $\mathsf{Sp}_i$ representing $i$-th row of $S\!p$;
	\textit{\textbf{ii.}} \changedB{For every $ \changedB{\{l \,|\, v_l \in \mathscr{N}_G(v_i)\}}$, the value of the resistance $\breve{\mathsf{r}}^\text{aprox-II}[i,l]$ on the edge $(v_i,v_l)$.}
	
	\begin{algorithmic}[1]
		\State Associative array $\mathsf{\widehat{A}Pow}_i \leftarrow \{(i\!:1)\}$ \algcomment{identity matrix}
		\State $\mathsf{t}_i \leftarrow 0$ \algcomment{iteration number for data in $\mathsf{\widehat{A}Pow}_i$}
		\State Associative array $\mathsf{\widehat{A}PowTmp}_i \leftarrow \{\}$ 
		\State $\mathsf{done}_i \leftarrow 0$ \algcomment{iteration number for data in $\mathsf{\widehat{A}PowTmp}_i$}
		\State Associative array $\mathsf{Sp}_i \leftarrow \{\}$
		\State \changedB{Associative array $\breve{\mathsf{r}}^\text{aprox-II} \leftarrow \{\}$}
		
		\While {$\mathsf{t}_i < p$}
			\For {$\changedB{\{l \,|\, v_l \!\in\! \mathscr{N}_G(v_i)\}}$} \algcomment{communicate w/neighbors} 
				\While{ $\mathsf{t}_l < \mathsf{t}_i$} 
					\textbf{wait} \algcomment{wait for latest data} 
				\EndWhile
				\For {$(j\!:\changedB{\mathsf{s}_j})\in \mathsf{\widehat{A}Pow}_l$} 
					\State $\mathsf{\widehat{A}PowTmp}_i[j] \leftarrow \frac{\mathsf{A}_{il}}{\epsilon + \mathsf{D}_i} * \changedB{\mathsf{s}_j}$ 
				\EndFor
			\EndFor
			\State $\mathsf{done}_i \leftarrow \mathsf{t}_i + 1$
			\For {$\changedB{\{l \,|\, v_l \!\in\! \mathscr{N}_G(v_i)\}}$} \algcomment{wait for neighbors} 
				\While{$\mathsf{done}_l < \mathsf{t}_i + 1$}
					\textbf{wait}
				\EndWhile
			\EndFor
			$\mathsf{\widehat{A}Pow}_i \leftarrow \mathsf{\widehat{A}PowTmp}_i$
			
			\State $\mathsf{t}_i \leftarrow \mathsf{t}_i + 1$ 
			

			\For {$(j\!:\mathsf{s}_j)\in \mathsf{\widehat{A}Pow}_i$} \algcomment{update partial sum}
				\If{$\mathsf{Sp}_i[j]$ does not exist}
					\State $\mathsf{Sp}_i[j] \leftarrow \mathsf{s}_j$
				\Else
					\State $\mathsf{Sp}_i[j] \leftarrow \mathsf{Sp}_i[j] + \mathsf{s}_j$
				\EndIf
			\EndFor
		\EndWhile
\changedB{
		\For {$\changedB{\{l \,|\, v_l \!\in\! \mathscr{N}_G(v_i)\}}$} \algcomment{wait for neighbors}
			\While{ $\mathsf{t}_l < p$} 
				\textbf{wait}  
			\EndWhile
		\EndFor
		\For {$\changedB{\{l \,|\, v_l \!\in\! \mathscr{N}_G(v_i)\}}$} \algcomment{communicate w/neighbors}
			\State $\breve{\mathsf{r}}^\text{aprox-II}[i,l] \leftarrow \epsilon \left(
						\frac{ \mathsf{Sp}_{i}[i] - \mathsf{Sp}_{l}[i]}{\epsilon + \mathsf{D}_{i}} 
						+ \frac{\epsilon \mathsf{Sp}_{l}[l] - \mathsf{Sp}_{i}[l]}{\epsilon + \mathsf{D}_{l}} 
						\right) $
		\EndFor
}
	\end{algorithmic}
\end{algorithm}
}

\changed{
Suppose \changedB{the $k$-th edge connects the $i$-th to the $l$-th vertex (\emph{i.e.}, $v_l\in \mathscr{N}_G(v_i)$ and $e_k = (v_i,v_l) \in \mathcal{E}(G)$)}.
Using the $i$- and $l$-th rows of the partial sum matrix, $S\!p$, stored on vertices $i$ and $l$ respectively, and using their respective degrees, 
the resistance on the edge (the $k$-th component of the resistance vector) can be \changedB{computed} in a distributed manner
as follows:
\vspace{0.5em}
\begin{eqnarray}
\breve{r}_k & \simeq & \breve{r}_k^\text{aprox-II}\nonumber \\
& := &
\epsilon \left[ B^{\mathsf{T}} \left( \sum_{t=0}^{p} \widehat{A}_\epsilon^{t} \right) (\epsilon + D)^{-1} B \right]_{kk}
\nonumber \\
& = & \epsilon \sum_{\changedB{a,b}} B_{ak} ~{S\!p}_{ab} 
~\frac{1}{\epsilon + D_{bb}}
~B_{bk} \nonumber \\
& = & \epsilon \sum_{a,b \in \{i,l\}} B_{ak} ~{S\!p}_{ab} 
~\frac{1}{\epsilon + D_{bb}}
~B_{bk} \nonumber \\
& = & \frac{\epsilon {S\!p}_{ii}}{\epsilon + D_{ii}} - \frac{\epsilon {S\!p}_{li}}{\epsilon + D_{ii}} 
+ \frac{\epsilon {S\!p}_{ll}}{\epsilon + D_{ll}} - \frac{\epsilon {S\!p}_{il}}{\epsilon + D_{ll}}  \quad
\end{eqnarray}
where the first approximation \changedF{follows} from Equation \eqref{eq:tilde-r}, Proposition~\ref{prop:approximation} \changedF{and} Equation \eqref{eq:neumannSeries}, and the simplification of the summation domain from \changedB{every pair of vertex indices} to $\{i,l\}$ is due to the fact that $B_{ck}$ is zero for any $c\notin \{i,l\}$.
The important thing to note here is that this approximate computation of the resistance on the $k$-th edge requires only values stored on its bounding vertices, thus making this computation distributable as well.
\changedB{The complete algorithm for the distributed computation of the resistances is described in Algorithm~\ref{alg:distributedComputationt}.}

 Figure~\ref{fig:compDistributed} illustrates a result in computation of the resistance vector using this distributed implementation.
}


\section{\changed{Inter-cluster Transmission Rate Control Using Modal Barriers}} \label{sec:information-containment}

\subsection{Transmission \changed{Model}}

\begin{figure}
	\centering
	\includegraphics[width=0.8\columnwidth]{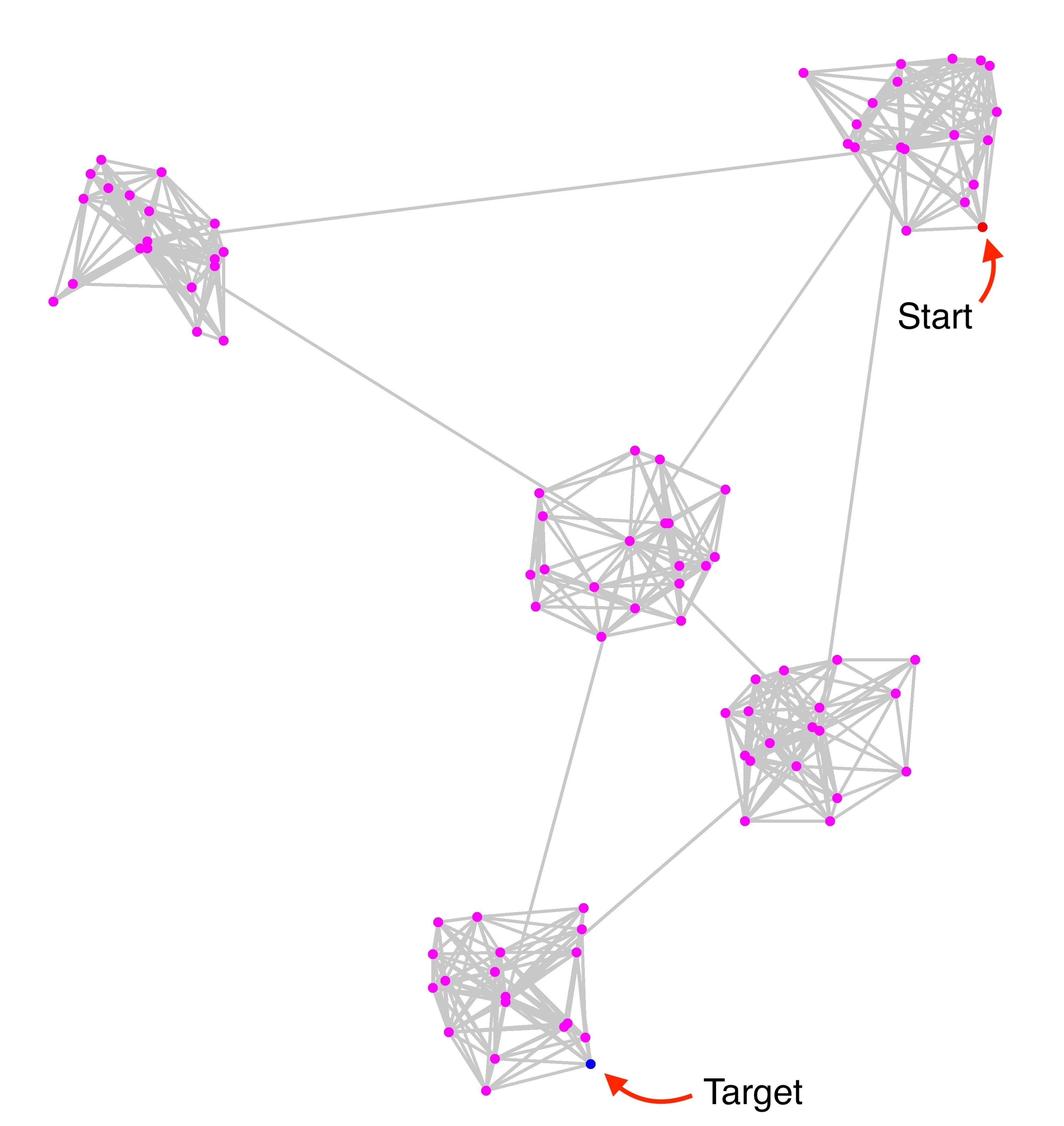}
	\vspace{-0.1in} \caption{Example of start and target vertices for transmission. \changed{This figure shows the setup for the simulation \changedF{result} in Figure~\ref{fig:trans_sim}}.}\vspace{-0.1in} 
	\label{fig:start_target_position}
\end{figure}

\changed{For a given graph and weighted Laplacian matrix, $L' = D' - A'$ (where $A'$ is the weighted adjacency matrix, and $D'$ the corresponding diagonal degree matrix),} consider the following \changed{dynamics on the graph}

\begin{equation} \label{eq:diff_equation}
\dot{\mathbf{x}} = - L' \mathbf{x}
\end{equation}
\changed{where $\mathbf{x}\in\mathbb{R}^n$ corresponds to a distribution over the vertices of the graph. 
	
A first-order} discrete time form of equation \eqref{eq:diff_equation} can be \changed{written} as 
\begin{eqnarray} \label{eq:discrete_eq}
          \mathbf{x}(\tau + \Delta \tau) &=& (I - \Delta \tau ~ L) \mathbf{x}(\tau) \nonumber \\
           \text{or more simply,} \quad     \mathbf{x}^{(t + 1)} &=& (I - \kappa L) \mathbf{x}^{(t)}     
\end{eqnarray}
\changed{where $\tau\in\mathbb{R}$ is the continuous time, while $t\in\mathbb{Z}$ is the (re-scaled) discrete time.}
\changed{$\kappa \ll \frac{2}{\lambda_n}$ is a small number chosen to ensure that $(I-\kappa L)$ has the same \changedF{eigenspace} as $L$, but its highest eigenvalue is $1$ corresponding to the eigenvector $\mathbf{u}_0$, while all other eigenvalues are in $(-1,1)$.}

As $t \rightarrow +\infty$, $\mathbf{x}^{(+\infty)}$ thus converges to the first eigenvector, $\mathbf{u}_{0}$, of $L$, which \changed{corresponds to \changedF{a} uniform distribution over the vertices. The effect of \changedF{the} dynamics is \emph{diffusion} -- starting with any non-zero distribution $\mathbf{x}^{(0)}$, the dynamics of \eqref{eq:diff_equation} or \eqref{eq:discrete_eq} ends up \changedF{re-distributing} the values on the vertices to their neighbors until an equilibrium (uniform distribution) is attained. 
\changedB{In fact the dynamics in \eqref{eq:diff_equation} can be written as
	\begin{equation} \begin{array}{rcl}
	\dot{x}_i & = & \displaystyle - (D'_i x_i ~- \!\!\!\! \sum_{\{j| \atop (v_i,v_j)\in\mathcal{E}(G)\}} \!\!\!\! A'_{ij} x_j ) \\
	& = & \displaystyle \!\!\!\! \sum_{\{j| \atop v_j \in \mathscr{N}_G(v_i)\}} \!\!\!\! A'_{ij} x_j ~-~ \left( \sum_{\{j| \atop v_j \in \mathscr{N}_G(v_i)\}} \!\!\!\! A'_{ij} \right) x_i
	\end{array} \label{eq:dynamics-i} \end{equation}
	In this form it's easy to see that the \changedF{diffusion} rate from $v_j$ to $v_i$ (\emph{i.e.}, across the edge $(v_i,v_j)$) is determined by the weight on the edge, $A'_{ij}$ (the $(i,j)$-th element of the weighted adjacency matrix). \changedF{The lower the rate, the slower the} diffusion across the edge. 
	Furthermore, it is clear that \changedF{the} dynamics required only neighborhood communication/transmission, since updating $x_i$ requires the values of $x_j$ for only the neighbor indices in $\{j \,|\, v_j \in \mathscr{N}_G(v_i)\}$.
}
	
We thus use this as a model for \changedF{network transmission} (\changedF{such as} information \changedF{in a communication network context}, or infection in context of \changedF{disease spread}) in the network. \changedF{The entity being transmitted is represented by the values of the elements of $\mathbf{x}$, with $X_k$ denoting the value on vertex $v_k$.}

In particular we consider the situation where the \changedF{entity being transmitted} originates from one particular \changedF{start} vertex, \changedF{$v_\mathsf{start}$} (\emph{i.e.}, the $k$-th element of $\mathbf{x}^{(0)}$ is $x_k^{(0)} = 1, \text{when}~k=s,~\text{or}~ x_k^{(0)} = 0~\text{otherwise}. $). We are interested in studying (and reducing) the time taken by the dynamics to make the value on a \changedF{target} vertex, \changedF{$v_\mathsf{target}$} (belonging to a different \changedF{cluster} than $v_\mathsf{start}$), go above a threshold value of $\gamma$, And we refer to the value of the target vertex $v_\mathsf{target}$ as $x_\mathsf{target}$}.
Figure \ref{fig:start_target_position} shows an example of start and target vertices for transmission on a simple graph.

\subsection{Edge Weight Computation} \label{sec:weights}

\changed{
%
For \changedB{the edge} weights 
\changedF{modeling the} dynamics \changedB{(elements of the weighted adjacency matrix that govern the dynamics \eqref{eq:dynamics-i})}
we consider three different cases:}
\begin{itemize}
	\item[i.] \changed{\emph{Unrestricted Transmission:} In this case we set \changedF{a} uniform unit \changedF{weight} for every edge in the graph. The corresponding adjacency matrix, degree matrix and the Laplacian matrix are denoted by $A^U, D^U$ and $L^U$. The (continuous-time) dynamics of unrestricted transmission is thus $\dot{\mathbf{x}} = - L^U \mathbf{x}$.}
	\item[ii.] \changed{\emph{Inter-cluster Barriered Transmission:} In order to reduce transmission across the network we propose to decrease weights on selected edges (reducing the rate of transmission across those edges). 
	For this approach, we use the resistance vector computed using the Laplacian for unrestricted transmission, $L^U$, to compute new weights. In particular, if $\breve{r}_k$ is the resistance computed for the $k$-th edge, we define the new \changedF{barrier weights} as 
	\begin{equation} \label{eq:bottleneckWeightVectorEle}
		w^B_{k} = \frac{\epsilon_B}{\epsilon_B + 
			\breve{r}_{k}}
	\end{equation}
	where $0 < \epsilon_B \ll 1$ is a small 
	positive \changedF{constant}. These barrier weights thus take \changedF{values} in $(0,1]$ \changedF{and} edges with computed resistance close to $0$ (such as the intra-cluster edges) \changedF{have} barrier weights close to $1$, while the barrier weights are lower for edges on which the computed resistance is high (for example the inter-cluster edges) 
	\footnote{Note that a lower weight implies lower transmission rate across an edge, and hence constitutes a \changedF{barrier}, while weights closer to $1$ correspond to \changedF{less restricted} transmission.}.
	The corresponding adjacency matrix, degree matrix and Laplacian are thus 
	\begin{eqnarray}
	A^B_{ij} & = & \left\{ \begin{array}{l} w^B_{k}, ~~\text{if } \exists~ e_k = (v_i,v_j) \in \mathcal{E}(G) \text{ or } \\ \qquad\qquad \exists~ e_k = (v_j,v_i) \in \mathcal{E}(G) \\ 0, ~~\text{otherwise.}\end{array} \right. \nonumber \\
	D^B & = & \mathrm{diag}_j \left( \sum_{i=1}^n A^B_{ij} \right) \nonumber \\
	L^B & = & D^B - A^B
	\end{eqnarray}
	The (continuous-time) dynamics of transmission with inter-cluster barrier is then defined as $\dot{\mathbf{x}} = - L^B \mathbf{x}$.}
	\item[iii.] \changed{\emph{Shuffled Weight Transmission:} In order to be able to compare/benchmark the performance of the barriered transmission described above, we randomly shuffle the \changedF{barrier weights} as computed in \eqref{eq:bottleneckWeightVectorEle} across the different edges of the graph. \changedF{We} define \emph{shuffled weight} on the $k$-th edge as $w^S_k = w^B_{\sigma(k)}$, for some permutation, $\sigma$, of $(1,2,\cdots,m)$.
	The corresponding adjacency matrix, degree matrix and Laplacian matrix are denoted by $A^S, D^S$ and $L^S$, and the (continuous-time) dynamics of transmission with shuffled weights is defined as $\dot{\mathbf{x}} = - L^S \mathbf{x}$.
	
	\changedF{If we regard decreasing the weight on an edge below 1 as consuming a certain amount of available resource}, the simple shuffling of the weights ensure that the same amount of \changedF{resource} is being used in lowering the weights (\emph{i.e.}, the same number of edges with the same weights as in the barrierred transmission case), but the edges are chosen randomly without using our method for mode-based barrier construction. This gives a baseline method to compare against.}
\end{itemize}



%


\section{Results} \label{sec:results}

\subsection{Inter-cluster Edge Detection Comparison} \label{sec:detectionComp}

We compare the performance of our method in detecting inter-cluster edges with a recent algorithm~\cite{cheng2017network} based on multiway Cheeger partitioning that uses an integer linear programming formulation to find partitions.
For this algorithm, the number of \changedF{clusters} is required as input and, for a total of $q$ clusters and $n$ vertices, the first $q - 1$ clusters are required to contain \changedF{less than an average of} $n/q$ vertices.
This requirement affects the accuracy of the bottleneck detection as shown in the example of Figure \ref{fig:cheegerExample}.
In this case, some vertices are forced to be separated from their well-aggregated cluster due to \changedF{the} vertex count limit per cluster, resulting in sub-optimal inter-cluster edge detection. 
Furthermore, because of the integer linear programming implementation, the multiway Cheeger partition algorithm took $54,304.6$ seconds to compute the partitions in this particular example on \changedF{an intel cpu}.
In comparison, our method took less than one second to 
compute the \emph{resistance}, \changedB{$\breve{\mathbf{r}}$,} on the edges \changedB{using \changedF{(\ref{eq:directEqn}})}, \changedF{and the result is shown}
in Figure \ref{fig:multiLinkEdgeValues}.


\begin{figure}
	\centering
	{\includegraphics[width = 0.8\columnwidth]{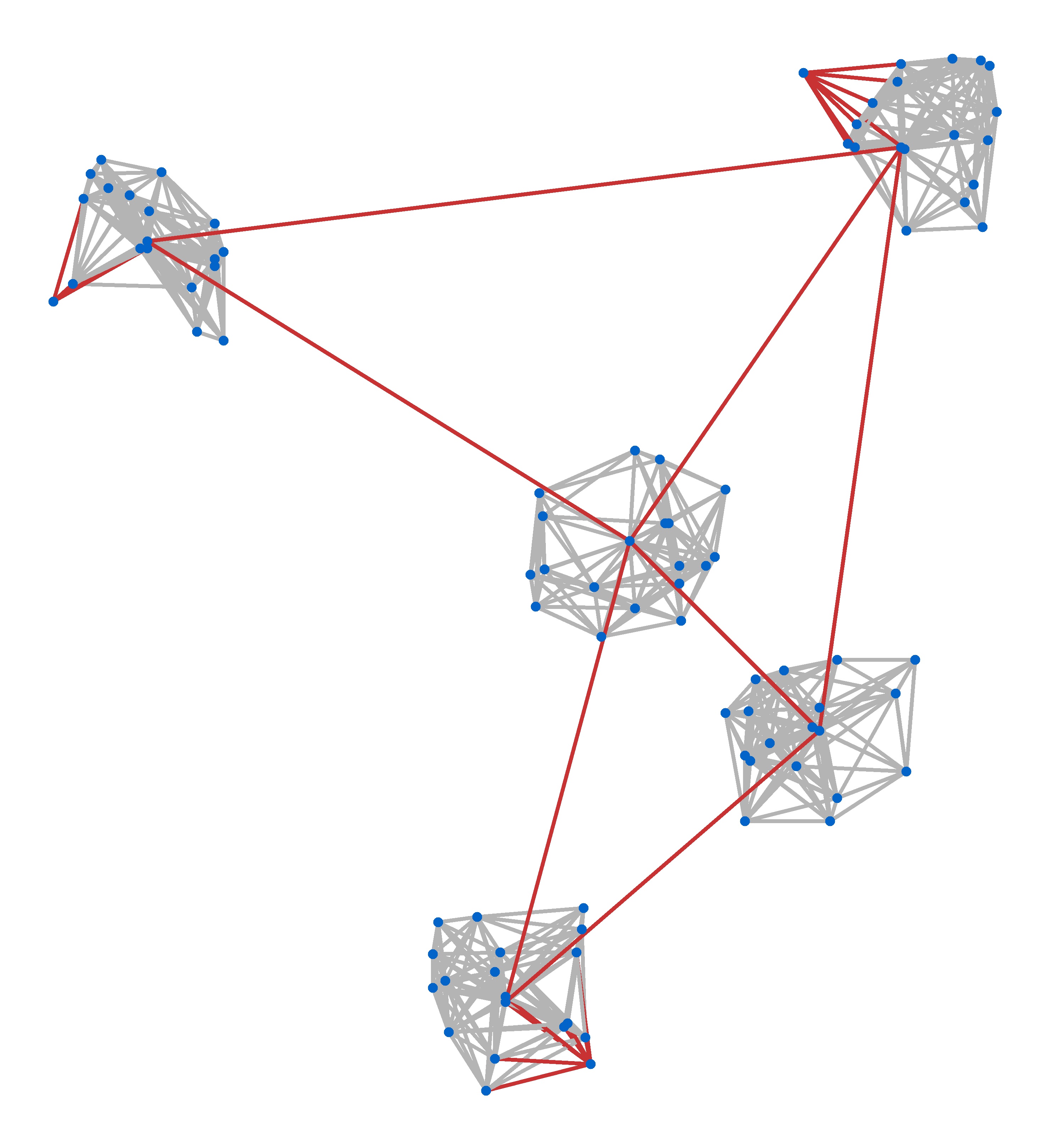}} \hspace{0.0in}
	\caption{Detected bottleneck edges using \changedF{the} multiway Cheeger partition algorithm~\cite{cheng2017network}.
	\changedF{The method is less accurate in clustering and significantly more computationally complex than our proposed resistance computation, illustrated in Figure ~\ref{fig:multiLinkEdgeValues} (see text Section \ref{sec:detectionComp}). }}
	\label{fig:cheegerExample}
\end{figure}

\subsection{Transmission Rate} \label{sec:transRate}

\begin{figure}
	\centering
	\includegraphics[width=0.95\columnwidth]{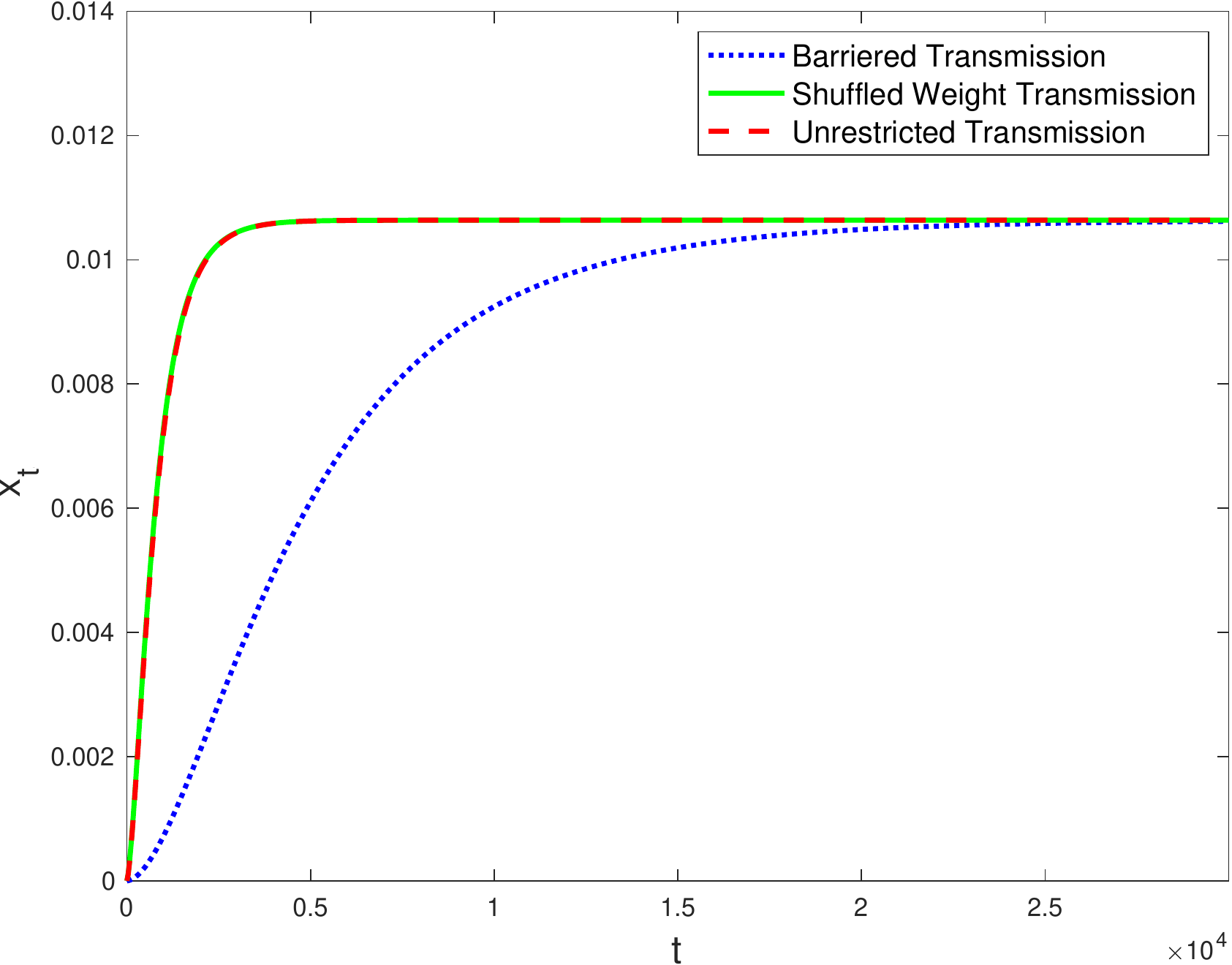}
	\vspace{-0.1in} \caption{Discrete time simulations of information transmission \changed{using the setup of Figure~\ref{fig:start_target_position}.
	\changedF{The barriered transmission slows the information conveyed across the network from start to target (see the example description, Section \ref{sec:transRate} ).}}}
	\vspace{-0.1in} 
	\label{fig:trans_sim}
\end{figure}

\changed{\changedF{A} comparison of transmission rates using the different weight assignment methods described in Section~\ref{sec:weights} is shown in Figure \ref{fig:trans_sim}.
As can be seen from the plot, the value of $\changedF{x_\mathsf{target}}$ grows significantly slower when using the \changedF{Barriered transmission} weights than using the \changedF{Unrestricted transmission} weights or the \changedF{Shuffled weights}.
While eventually any non-zero weights will result in convergence to a uniform distribution (corresponding to the null-space vector of either of $L^U$, $L^B$ or $L^S$), the objective of slowing the rate of transmission using the \changedF{barriered} weights is clearly achieved.}

\begin{figure}[h]
\centering
\includegraphics[width=0.95\columnwidth]{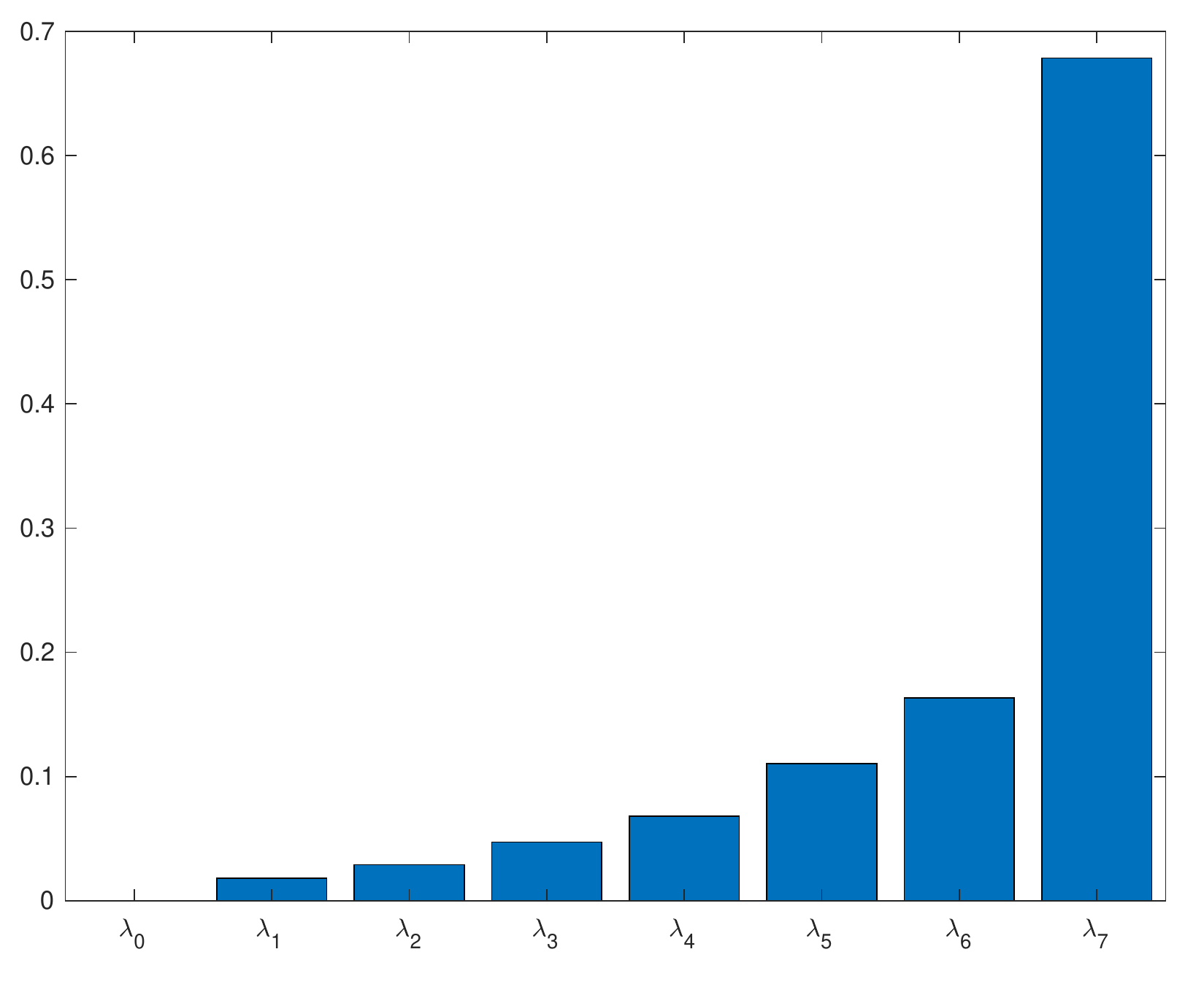} 
\vspace{-0.1in} \caption{First (smallest) 8 eigenvalues of the graph Laplacian of ego-Facebook network}\vspace{-0.1in} 
\label{fig:eigval_facebook}
\end{figure}

\begin{figure}[h]
\centering
\includegraphics[width=0.95\columnwidth]{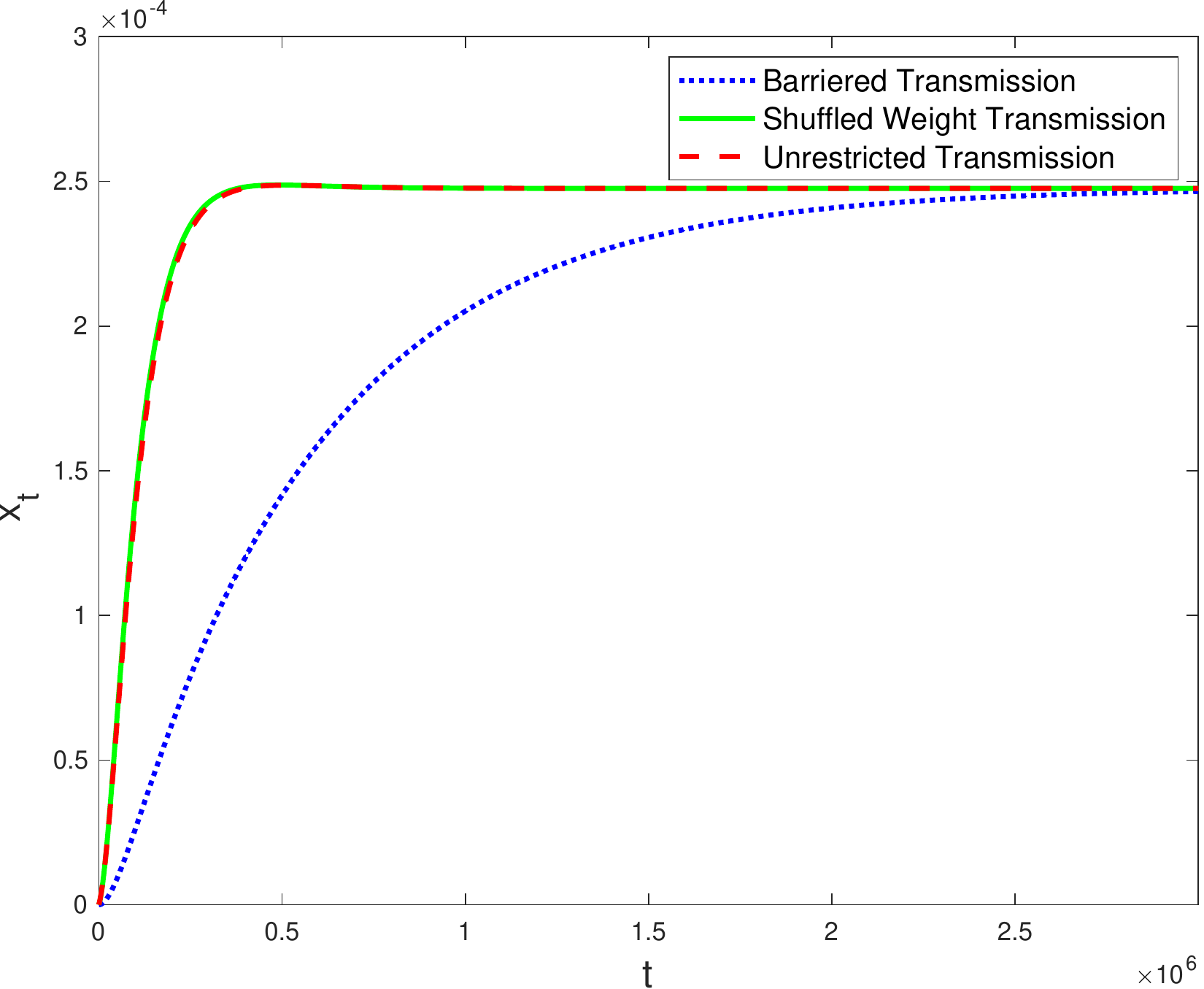}
\vspace{-0.1in} \caption{Discrete time simulations of information transmission in the ego-Facebook network. 
}\vspace{-0.1in} 
\label{fig:facebookTrans}
\end{figure}

For a large scale \changedF{network graph example}, we use an open network dataset called \textit{ego-Facebook} from \changedF{the} \textit{Stanford Large Network Dataset Collection}~\cite{snapnets}.
This anonymized social network has $4039$ nodes and $88234$ edges in total.
As shown is Figure \ref{fig:eigval_facebook}, there is a big gap between $\lambda_{6}$ and $\lambda_{7}$ \changed{for this network}.
We choose $q = 7$ in Equation \eqref{eq:tildeU} to construct our barrier weighted graph Laplacian. 
We also choose node $\#1$ as the start vertex and node $\#4038$ as the target vertex \changed{for information transmission across the network.}
The simulation result in Figure \ref{fig:facebookTrans} shows a similar trend as the result of the small scale simulation (Figure \ref{fig:trans_sim}).
The \changed{barrier} edge weights are able to slow the information transmission effectively.

\subsection{Comparison of Approximate Methods for Distributed Computation}

We use 
$3$ different methods 
for computing $\breve{U}\breve{U}^{\mathsf{T}}$ for computing the resistance vector, $\breve{\mathbf{r}}$, in equation~(\ref{eq:tilde-r}).
\changed{The first is the direct eigen-decomposition of the graph Laplacian for computing $\breve{U}$. The second is the approximate method for \changedF{mode-independent computation} as was described under Section~\ref{sec:approx-mode-independent}. And finally we use the approximation method for \changedF{distributed computation} as was described in Section~\ref{sec:approx-distributed}.}


\begin{figure*}
\centering
\subfloat[\changed{Direct computation of $\breve{\mathbf{r}}$.}]{\includegraphics[width = .33\linewidth]{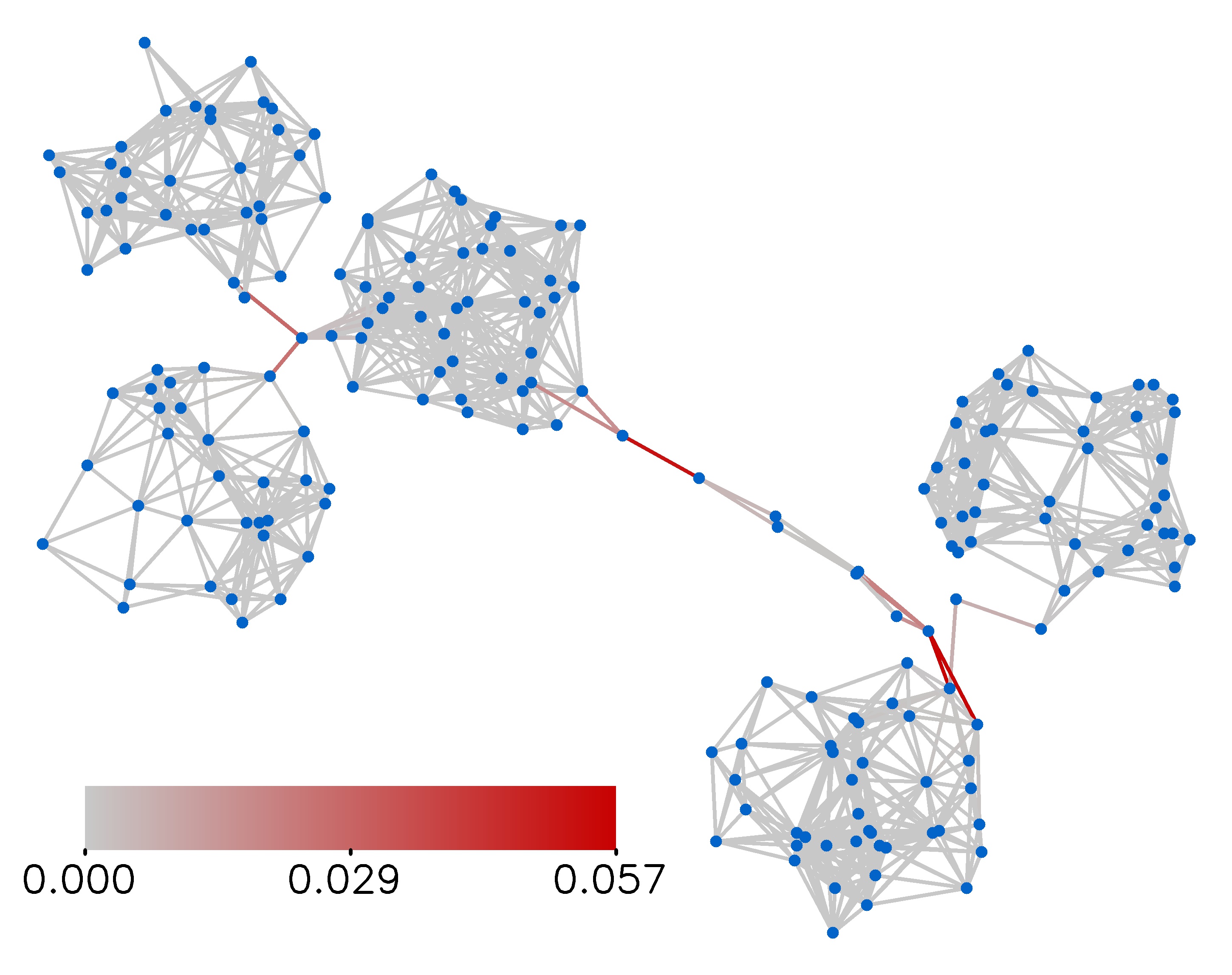}\label{fig:compBench}}
\subfloat[\changed{Approximation for mode-independent computation for computing $\breve{\mathbf{r}}_\text{aprox-I}$.}]{\includegraphics[width = .33\linewidth]{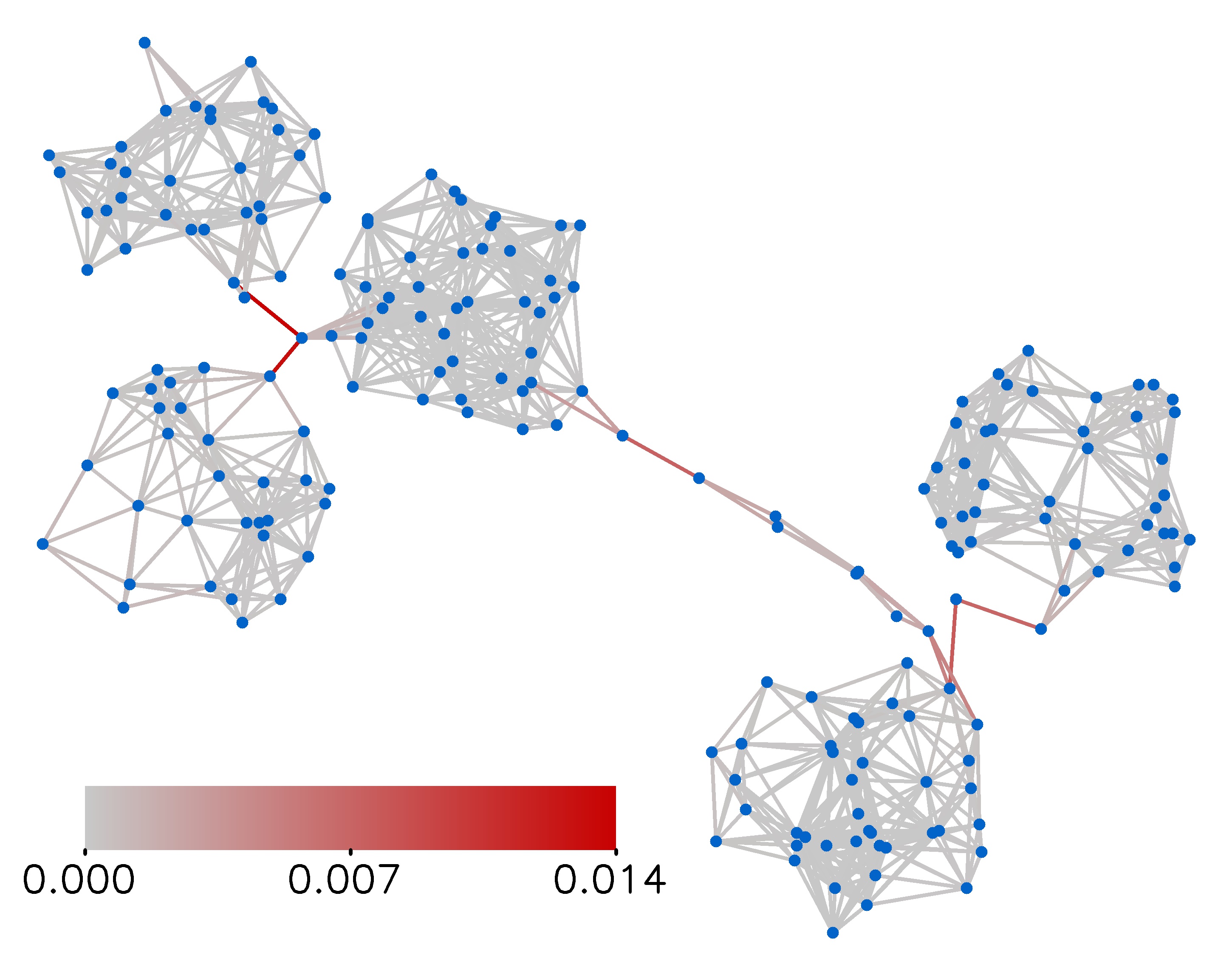}\label{fig:compDirect}}
\subfloat[\changed{Approximation for distributed computation for computing $\breve{\mathbf{r}}_\text{aprox-II}$.}]{\includegraphics[width = .33\linewidth]{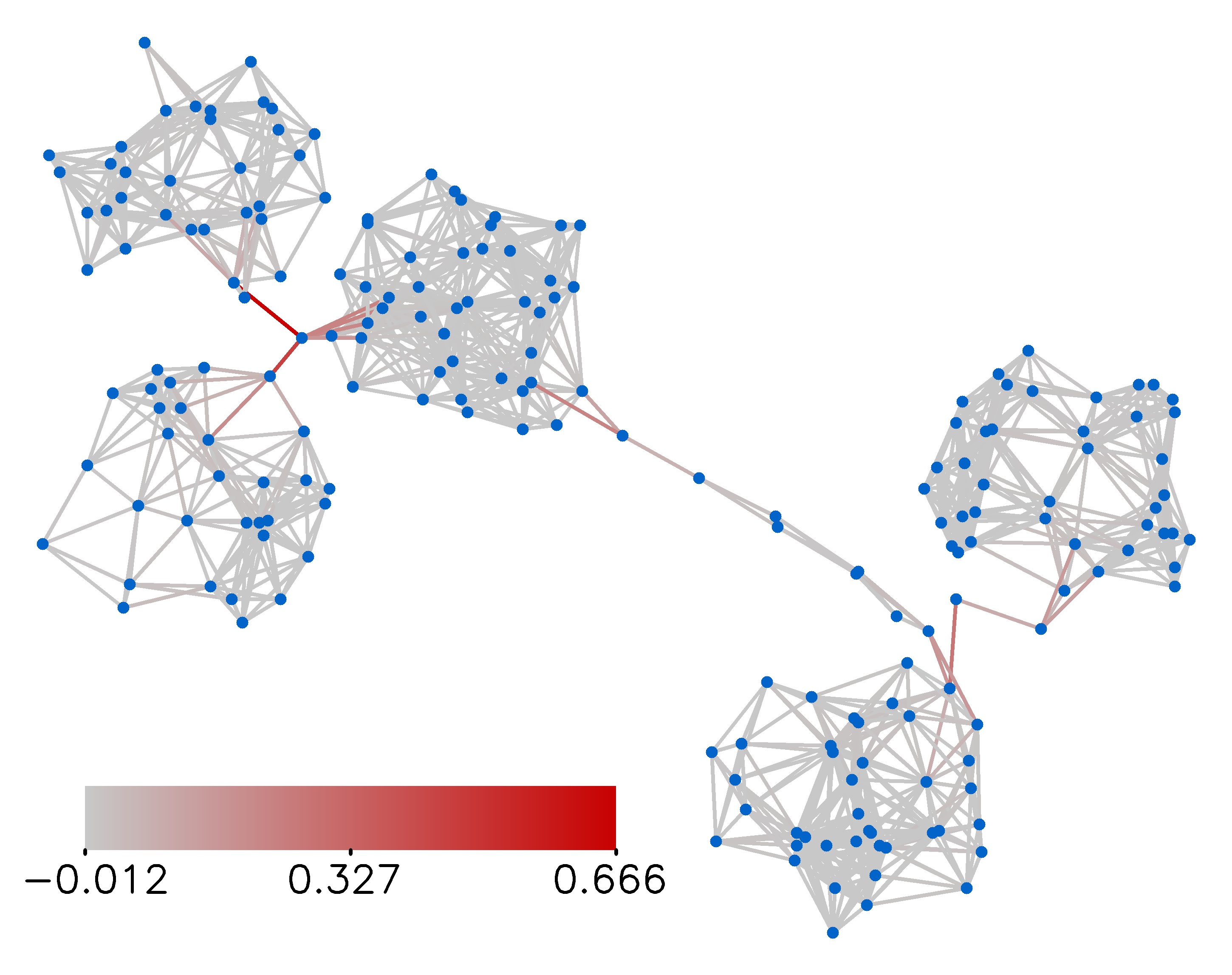}\label{fig:compDistributed}}
\caption{A qualitative comparison of accuracy in computing $\breve{\mathbf{r}}$. 
$\epsilon = 0.1$ is chosen in the approximation for mode-independent computation.
$p = \frac{N}{2}$ is chosen for the distributed approximation approach.
The resulting $L2$ norms$: \|\breve{\mathbf{r}}_\text{aprox-I} - \breve{\mathbf{r}}\| = 0.095$ and $\|\breve{\mathbf{r}}_\text{aprox-II} - \breve{\mathbf{r}}\| = 1.076$.
}
\label{fig:accuracyComparison}
\end{figure*}

Figure \ref{fig:accuracyComparison} shows the \changed{results obtained on a simple graph using} these three methods.
\changed{As can be seen from the value of $\|\breve{\mathbf{r}}_\text{aprox-I} - \breve{\mathbf{r}}\|$, the} result 
\changedF{using} the approximation methods for mode-independent computation is very close to the \changed{direct computation}.
The advantage of the approximation is that it can be \changed{computed in a} decentralized manner, which is especially useful when dealing with
\changedF{a large scale graph whose nodes have limited computational resources and communicate only with edge-connected neighbors.}

\subsection{Pandemic simulation}

\begin{figure}[h]
\centering
\includegraphics[width=0.95\columnwidth]{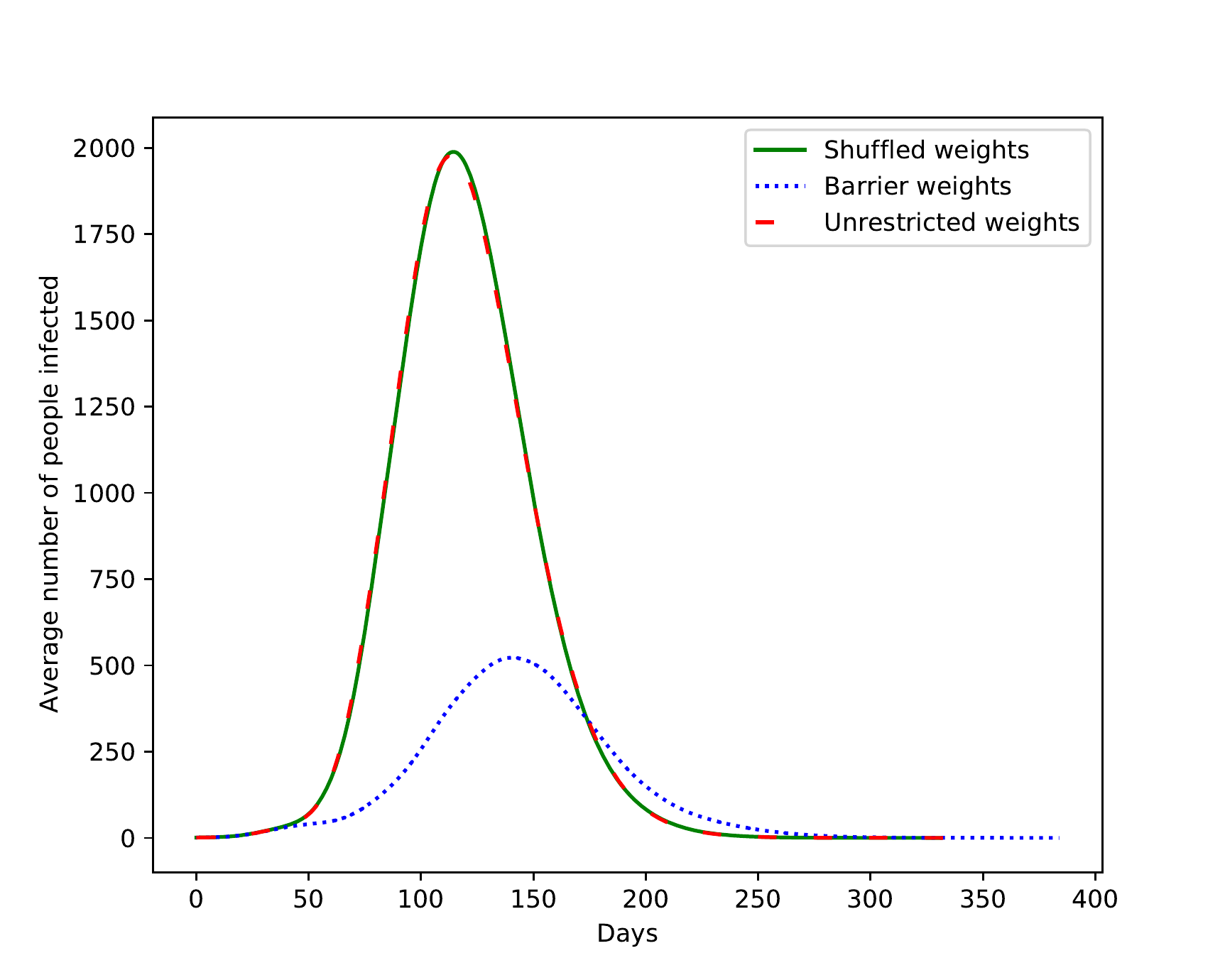} 
\vspace{-0.1in} \caption{Simulation results of the number of people infected using \changed{unrestricted weights}, barrier weights (inter-cluster restrictions) and shuffled weights. For each day, the number of people infected is averaged over $1000$ simulations (each simulation has \changed{randomized} initial patient-zero). 
}\vspace{-0.1in} 
\label{fig:pandemicCurve}
\end{figure}

\begin{figure*}
\centering
\subfloat[Day 1]{\includegraphics[width = .24\linewidth]{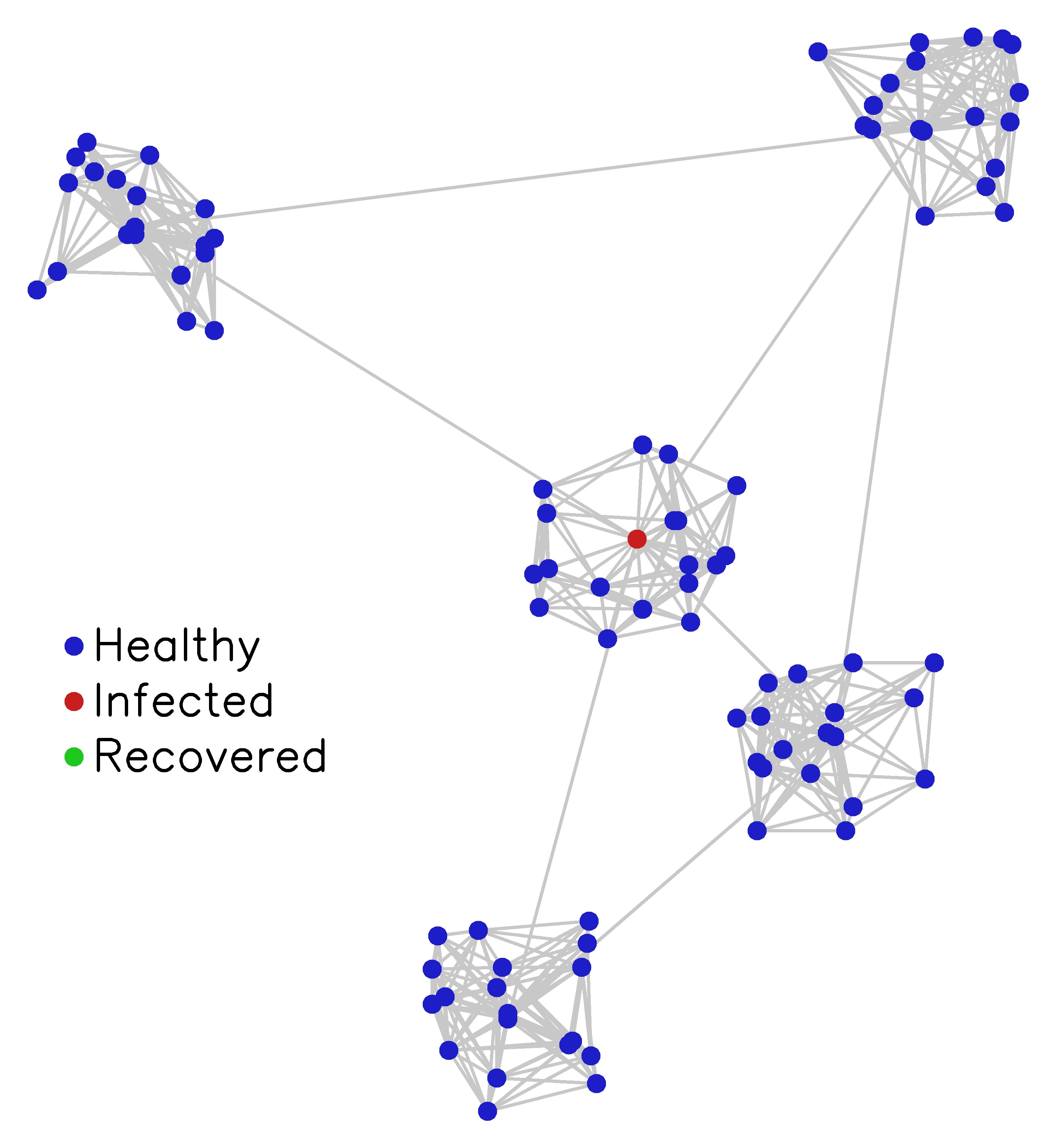}}
\subfloat[Day 30]{\includegraphics[width = .24\linewidth]{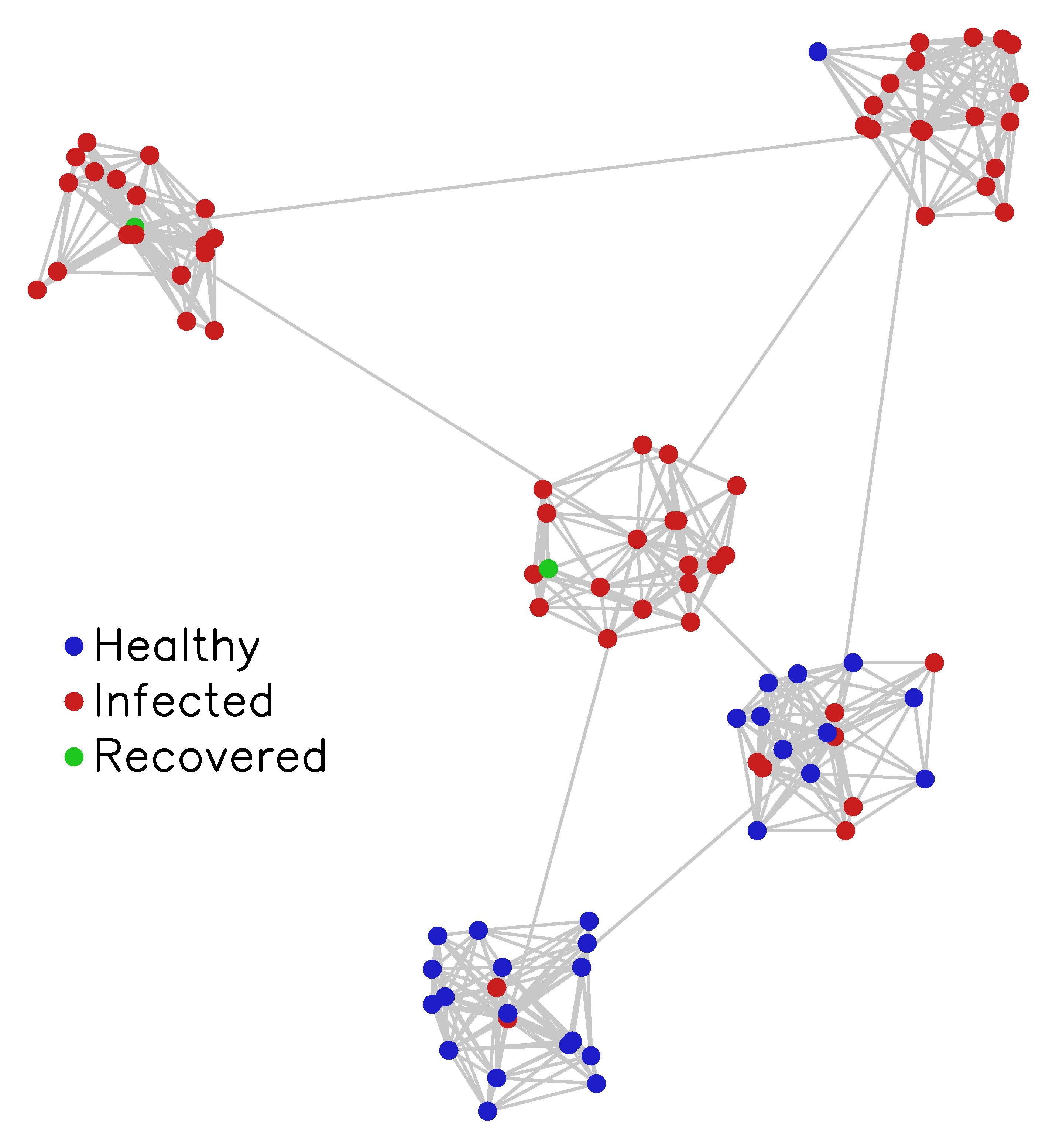}}
\subfloat[Day 60]{\includegraphics[width = .24\linewidth]{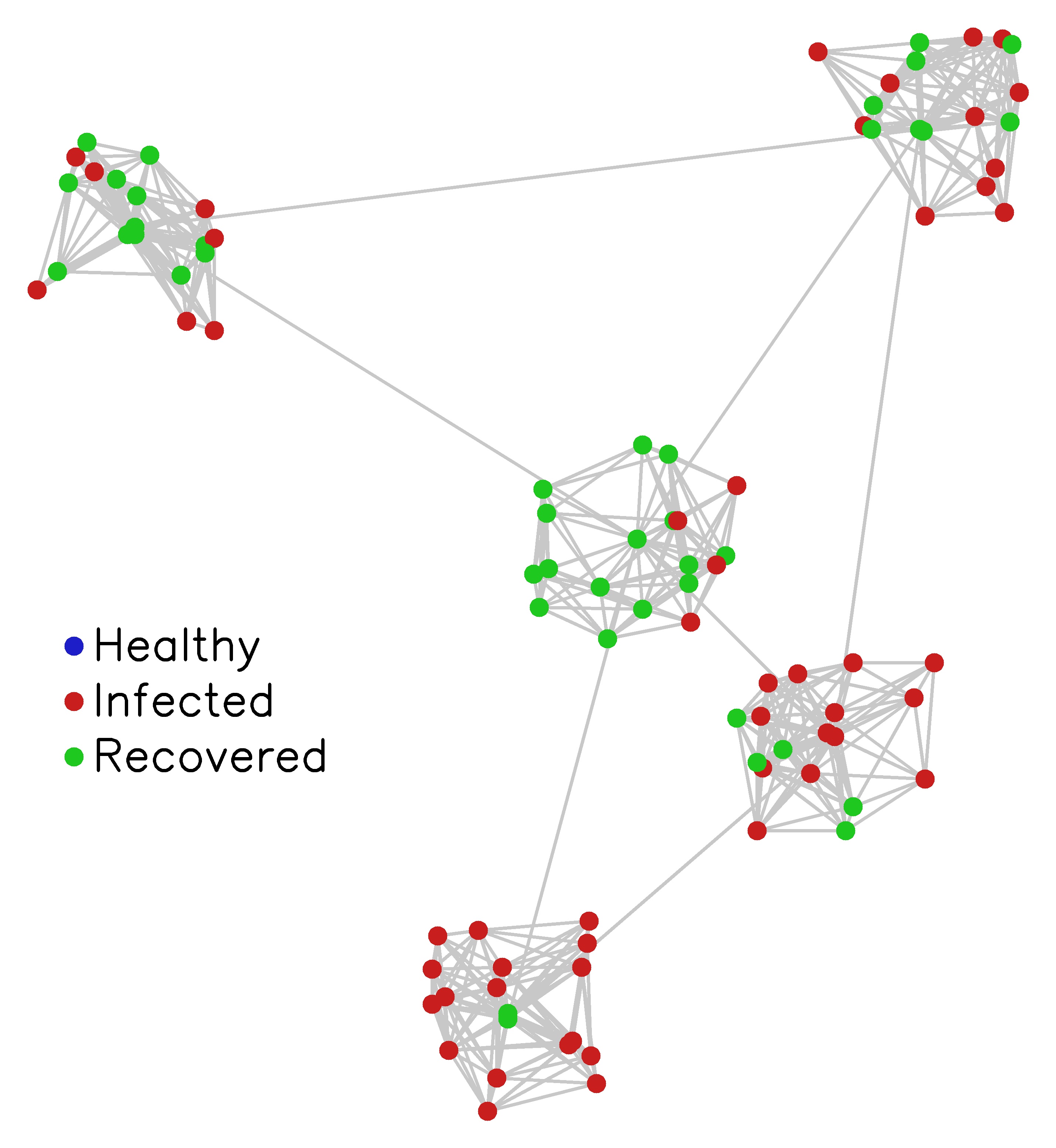}}
\subfloat[Day 90]{\includegraphics[width = .24\linewidth]{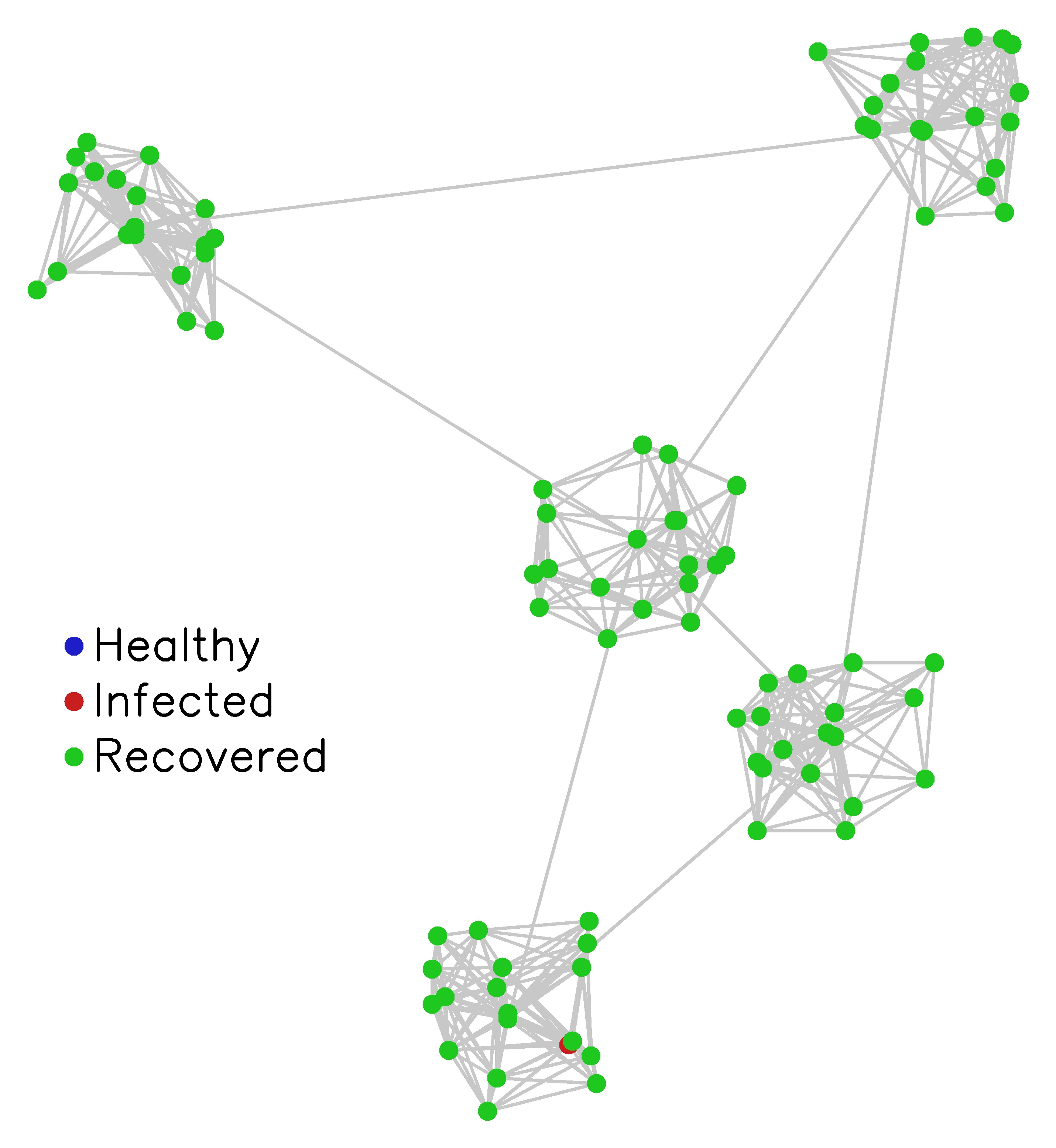}}\\
\subfloat[Day 1] {\includegraphics[width = .24\linewidth]{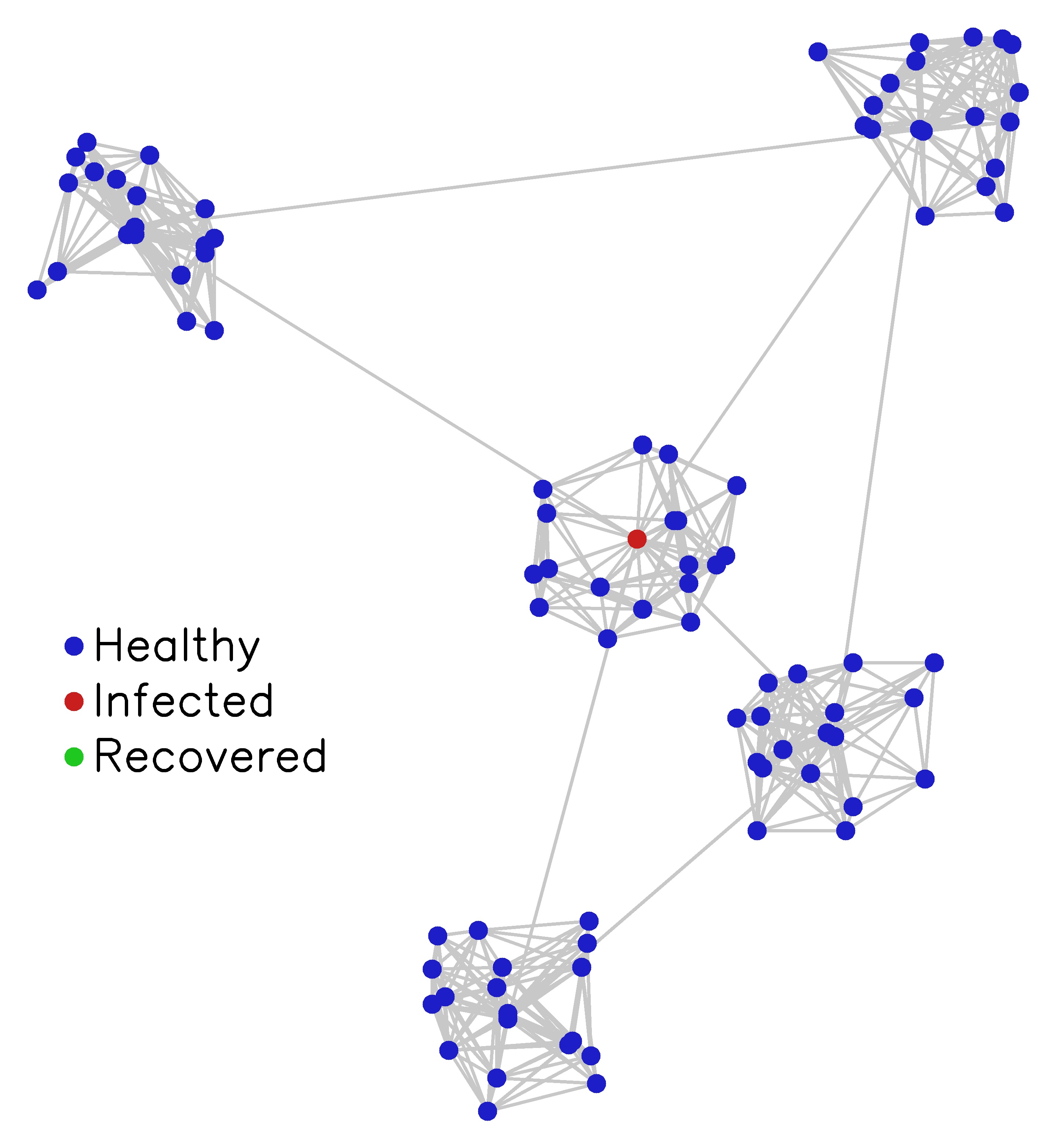}}
\subfloat[Day 45] {\includegraphics[width = .24\linewidth]{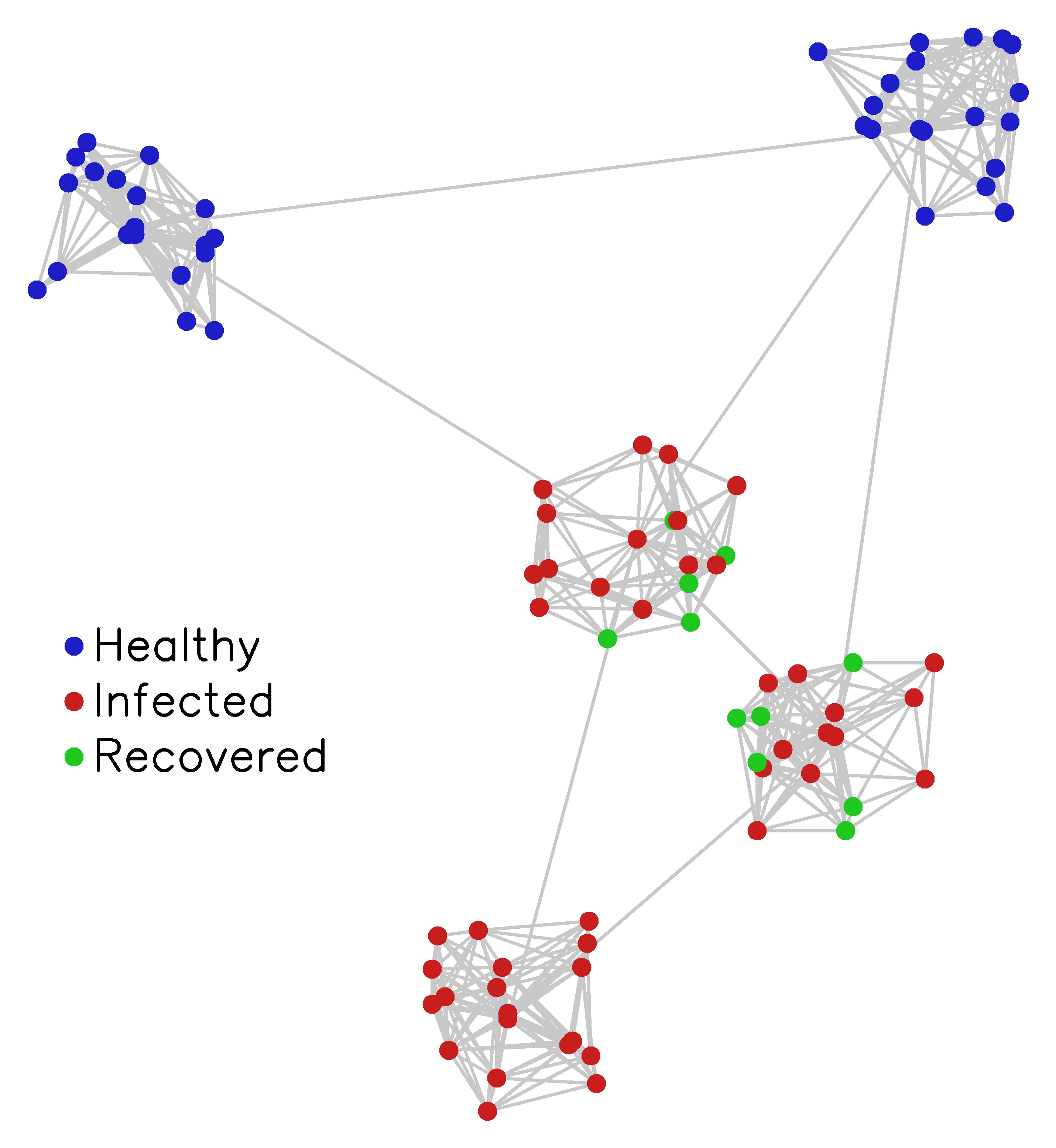}}
\subfloat[Day 90] {\includegraphics[width = .24\linewidth]{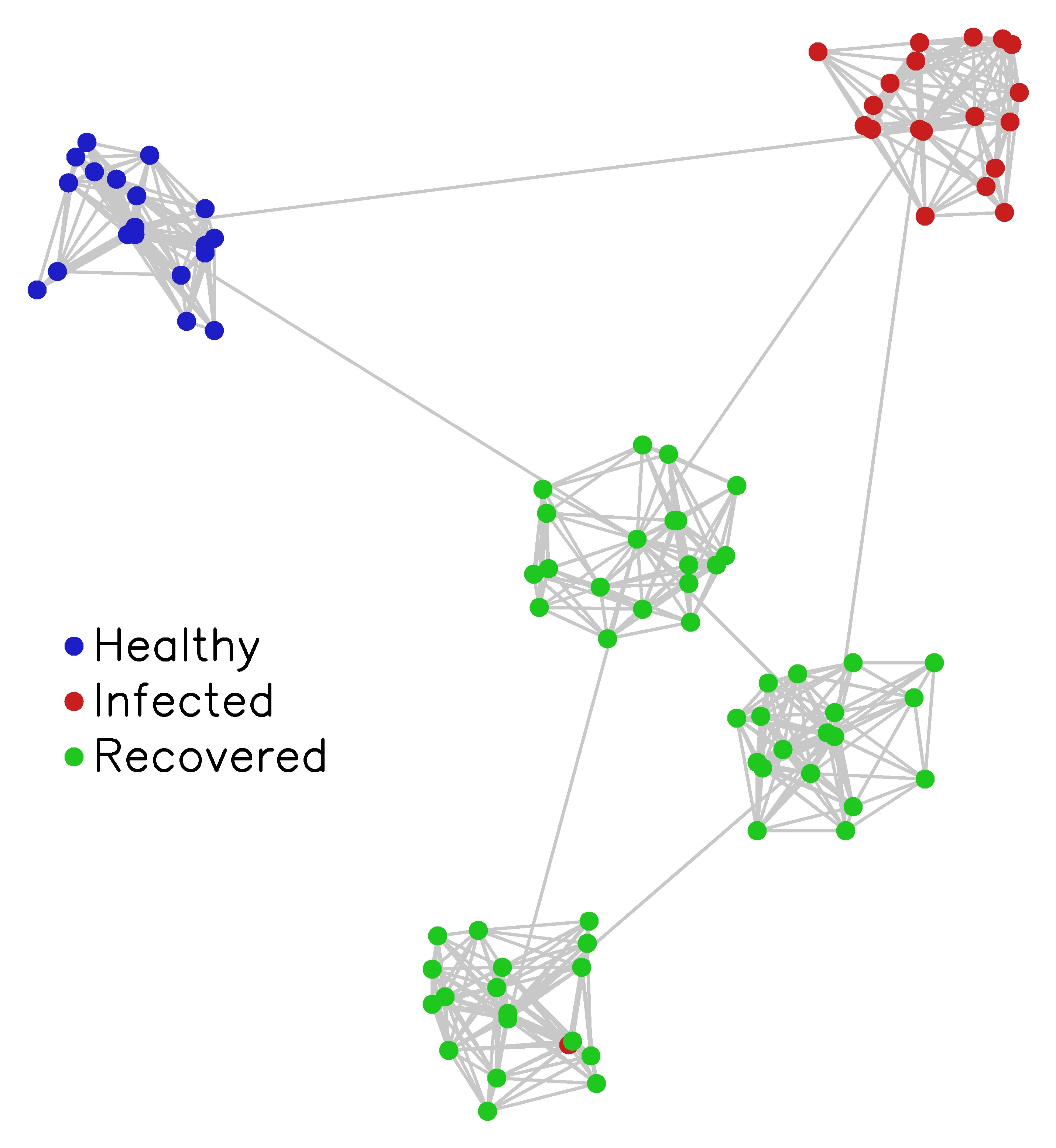}}
\subfloat[Day 135] {\includegraphics[width = .24\linewidth]{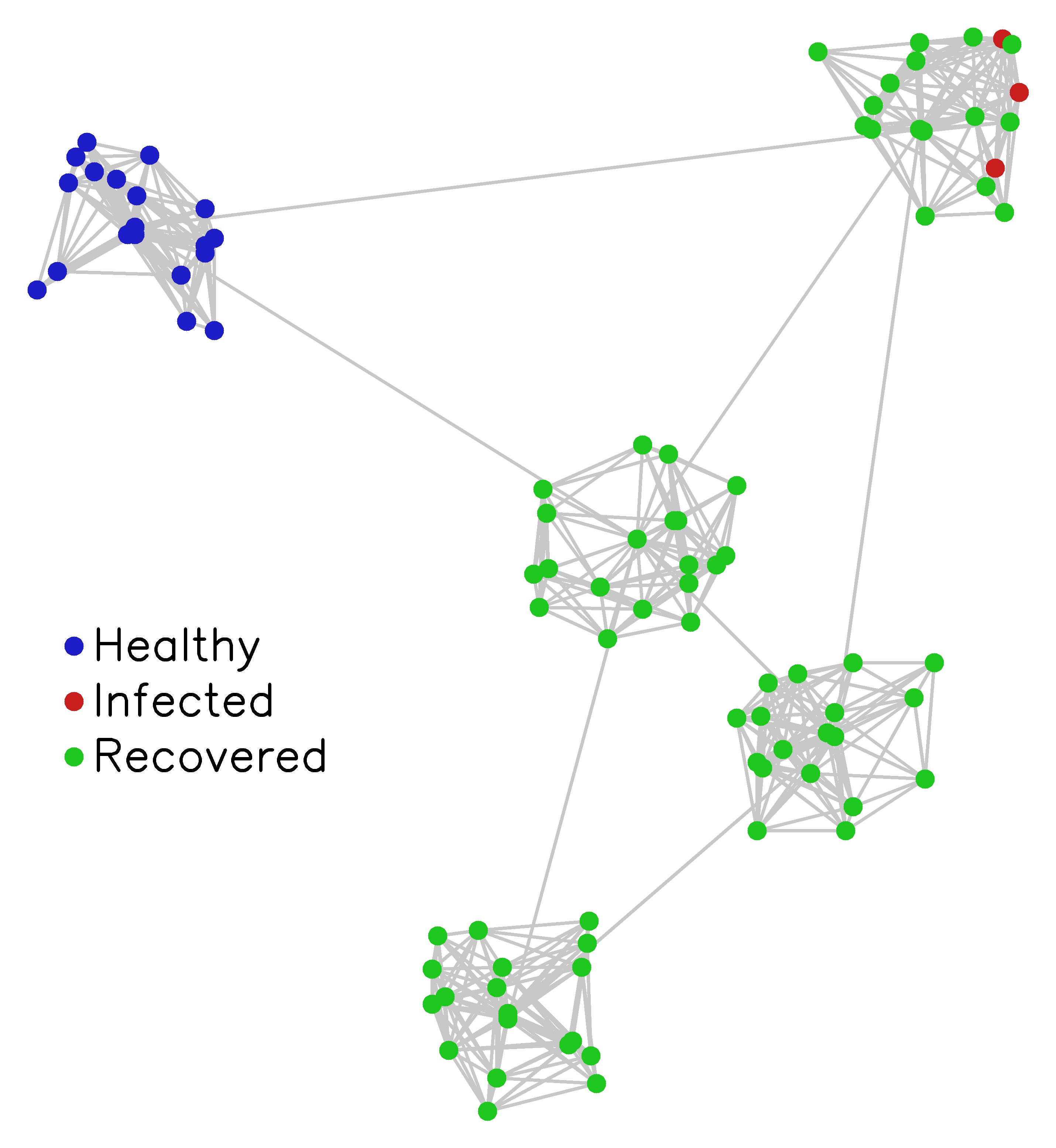}}
\caption{Snapshots of \changedF{the viral spread} simulation. (a), (b), (c) and (d) are snapshots from the simulation without social restrictions; (e), (f), (g) and (h) are snapshots from the simulation with inter-cluster social restrictions \changed{using barrier weights}.}
\label{fig:pandemicSnapshots}
\end{figure*}

As of writing this paper, our world is undergoing the \textit{COVID-19} pandemic.
Social restriction \changedB{has been shown to be} a proven way of slowing the spread of the virus \changedF{and providing time for treatment and eradication measures}~\changedB{\cite{covid-social-dist}}.
From our previous simulations, we can already see how effective the graph modal barriers are \changedF{for} slowing transmission across a network.
We thus apply the method developed in this paper to a pandemic model that can lower the probability of virus transmission along inter-cluster connections.
In this context a \emph{cluster} can be a representation of a closely-knit community of individuals (\changedF{e.g,} with a vertex representing a household and the flow representing interaction/travel of individuals).
Thus, given limited resources for imposing restrictions on travel or social interactions, identifying the most critical edges that \changedF{predominate} in transmitting the virus from one community to another (as is done by the proposed modal barrier method) is essential. 

For the simulation, 
the daily contagious probability of the $k$-th edge, $e_k = (v_i,v_j)$ (\emph{i.e.}, the probability that $v_j$ will get infected if $v_i$ is infected), is assumed to be
\begin{equation}
P_{k} = w'_{k} P_{a} 
\end{equation}
where $P_{a}$ is a constant which is set to be $0.03$.
\changedF{The lower} the weight on the $k$-th edge, \changedF{the lower} is the probability of transmission.

\changedF{Initially} we start with only one infected vertex (``patient zero'').
Once infected, a vertex will turn contagious in a range of $3$ to $5$ days, and will stay infected for a range of $15$ to $70$ days.
We compare the results with uniform unrestricted weights, $w^U_k (=1~\forall k)$, barriered weights, $w^B_k$, and shuffled weights, $w^S_k$.

%
Results for the ego-Facebook network are shown in Figure \ref{fig:pandemicCurve}. Figure~\ref{fig:pandemicSnapshots} shows some snapshots on the spread of the infection across the example \changedF{graph} shown in Figure \ref{fig:multiLinkEdgeValues}.

\section{Conclusion} \label{sec:conclusion}

\changed{In this paper we used the eigenvectors of \changedF{the} graph Laplacian to construct \changedF{modal barrier weights} on the edges of a \changedF{graph} to lower the transmission rate on \changedF{those critical} edges. Our method relies on using the eigenvectors to \changedF{identify} clusters in the \changedF{graph} and hence the edges that form weak inter-cluster connections form the basis of preventing inter-cluster transmission. We provided theoretical foundations \changedF{for} the proposed method and developed \changedF{an} approximation that \changedF{enabled low complexity} distributed computation of the barrier weights. We demonstrated that the proposed method is more effective in slowing transmission across a network compared to a method that uses the same amount of \changed{resources} in lowering the edge weights but does so in a randomized manner. We \changedF{discussed} the applicability of the proposed method in deciding how to impose social distancing and travel restrictions in preventing transmission of infection across communities.}

\changed{

\appendix

\subsection{Proof of Proposition~\ref{prop:q-conn-dist}} \label{appendix:prop-q-conn-dist-proof}

We use the same notations as in Definition~\ref{def:q-conn}, Definition~\ref{def:relout} and Proposition~\ref{prop:q-conn-dist}.



\vspace{0.5em}
\begin{lemma}\label{lemma:reloput-2-norm}
	\[ \sqrt{2} ~ \mathrm{relout}_G(G_j) ~\geq~ \|(L-\overline{L}) \overline{\mathbf{u}}_j\|_2 \]
\end{lemma}

\vspace{0.5em}
\noindent\emph{Lemma~\ref{lemma:reloput-2-norm}} \!\!\!\!\!\!\!\!\!\!\!\!
\begin{proof}
Suppose $D_k$ and $\overline{D}_k$ are the degrees of the vertex $v_k\in \mathcal{V}(G_j)$ in the graphs $G$ and $\overline{G}$ respectively.
Then $D_k$ and $\overline{D}_k$ are equal iff all the neighbors of $v_k$ are in $G_j$. Otherwise $D_{k} - \overline{D}_{k}$ is the net outgoing degree of the vertex $v_k$ from the sub-graph $G_j$. That is, if $v_k\in \mathcal{V}(G_j)$,
\begin{eqnarray} \label{eq:D-diff}
D_{k} - \overline{D}_{k} & = & 
\sum_{\{l\,|\, v_l\in \mathcal{V}(G_i), i\neq j\}} A_{kl}
\end{eqnarray}

Suppose $A$ and $\overline{A}$ are the adjacency matrices of $G$ and $\overline{G}$.
An edge $(v_k,v_l)$ exists (and have the same weight, \emph{i.e.}, $A_{kl} = \overline{A}_{kl}$) in both $G$ and $\overline{G}$ iff $v_k$ and $v_l$ belong to the same sub-graph. Otherwise $\overline{A}_{kl}=0$ (the edge is non-existent in $\overline{G}$). Thus,
\begin{equation} \label{eq:A-diff}
A_{kl} - \overline{A}_{kl} ~=~ \left\{ \begin{array}{l} A_{kl}, ~\text{if $v_k \in \mathcal{V}(G_j), v_l \in \mathcal{V}(G_i),i\neq j$} \\ 0, ~~\text{otherwise.} \end{array} \right.
\end{equation}


Next we consider the vector $\overline{\mathbf{u}}_j$ (for $j=0,1,\cdots,q-1$), which by definition is non-zero and uniform only on vertices in the sub-graph $G_j$.
Let $\overline{u}_{lj}$ be the $l$-th element of the unit vector $\overline{\mathbf{u}}_j$. Thus,
\begin{equation}
\overline{u}_{lj} = \left\{ \begin{array}{l} \frac{1}{\sqrt{|\mathcal{V}(G_j)|}}, ~~\text{if $v_l \in \mathcal{V}(G_j)$} \\ 0, ~~\text{otherwise} \end{array} \right.
\end{equation}

\noindent
Thus the $k$-th element of the vector $(L-\overline{L}) \overline{\mathbf{u}}_j$,
{\small \begin{align} 
 & [(L-\overline{L}) \overline{\mathbf{u}}_j]_k \nonumber \\
 = & \displaystyle \sum_{l} (D_{kl} - A_{kl} - \overline{D}_{kl} + \overline{A}_{kl}) \overline{u}_{lj} \nonumber \\
 = & \displaystyle (D_{k} - \overline{D}_{k}) \overline{u}_{kj} - \sum_{l} (A_{kl} - \overline{A}_{kl}) \overline{u}_{lj} \nonumber \\
 		& \qquad\qquad\qquad\qquad\text{\small(since $D$ and $\overline{D}$ are diagonal marices)} \nonumber \\
 = & \displaystyle \left( \frac{1}{\sqrt{|\mathcal{V}(G_j)|}} \left\{ \begin{array}{l}  (D_{k} - \overline{D}_{k}), ~\text{if $v_k \in \mathcal{V}(G_j)$} \nonumber \\ 0, ~\text{otherwise} \end{array} \right.\right) \nonumber \\
 		& \displaystyle \qquad\qquad - \left( \frac{1}{\sqrt{|\mathcal{V}(G_j)|}} \sum_{\{l\,|\,v_l \in \mathcal{V}(G_j)\}} (A_{kl} - \overline{A}_{kl}) \right) \nonumber \\
		& \qquad\qquad\qquad\qquad\qquad\qquad\text{\small(using the definition of $\overline{u}_{lj}$)} \nonumber \\
= & \displaystyle \frac{1}{\sqrt{|\mathcal{V}(G_j)|}} \left( \left\{ \begin{array}{l}  \displaystyle \sum_{\{l\,|\, v_l\in \mathcal{V}(G_i), i\neq j\}} \!\! A_{kl}, ~~~\text{if $v_k \in \mathcal{V}(G_j)$} \nonumber \\ 0, \qquad\qquad\text{otherwise} \end{array} \right.\right. \nonumber \\
		& \displaystyle \qquad\qquad - \left. \sum_{\{l\,|\,v_l \in \mathcal{V}(G_j)\}} \left\{ \begin{array}{l} A_{kl},~~\text{if $v_k\in \mathcal{V}(G_i), i\neq j$} \nonumber \\ 0, ~\text{otherwise.} \end{array} \right. \right) \nonumber \\
		& \qquad\qquad\qquad\qquad\qquad\qquad\qquad\text{\small(using \eqref{eq:D-diff} and \eqref{eq:A-diff})} \nonumber \\
= & \displaystyle \frac{1}{\sqrt{|\mathcal{V}(G_j)|}} 
		\left\{ \begin{array}{l}  
				\displaystyle \sum_{\{l\,|\, v_l\in \mathcal{V}(G_i), i\neq j\}} \!\!\!\! A_{kl}, ~~~\text{if $v_k \in \mathcal{V}(G_j)$} \\
				\displaystyle -\sum_{\{l\,|\,v_l \in \mathcal{V}(G_j)\}} A_{kl}, ~~~~~\text{if $v_k\in \mathcal{V}(G_i), i\neq j$} 
		 \end{array} \right.
%
\end{align}}

\noindent
Thus, 
{\small \begin{align} 
& \|(L-\overline{L}) \overline{\mathbf{u}}_j\|_2^2  \nonumber \\
= & \displaystyle \frac{1}{|\mathcal{V}(G_j)|}~~
		\left(
			\displaystyle \sum_{\{k\,|\,v_k \in \mathcal{V}(G_j)\}} \left( \sum_{\{l\,|\, v_l\in \mathcal{V}(G_i), i\neq j\}} \!\!\!\! A_{kl} \right)^2 \right. \nonumber \\ &
			\quad\qquad\qquad + \left. \displaystyle \sum_{\{k\,|\,v_k\in \mathcal{V}(G_i), i\neq j\}} \left( \sum_{\{l\,|\,v_l \in \mathcal{V}(G_j)\}} A_{kl} \right)^2
		\right) \nonumber \\
\leq & \displaystyle \frac{1}{|\mathcal{V}(G_j)|} \left( 
			\left( \sum_{\{k,l \,|\, v_k \in \mathcal{V}(G_j), v_l\in \mathcal{V}(G_i), i\neq j \}} \!\!\!\! A_{kl} \right)^2 \right. \nonumber \\ &
			\quad\qquad\qquad\quad\qquad\qquad + \left. \left(\displaystyle \sum_{\{k,l\,|\, v_k\in \mathcal{V}(G_i), v_l \in \mathcal{V}(G_j), i\neq j \}} \!\!\!\! A_{kl} \right)^2 
		\right) \nonumber \\
		& \qquad\qquad\qquad\text{\small(since for postive $\alpha_h$, $\sum_h \alpha_h^2 \leq (\sum_h \alpha_h)^2$)} \nonumber \\
= & \displaystyle \frac{2}{|\mathcal{V}(G_j)|} \left( \sum_{\{k,l \,|\, v_k \in \mathcal{V}(G_j), v_l\in \mathcal{V}(G_i), i\neq j \}} \!\!\!\! A_{kl} \right)^2 \nonumber \\
& \qquad\qquad\qquad\qquad\qquad\qquad\text{\small(since $A_{kl}=A_{lk}$.)} \nonumber \\
= & 2 ~\left( \mathrm{relout}_G(G_j) \right)^2
\end{align}}
\end{proof}


\vspace{0.5em}
\noindent\emph{Proposition~\ref{prop:q-conn-dist}} \!\!\!\!\!\!\!\!\!\!\!\!
\begin{proof}
	For any 
	$j \in \{0,1,\cdots,q-1\}$,
	{\small \begin{eqnarray}
	 \|(L-\overline{L}) \overline{\mathbf{u}}_j\|_2^{\changedC{2}} & = & \|L \overline{\mathbf{u}}_j \|_2^{\changedC{2}} ~~\text{\small (since $\overline{\mathbf{u}}_j$ is in the null-space of $\overline{L}$)} \nonumber \\
	 & = & \left\|L \sum_{l=0}^{n-1} (\mathbf{u}_l^\mathsf{T} \overline{\mathbf{u}}_j) \mathbf{u}_l \right\|_2^{\changedC{2}} \nonumber \\ & & \quad\text{\small ($\overline{\mathbf{u}}_j$ expressed in basis $\{\mathbf{u}_0,\mathbf{u}_1,\cdots,\mathbf{u}_n\}$)} \nonumber \\
	 & = & \left\|\sum_{l=0}^{n-1} (\mathbf{u}_l^\mathsf{T} \overline{\mathbf{u}}_j) \lambda_l \mathbf{u}_l \right\|_2^{\changedC{2}} ~~\text{\small (since $L \mathbf{u}_l = \lambda_l \mathbf{u}_l$)} \nonumber \\
	& = & \changedC{ \sum_{l=0}^{n-1} |\mathbf{u}_l^\mathsf{T} \overline{\mathbf{u}}_j|^2 \lambda_l^2 } ~~\changedC{\text{\small (since $\{\mathbf{u}_l\}_{l=0,1,\cdots,n-1}$}} \nonumber \\ & & \qquad\qquad\qquad\quad \changedC{\text{\small forms an orthogonal basis.)} } \nonumber \\
	& = & \changedC{ \sum_{k=0}^{q-1} \lambda_k^2 ~|\mathbf{u}_k^\mathsf{T} \overline{\mathbf{u}}_j|^2 ~+~ \sum_{k=q}^{n-1} \lambda_k^2 ~|\mathbf{u}_k^\mathsf{T} \overline{\mathbf{u}}_j|^2 } 
%
	 \label{eq:fundamental-equality}
	\end{eqnarray}}
\changedC{From the above, 
}
{\small \begin{eqnarray}
\|(L-\overline{L}) \overline{\mathbf{u}}_j\|_2^2 
			& \geq & \sum_{k=q}^{n-1} \lambda_k^2 ~| \overline{\mathbf{u}}_j^\mathsf{T} \mathbf{u}_k |^2 \nonumber \\
			& \geq & \left( \min_{k \in \{q,q+1,\cdots,n-1\}} \lambda_k^2 \right) \left( \sum_{k=q}^{n-1} | \overline{\mathbf{u}}_j^\mathsf{T} \mathbf{u}_k |^2 \right) \nonumber \\
	& = & \changedC{\lambda_q^2}  \sum_{k=q}^{n-1} | \overline{\mathbf{u}}_j^\mathsf{T} \mathbf{u}_k |^2  
	\end{eqnarray}}
Thus, using the above and Lemma~\ref{lemma:reloput-2-norm},
\begin{equation*}
2 \left( \mathrm{relout}_G(G_j) \right)^2 ~~\geq~~ \changedC{\lambda_q^2}  ~\sum_{k=q}^{n-1} | \overline{\mathbf{u}}_j^\mathsf{T} \mathbf{u}_k |^2  
\end{equation*}
Since the above is true for any $j \in \{0,1,\cdots,q-1\}$, 
{\small \begin{equation}\begin{array}{rl}
& 2 \displaystyle \sum_{j=0}^{q-1} \left( \mathrm{relout}_G(G_j) \right)^2
		 ~\geq~ ~\changedC{\lambda_q^2}  ~ \sum_{j=0}^{q-1} \sum_{k=q}^{n-1} | \overline{\mathbf{u}}_j^\mathsf{T} \mathbf{u}_k |^2   \\
\Rightarrow & \displaystyle \frac{1}{q(n-q)}  ~ \sum_{j=0}^{q-1} \sum_{k=q}^{n-1} | \overline{\mathbf{u}}_j^\mathsf{T} \mathbf{u}_k |^2 \nonumber \\ & \qquad\qquad\qquad
		 ~~\leq~~ \frac{\displaystyle 2}{\displaystyle \lambda_q^2 \changedC{(n-q)}} ~\displaystyle \frac{1}{q} \sum_{j=0}^{q-1} \left( \mathrm{relout}_G(G_j) \right)^2 \\
\Rightarrow & \mathrm{modaldist}_{G}(\widetilde{\mathcal{G}}) 
~~\leq~~ \changedC{\frac{
	\displaystyle 1}{\displaystyle \lambda_q} \sqrt{\displaystyle \frac{2}{n-q}} } ~\mathrm{avgrelout}_{G}(\widetilde{\mathcal{G}})
\end{array}\nonumber \end{equation}}
\end{proof}


\vspace{1em}
\subsection{Proof of Proposition~\ref{prop:lambda-q-1-upper-bound}}  \label{appendix:prop-lambda-q-1-upper-bound}

\changedC{
\noindent\emph{Proposition~\ref{prop:lambda-q-1-upper-bound}} \!\!\!\!\!\!\!\!\!\!\!\!
\begin{proof}
Since \changedC{\eqref{eq:fundamental-equality}} is true for any $j \in \{0,1,\cdots,q-1\}$,
\begin{align}
& \sum_{j=0}^{q-1} \|(L - \overline{L}) \overline{\mathbf{u}}_j\|_2^2 \nonumber \\
=~ & \sum_{k=0}^{q-1} \lambda_k^2 \sum_{j=0}^{q-1} |\mathbf{u}_k^\mathsf{T} \overline{\mathbf{u}}_j|^2 
	~+~ \sum_{k=q}^{n-1} \lambda_k^2 \sum_{j=0}^{q-1} |\mathbf{u}_k^\mathsf{T} \overline{\mathbf{u}}_j|^2  \nonumber \\
=~ & \sum_{k=0}^{q-1} \lambda_k^2 \left( 1 - \sum_{j=q}^{n-1} |\mathbf{u}_k^\mathsf{T} \overline{\mathbf{u}}_j|^2 \right) 
	~+~ \sum_{k=q}^{n-1} \lambda_k^2 \sum_{j=0}^{q-1} |\mathbf{u}_k^\mathsf{T} \overline{\mathbf{u}}_j|^2  \nonumber \\
=~ & \sum_{k=0}^{q-1} \lambda_k^2  
	~-~ \sum_{k=0}^{q-1} \lambda_k^2 \sum_{j=q}^{n-1} |\mathbf{u}_k^\mathsf{T} \overline{\mathbf{u}}_j|^2 
	~+~ \sum_{k=q}^{n-1} \lambda_k^2 \sum_{j=0}^{q-1} |\mathbf{u}_k^\mathsf{T} \overline{\mathbf{u}}_j|^2  \nonumber \\
\geq~ & \lambda_{q-1}^2
	~-~ \lambda_{q-1}^2 \sum_{k=0}^{q-1} \sum_{j=q}^{n-1} |\mathbf{u}_k^\mathsf{T} \overline{\mathbf{u}}_j|^2 
	~+~ \lambda_q^2 \sum_{k=q}^{n-1} \sum_{j=0}^{q-1} |\mathbf{u}_k^\mathsf{T} \overline{\mathbf{u}}_j|^2  \nonumber \\
	& \qquad\qquad \text{\small (since for $\alpha_h,\beta_h\geq 0$,~~ $\sum_h \alpha_k \geq \max_h \alpha_h$,} \nonumber \\ & \qquad\qquad\qquad\qquad\qquad \text{\small $\sum_h \alpha_h \beta_h \leq \max_h \alpha_h \sum_h \beta_h$, and} \nonumber \\ & \quad\qquad\qquad\qquad\qquad\qquad\qquad \text{\small $\sum_h \alpha_h \beta_h \geq \min_h \alpha_h \sum_h \beta_h$.)} \nonumber \\
=~ & \lambda_{q-1}^2 \left(1 - \sum_{k=0}^{q-1} \sum_{j=q}^{n-1} |\mathbf{u}_k^\mathsf{T} \overline{\mathbf{u}}_j|^2 \right)
	~+~ \lambda_q^2 \sum_{k=q}^{n-1} \sum_{j=0}^{q-1} |\mathbf{u}_k^\mathsf{T} \overline{\mathbf{u}}_j|^2  \nonumber \\
\end{align}
Using Lemma~\ref{lemma:reloput-2-norm},
\begin{align}
 & 2 \sum_{j=0}^{q-1} \left( \mathrm{relout}_G(G_j) \right)^2 ~~\geq~~ \nonumber \\
	 & \qquad \lambda_{q-1}^2 \left(1 - \sum_{k=0}^{q-1} \sum_{j=q}^{n-1} |\mathbf{u}_k^\mathsf{T} \overline{\mathbf{u}}_j|^2 \right)
	 ~+~ \lambda_q^2 \sum_{k=q}^{n-1} \sum_{j=0}^{q-1} |\mathbf{u}_k^\mathsf{T} \overline{\mathbf{u}}_j|^2  \nonumber \\
\Rightarrow~ & \frac{2}{\lambda_q^2} ~q \left( \mathrm{avgrelout}_G(\widetilde{G}) \right)^2 \nonumber \\
	& \qquad \geq~ \frac{\lambda_{q-1}^2}{\lambda_q^2} \left(1 - \sum_{k=0}^{q-1} \sum_{j=q}^{n-1} |\mathbf{u}_k^\mathsf{T} \overline{\mathbf{u}}_j|^2 \right)
		~+~  \sum_{k=q}^{n-1} \sum_{j=0}^{q-1} |\mathbf{u}_k^\mathsf{T} \overline{\mathbf{u}}_j|^2  \nonumber \\
	& \qquad \geq~ \frac{\lambda_{q-1}^2}{\lambda_q^2} \left(1 - \sum_{k=0}^{q-1} \sum_{j=q}^{n-1} |\mathbf{u}_k^\mathsf{T} \overline{\mathbf{u}}_j|^2 \right)   \nonumber \\
	& \qquad =~ \frac{\lambda_{q-1}^2}{\lambda_q^2} \left(1 - q(n-q) \left( \mathrm{modaldist}_G(\widetilde{G}) \right)^2 \right)   \nonumber \\
	& \qquad \geq~ \frac{\lambda_{q-1}^2}{\lambda_q^2} \left(1 - \frac{2 q}{\lambda_q^2} \left( \mathrm{avgrelout}_G(\widetilde{G}) \right)^2 \right)  \nonumber \\ & \qquad\qquad\qquad\qquad\qquad\qquad\qquad \text{\small (using Proposition~\ref{prop:q-conn-dist})} \nonumber \\
%
	& \qquad \geq~ \frac{\lambda_{q-1}^2}{\lambda_q^2} \left(1 - \alpha^2 \right) ~~\text{\small (since $\widetilde{\mathcal{G}}$ is an $\alpha$-realizable $q$-partition,} \nonumber \\ & \qquad\qquad\qquad\qquad \text{\small using Definition~\ref{def:alpha-realizable}, $\mathrm{avgrelout}_G(\widetilde{G}) \leq \frac{\lambda_q \alpha}{\sqrt{2q}}$)} \nonumber \\ 
\Rightarrow~ & \frac{\lambda_{q-1}}{\lambda_q} ~\leq~ \frac{1}{\sqrt{1-\alpha^2}} \frac{\sqrt{2q}}{\lambda_q} \mathrm{avgrelout}_G(\widetilde{G}) \label{eq:lambda-ratio-avgrelout} \\ 
 & \qquad~~ ~\leq~ \frac{\alpha}{\sqrt{1-\alpha^2}} ~~~~\text{\small (using Definition~\ref{def:alpha-realizable})} \label{eq:lambda-ratio-alpha}
\end{align}
\end{proof}
}

\vspace{1em}
\subsection{\changedB{Proof of Proposition~\ref{prop:r-robustness}}}  \label{appendix:prop-r-robustness}

\noindent\emph{Proposition~\ref{prop:r-robustness}} \!\!\!\!\!\!\!\!\!\!\!\!
\begin{proof}
\changedB{
	Let $\{\overline{\mathbf{u}}_0, \overline{\mathbf{u}}_1, \cdots, \overline{\mathbf{u}}_{n-1}\}$ be the eigenvectors of the $q$-partitioned graph $\overline{G}$ corresponding to the partition $\widetilde{\mathcal{G}}$ of $G$. Then,
	\begin{align} 
	& \displaystyle \breve{U} \breve{U}^\mathsf{T} \nonumber \\ 
	= & \displaystyle \sum_{l=0}^{q-1} \mathbf{u}_l \mathbf{u}_l^\mathsf{T}  \nonumber \\
	= & \displaystyle \sum_{l=0}^{q-1} \left( \sum_{j=0}^{n-1} (\overline{\mathbf{u}}_j^\mathsf{T} \mathbf{u}_l) \overline{\mathbf{u}}_j \right) \mathbf{u}_l^\mathsf{T}  \nonumber \\ & \qquad\qquad\qquad \text{\small ($\mathbf{u}_l$ written in basis $\{\overline{\mathbf{u}}_j\}_{j=0,1,\cdots,n-1}$)}  \nonumber \\
	= & \displaystyle \left( \sum_{j=0}^{n-1} \overline{\mathbf{u}}_j \overline{\mathbf{u}}_j^\mathsf{T} \right)
	\left( \sum_{l=0}^{q-1} \mathbf{u}_l \mathbf{u}_l^\mathsf{T} \right)  \nonumber \\ & \qquad \text{\small (since $(\overline{\mathbf{u}}_j^\mathsf{T} \mathbf{u}_l) \overline{\mathbf{u}}_j \mathbf{u}_l^\mathsf{T} =  \overline{\mathbf{u}}_j (\overline{\mathbf{u}}_j^\mathsf{T} \mathbf{u}_l) \mathbf{u}_l^\mathsf{T} = (\overline{\mathbf{u}}_j \overline{\mathbf{u}}_j^\mathsf{T}) (\mathbf{u}_l \mathbf{u}_l^\mathsf{T})$)}  \nonumber \\
	= & \displaystyle \left( \sum_{j=0}^{q-1} \overline{\mathbf{u}}_j \overline{\mathbf{u}}_j^\mathsf{T} + \sum_{j=q}^{n-1} \overline{\mathbf{u}}_j \overline{\mathbf{u}}_j^\mathsf{T}\right)
	\left( \sum_{l=0}^{q-1} \mathbf{u}_l \mathbf{u}_l^\mathsf{T} \right)  \nonumber \\
	= & \displaystyle \left( \sum_{j=0}^{q-1} \overline{\mathbf{u}}_j \overline{\mathbf{u}}_j^\mathsf{T} \right) \left( \sum_{l=0}^{n-1} \mathbf{u}_l \mathbf{u}_l^\mathsf{T} - \sum_{l=q}^{n-1} \mathbf{u}_l \mathbf{u}_l^\mathsf{T} \right)  \nonumber \\
	& \displaystyle \qquad\qquad +~ \left( \sum_{j=q}^{n-1} \overline{\mathbf{u}}_j \overline{\mathbf{u}}_j^\mathsf{T}\right) \left( \sum_{l=0}^{q-1} \mathbf{u}_l \mathbf{u}_l^\mathsf{T} \right)  \nonumber \\
	= & \displaystyle \sum_{j=0}^{q-1} \overline{\mathbf{u}}_j \overline{\mathbf{u}}_j^\mathsf{T} ~-~ 
	\sum_{j=0}^{q-1} \sum_{l=q}^{n-1} \overline{\mathbf{u}}_j \overline{\mathbf{u}}_j^\mathsf{T} \mathbf{u}_l \mathbf{u}_l^\mathsf{T} + \sum_{j=q}^{n-1} \sum_{l=0}^{q-1} \overline{\mathbf{u}}_j \overline{\mathbf{u}}_j^\mathsf{T}  \mathbf{u}_l \mathbf{u}_l^\mathsf{T}  \nonumber \\ & \qquad\qquad\qquad\qquad\qquad\qquad \text{\small (since $\sum_{l=0}^{n-1} \mathbf{u}_l \mathbf{u}_l^\mathsf{T} = I$)}  \nonumber \\
	\end{align}
	
	\noindent
	Thus,
	\begin{equation} \begin{array}{rl}
	\breve{U} \breve{U}^\mathsf{T} - \breve{\overline{U}} \breve{\overline{U}}^\mathsf{T} = & \displaystyle
	\sum_{j=0}^{q-1} \sum_{l=q}^{n-1} 
	\left(  \overline{\mathbf{u}}_l \overline{\mathbf{u}}_l^\mathsf{T}  \mathbf{u}_j \mathbf{u}_j^\mathsf{T} 
	- \overline{\mathbf{u}}_j \overline{\mathbf{u}}_j^\mathsf{T} \mathbf{u}_l \mathbf{u}_l^\mathsf{T} \right)
	\end{array} \label{eq:UUT-diff} \end{equation}
	where $\breve{\overline{U}} = [\overline{\mathbf{u}}_{0}~\overline{\mathbf{u}}_{1}~\cdots~\overline{\mathbf{u}}_{q-1}]$.
	
	We next consider the square of the Frobenius norm~\cite{Meyer:matrix:algebra} of the quantities on both sides of \eqref{eq:UUT-diff}.
	By the definition of Frobenius norm, $\displaystyle \|M\|_F^2 = \sum_{a,b} |M_{ab}^2| = \text{tr}(M^\mathsf{T} M)$. Thus,
	\begin{align}  
	& \| \breve{U} \breve{U}^\mathsf{T} - \breve{\overline{U}} \breve{\overline{U}}^\mathsf{T} \|_F^2 \nonumber \\
	= &~ \displaystyle 
	\text{tr}\Bigg( \sum_{j,j'=0}^{q-1} ~~\sum_{l,l'=q}^{n-1} 
		\left(  \overline{\mathbf{u}}_{l'} \overline{\mathbf{u}}_{l'}^\mathsf{T}  \mathbf{u}_{j'} \mathbf{u}_{j'}^\mathsf{T} \right)^\mathsf{T} \left(  \overline{\mathbf{u}}_l \overline{\mathbf{u}}_l^\mathsf{T}  \mathbf{u}_j \mathbf{u}_j^\mathsf{T} \right) \nonumber \\
		& \qquad\qquad\qquad\quad - \left( \overline{\mathbf{u}}_{l'} \overline{\mathbf{u}}_{l'}^\mathsf{T}  \mathbf{u}_{j'} \mathbf{u}_{j'}^\mathsf{T} \right)^\mathsf{T} \left( \overline{\mathbf{u}}_j \overline{\mathbf{u}}_j^\mathsf{T} \mathbf{u}_l \mathbf{u}_l^\mathsf{T} \right) \nonumber \\
		& \qquad\qquad\qquad\quad ~- \left( \overline{\mathbf{u}}_{j'} \overline{\mathbf{u}}_{j'}^\mathsf{T} \mathbf{u}_{l'} \mathbf{u}_{l'}^\mathsf{T} \right)^\mathsf{T} \left( \overline{\mathbf{u}}_l \overline{\mathbf{u}}_l^\mathsf{T}  \mathbf{u}_j \mathbf{u}_j^\mathsf{T} \right) \nonumber \\
		& \qquad\qquad\qquad\quad ~~+ \left( \overline{\mathbf{u}}_{j'} \overline{\mathbf{u}}_{j'}^\mathsf{T} \mathbf{u}_{l'} \mathbf{u}_{l'}^\mathsf{T} \right)^\mathsf{T} \left( \overline{\mathbf{u}}_j \overline{\mathbf{u}}_j^\mathsf{T} \mathbf{u}_l \mathbf{u}_l^\mathsf{T} \right) 
	\Bigg) \nonumber \\ 
	= & \displaystyle
	\sum_{j,j'=0}^{q-1} ~~\sum_{l,l'=q}^{n-1} \Big(
	 ( \mathbf{u}_{j'}^\mathsf{T} \overline{\mathbf{u}}_{l'} ) ( \overline{\mathbf{u}}_{l'}^\mathsf{T}
	\overline{\mathbf{u}}_l ) ( \overline{\mathbf{u}}_l^\mathsf{T}  \mathbf{u}_j) \,\mathsf{tr}( \mathbf{u}_{j'} \mathbf{u}_j^\mathsf{T} ) \nonumber \\
	& \qquad\qquad\qquad\quad -   ( \mathbf{u}_{j'}^\mathsf{T} \overline{\mathbf{u}}_{l'} ) ( \overline{\mathbf{u}}_{l'}^\mathsf{T} \overline{\mathbf{u}}_j ) ( \overline{\mathbf{u}}_j^\mathsf{T} \mathbf{u}_l) \,\mathsf{tr}( \mathbf{u}_{j'} \mathbf{u}_l^\mathsf{T} ) \nonumber \\
	& \qquad\qquad\qquad\quad ~-  ( \mathbf{u}_{l'}^\mathsf{T} \overline{\mathbf{u}}_{j'} ) ( \overline{\mathbf{u}}_{j'}^\mathsf{T} \overline{\mathbf{u}}_l ) ( \overline{\mathbf{u}}_l^\mathsf{T}  \mathbf{u}_j ) \,\mathsf{tr}( \mathbf{u}_{l'} \mathbf{u}_j^\mathsf{T} ) \nonumber \\
	& \qquad\qquad\qquad\quad ~~+  ( \mathbf{u}_{l'}^\mathsf{T} \overline{\mathbf{u}}_{j'} ) ( \overline{\mathbf{u}}_{j'}^\mathsf{T} \overline{\mathbf{u}}_j ) ( \overline{\mathbf{u}}_j^\mathsf{T} \mathbf{u}_l ) \,\mathsf{tr}( \mathbf{u}_{l'} \mathbf{u}_l^\mathsf{T} )
	\Big) \nonumber \\ 
	= & \displaystyle
	\sum_{j=0}^{q-1} \sum_{l=q}^{n-1} \left( |\overline{\mathbf{u}}_l^\mathsf{T} \mathbf{u}_j|^2 ~+~ |\overline{\mathbf{u}}_j^\mathsf{T} \mathbf{u}_l|^2 \right) \nonumber \\ 
	&  \qquad\qquad \text{\small (since 
	$\text{tr}(\mathbf{a} \mathbf{b}^\mathsf{T}) = \mathbf{a}^\mathsf{T} \mathbf{b}$ and
	$\mathbf{u}_c^\mathsf{T} \mathbf{u}_d = \overline{\mathbf{u}}_c^\mathsf{T} \overline{\mathbf{u}}_d = \delta_{cd}$)} \nonumber \\
	= &~ \displaystyle 2 q(n-q) \left( \mathrm{modaldist}_{G}(\widetilde{\mathcal{G}}) \right)^2 \text{\small ~~(Definition~\ref{def:q-conn})} \nonumber \\
	\leq &~ \displaystyle 2 q(n-q)  \frac{2}{\lambda_q^2 \changedC{(n-q)}} \left( \mathrm{avgrelout}_G (\widetilde{\mathcal{G}}) \right)^2 ~~\text{\small (Proposition~\ref{prop:q-conn-dist})} \nonumber \\
	\leq &~ \displaystyle 2 q(n-q)  \frac{2}{\lambda_q^2 \changedC{(n-q)}} \left( \frac{\alpha^2 \lambda_q^2}{2 q} \right) ~~\text{\small (Definition~\ref{def:alpha-realizable})} \nonumber \\
	= &~ 2 \alpha^2
	\label{eq:frob-norm-1}
	\end{align}
	
	\noindent
	Using a similar analysis we get
	\begin{equation} \label{eq:frob-norm-2}
	\| \breve{U}' \breve{U}'{}^\mathsf{T} - \breve{\overline{U}} \breve{\overline{U}}^\mathsf{T} \|_F^2 ~\leq~ 2 \alpha^2
	\end{equation}
	Using triangle inequality of Frobenius norm,
	{\small \begin{eqnarray}
	\| \breve{U} \breve{U}^\mathsf{T} - \breve{U}' \breve{U}'{}^\mathsf{T} \|_F 
	& \leq & \| \breve{U} \breve{U}^\mathsf{T} - \breve{\overline{U}} \breve{\overline{U}}^\mathsf{T} \|_F + \| \breve{U}' \breve{U}'{}^\mathsf{T} - \breve{\overline{U}} \breve{\overline{U}}^\mathsf{T} \|_F \nonumber \\
	& \leq & 2 \sqrt{2} \alpha ~~\text{\small (using \eqref{eq:frob-norm-1} and \eqref{eq:frob-norm-2})}
	\end{eqnarray}}
	
	\noindent
	For notational simplicity, define $\Delta = \breve{U} \breve{U}^\mathsf{T} - \breve{U}' \breve{U}'{}^\mathsf{T}$.
	Suppose the $k$-th edge connects the $a$-th and $b$-th vertices (\emph{i.e.} $e_k = (v_a,v_b) \in \mathcal{E}(G)$). 
	Then from \eqref{eq:tilde-r} it follows
	\begin{equation}\begin{array}{rl}
	& \displaystyle \breve{r}_k - \breve{r}'_k \\
	= & \displaystyle \sum_{i,j=0}^n B_{ik} \Delta_{ij} B_{jk} \\
	= &  \displaystyle \sum_{i,j\in \{a,b\}} B_{ik} \Delta_{ij} B_{jk} ~~\text{\small (since $B_{ik} = 0$ for $i\notin \{a,b\}$)} \\
	= & \displaystyle \Delta_{aa} - \Delta_{ab} - \Delta_{ba} + \Delta_{bb} \\ & \qquad\qquad \text{\small (since $B_{ak}, B_{bk} \in \{1,-1\}$ and $B_{ak} \neq B_{bk}$)}
	\end{array} \label{eq:dr-Delta} \end{equation}
	Thus,
	\begin{equation}\begin{array}{l}
	\frac{1}{4} |\breve{r}_k - \breve{r}'_k |^2 ~\leq~  |\Delta_{aa}|^2 + |\Delta_{ab}|^2 + |\Delta_{ba}|^2 + |\Delta_{bb}|^2 
	\\ \quad\qquad \text{\small (since $(a+b+c+d)^2 \leq 2((a+b)^2 + (c+d)^2) $} \\ \qquad\qquad\qquad ~~\text{\small $\leq 4(a^2+b^2+c^2+d^2)$ for any $a,b,c,d\in\mathbb{R}$)}
	\end{array}\end{equation}
	We next proceed to sum both sides of the above equation over all edges in $G$. 
	\begin{equation}\begin{array}{l}
	\displaystyle \frac{1}{4} \sum_{k=1}^m 
	|\breve{r}_k - \breve{r}'_k |^2 \nonumber \\ 
	\qquad\qquad \leq \displaystyle \sum_{\{a,b \,|\, \atop (v_a,v_b)\in\mathcal{E}(G)\}} \!\!\! \left( |\Delta_{aa}|^2 + |\Delta_{ab}|^2 + |\Delta_{ba}|^2 + |\Delta_{bb}|^2  \right)
	\end{array}\end{equation}
	In the above equation 
	we observe that, in the sum on the right hand side, terms of the form $|\Delta_{ii}|^2$ appears as many times as there are edges connected to the $i$-th vertex (\emph{i.e.}, the degree of the vertex). However, terms of the form $|\Delta_{ij}|^2$ (for $i\neq j$) appear twice for every edge $(v_i,v_j)$ that exists. Thus
	we have
	\begin{align*}
	  \displaystyle \frac{1}{4} \sum_{k=1}^m |\breve{r}_k - \breve{r}'_k |^2 \!\! & ~\leq~ \displaystyle \sum_{i=1}^n \text{deg}(v_i) |\Delta_{ii}|^2 ~~+~~ 2 \!\!\!\!\!\!\sum_{\{i,j \,|\, \atop (v_i,v_j) \in \mathcal{E}(G) \}} \!\!\!\! |\Delta_{ij}|^2  \\
	& ~\leq~ \displaystyle \left( \max_{v_i \in \mathcal{V}(G)} \text{deg}(v_i) \right) \sum_{i=1}^n |\Delta_{ii}|^2  \\
	&        \displaystyle \qquad\qquad\qquad +~~~ 2 \!\!\!\!\!\! \sum_{\{i,j \,| \atop v_i,v_j \in \mathcal{V}(G), i\neq j \}} \!\!\!\!\!\! |\Delta_{ij}|^2  \\
	& ~\leq~ \displaystyle \left( \max_{v_i \in \mathcal{V}(G)} \text{deg}(v_i) \right) \sum_{i=1}^n \sum_{j=1}^n |\Delta_{ij}|^2  \\
	& \qquad \text{\small (for a connected graph with at least $2$} \\ & \qquad\qquad \text{\small vertices, $\max_{v_i \in \mathcal{V}(G)} \text{deg}(v_i) \geq 2$.)} \\
	& ~=~ \displaystyle d_{\text{max}} ~\|\Delta\|_F^2  \\ 
	& ~=~ \displaystyle 2 d_{\text{max}} \alpha^2 ~~\text{\small (from \eqref{eq:frob-norm-1})} \\
	\Rightarrow ~ \| \breve{\mathbf{r}} - \breve{\mathbf{r}}' \|_2 & ~\leq~~ 2 \alpha \sqrt{2 d_\text{max}}
	\end{align*} 
	
}
%
\end{proof}


}


\changedC{
\vspace{1em}
\subsection{Proof of Proposition~\ref{prop:approximation}}  \label{appendix:prop-approximation}

\noindent\emph{Proposition~\ref{prop:approximation}}} \!\!\!\!\!\!\!\!\!\!\!\!
\begin{proof}
\changed{\changedC{From the definition of $\breve{U}$, since it constitutes of the first $q$ columns of $U$, we have} $\breve{U} \breve{U}^\mathsf{T} = U \text{diag}([\mathbf{1}_q, \mathbf{0}_{n-q}]) U^\mathsf{T}$.
	\changedC{Furthermore,} $\epsilon(\epsilon I + L)^{-1} = U \text{diag}_{i=0}^{n-1}(\frac{\epsilon}{\epsilon + \lambda_{i}}) U^\mathsf{T}$ \changedC{(since $U$ is the orthogonal matrix that diagonalizes $L$)}.
	Thus, \begin{equation}\begin{array}{@{}l}
	\breve{U} \breve{U}^\mathsf{T} - \epsilon(\epsilon I + L)^{-1} \\
	= U ~\text{diag}\left(
	1 \!-\! \frac{\epsilon}{\epsilon + \lambda_{0}}, 1 \!-\! \frac{\epsilon}{\epsilon + \lambda_{1}}, \cdots, 1 \!-\! \frac{\epsilon}{\epsilon + \lambda_{q-1}}, \right.
	\\ \qquad\qquad\qquad
	\left. -\frac{\epsilon}{\epsilon + \lambda_{q}}, -\frac{\epsilon}{\epsilon + \lambda_{q+1}}, \cdots, -\frac{\epsilon}{\epsilon + \lambda_{n-1}}
	\right) U^\mathsf{T}
	\end{array} \end{equation}
	Furthermore, since the matrices are symmetric 
	their operator $2$-norms are the 
	maximum of the absolute values of their eigenvalues. Thus,
	%
	{\small \begin{align}
		& \frac{ \|\breve{U} \breve{U}^\mathsf{T} - \epsilon (\epsilon I + L)^{-1}\|_\text{2} }{ \|\breve{U} \breve{U}^\mathsf{T} \|_\text{2} } \nonumber \\
		=~ & \frac{ \max\left( 
			\max_{i=0}^{q-1}\left( \frac{\lambda_i}{\epsilon + \lambda_{i}} \right),
			\max_{i=q}^{n-1}\left( \frac{\epsilon}{\epsilon + \lambda_{i}} \right)
			\right) }{1} \nonumber \\
		=~ & \max\left( \frac{\lambda_{q-1}}{\epsilon + \lambda_{q-1}} , \frac{\epsilon}{\epsilon + \lambda_{q}} \right) \nonumber \\
		%
		<~ & \max\left( \frac{\lambda_{q-1}}{\epsilon} , \frac{\epsilon}{\lambda_{q}} \right)
		\nonumber \\
		\leq~ & \max \left( 
		\frac{1}{\epsilon} \frac{\displaystyle \sqrt{\changedC{2q}} ~\textrm{avgrelout}_{G}(\widetilde{\mathcal{G}})}{\displaystyle \sqrt{1-\alpha^2}} ~,
		\epsilon \frac{\alpha}{\sqrt{\changedC{2q}} ~\textrm{avgrelout}_{G}(\widetilde{\mathcal{G}}) }
		\right) \nonumber\\
		& \qquad\text{(using 
			\changedC{equation \eqref{eq:lambda-ratio-avgrelout} from the derivation of} Proposition~\ref{prop:lambda-q-1-upper-bound}
			} \nonumber\\ & \qquad\qquad\qquad\quad\text{and Definition~\ref{def:alpha-realizable} \changedC{for} $\alpha$-realizable $q$-partition)} \nonumber \\
		\leq~ & \max \left( 
		\frac{1}{\epsilon} \frac{\displaystyle \sqrt{\changedC{2q}} ~\widehat{a}/\beta}{\displaystyle \sqrt{1-\alpha^2}} ~,
		\epsilon \frac{\alpha}{\sqrt{\changedC{2q}} ~\widehat{a}\beta }
		\right) \nonumber\\
		\leq~ & \changedC{ \max \left( 
			\frac{1}{\epsilon} \frac{\displaystyle \sqrt{\changedC{2q}} ~\widehat{a}/\beta}{\displaystyle \sqrt{1-\widehat{\alpha}^2}} ~,
			\epsilon \frac{\widehat{\alpha}}{\sqrt{\changedC{2q}} ~\widehat{a}\beta }
			\right)  \qquad\text{(since $\alpha \leq \widehat{\alpha}$)} } \nonumber\\
		=~ & \max \left( 
		\changedC{\frac{\left(\widehat{\alpha}\sqrt{1-\widehat{\alpha}^2}\right)^{1/2}}{\sqrt{2q} ~\widehat{a}}} \frac{\displaystyle \sqrt{\changedC{2q}} ~\widehat{a}/\beta }{\displaystyle \sqrt{(1-\alpha^2)}  } ~,~ 
		\changedC{\frac{\sqrt{2q} ~\widehat{a}}{\left(\widehat{\alpha}\sqrt{1-\widehat{\alpha}^2}\right)^{1/2}}} \frac{\changedC{\widehat{\alpha}}}{\sqrt{\changedC{2q}} ~\widehat{a}\beta }
		\right) \nonumber \\ & \qquad\qquad\qquad\qquad\qquad\qquad 
		\qquad\text{(setting $\changedC{\epsilon = \displaystyle \frac{\sqrt{2q} ~\widehat{a}}{\left(\widehat{\alpha}\sqrt{1-\widehat{\alpha}^2}\right)^{1/2}} }$)} \nonumber \\
		=~ &  \frac{1}{\beta} ~\max \left( 
		\changedC{ \frac{\sqrt{\widehat{\alpha}}}{\left(1-\widehat{\alpha}^2\right)^{1/4}} } ,~
		\changedC{ \frac{\sqrt{\widehat{\alpha}}}{\left(1-\widehat{\alpha}^2\right)^{1/4}} }
		\right)  \nonumber\\
		%
		%
		=~ & \frac{\changedC{1}}{\beta} \changedC{ \frac{\sqrt{\widehat{\alpha}}}{\left(1-\widehat{\alpha}^2\right)^{1/4}} }
		\end{align}}
}%
\end{proof}

\bibliographystyle{IEEEtranS}
\bibliography{references}

\end{document}